\documentclass[10pt]{article}
\usepackage[utf8]{inputenc}
\usepackage[a4paper]{geometry}
\usepackage{amsmath,amsbsy,amsgen,amscd,amsthm,amsfonts,amssymb,mathrsfs}

\usepackage{natbib}
\setcitestyle{authoryear,round,comma}

\usepackage{multirow}
\usepackage{makecell}

\usepackage[dvipsnames]{xcolor}
\usepackage{xr-hyper}
\usepackage[colorlinks=true, linkcolor=blue, breaklinks=true]{hyperref}

\usepackage{algorithm}
\usepackage{algpseudocode}

\usepackage{graphicx}  
\usepackage{caption}
\usepackage{subcaption}
\graphicspath{ {./figures_paper/} }

\usepackage{pgfplots}

\usepackage{titlesec}
\titleformat*{\section}{\large\bfseries}
\titleformat*{\subsection}{\normalsize\bfseries}
\titleformat*{\subsubsection}{\small\bfseries}

\theoremstyle{plain}
\newtheorem{theorem}{Theorem}[section]
\newtheorem{lemma}[theorem]{Lemma}
\newtheorem{proposition}[theorem]{Proposition}

\theoremstyle{definition}
\newtheorem{definition}[theorem]{Definition}
\newtheorem{remark}[theorem]{Remark}

\newcommand{\KL}[2]{D_{\operatorname{KL}}({#1}\,||\,#2)}
\renewcommand{\phi}{\varphi}
\newcommand{\Ymin}{Y_{\operatorname{min}}}
\newcommand{\Ymax}{Y_{\operatorname{max}}}
\newcommand{\supp}{\operatorname{supp}}
\newcommand{\dist}{\operatorname{dist}}
\newcommand{\diam}{\operatorname{diam}}
\newcommand{\unif}{\operatorname{Unif}}

\newcommand{\R}{\mathbb{R}}
\newcommand{\M}{\mathcal{M}}
\renewcommand{\P}{\mathcal{P}}

\newcommand{\projmixing}{G_\sigma}
\newcommand{\estmixing}{\widehat{G}_{n,\sigma}}
\newcommand{\selected}{\widehat{G}_{n,2\sigmahat}}
\renewcommand{\square}{\mathcal{SQ}}
\renewcommand{\circle}{\mathcal{C}}
\newcommand{\arc}{\mathcal{A}}

\renewcommand{\sp}[2]{#1_{\sigma,#2}}
\newcommand{\sphat}[2]{\widehat{#1}_{\sigma,#2}}

\newcommand{\revise}[1]{\textcolor{blue}{#1}}
\renewcommand{\revise}[1]{\textcolor{black}{#1}}

\newcommand{\sigmahat}{\widehat{\sigma}}
\newcommand{\npmleclusterk}{\widehat{E}_k}
\newcommand{\voronoik}{E_k}
\newcommand{\suppk}{S_{\sigma,k}}

\title{Model-free Estimation of Latent Structure via \\Multiscale Nonparametric Maximum Likelihood}
\author{Bryon Aragam and Ruiyi Yang} 
\date{\emph{University of Chicago and Princeton University}}

\begin{document}

\maketitle
{\let\thefootnote\relax\footnote{Contact: \texttt{bryon@chicagobooth.edu}, \texttt{ry8311@princeton.edu}}}
\pagenumbering{arabic}

\begin{abstract}
Multivariate distributions often carry latent structures that are difficult to identify and estimate, and which better reflect the data generating mechanism than extrinsic structures exhibited simply by the raw data.
In this paper, we propose a model-free approach for estimating such latent structures whenever they are present, without assuming they exist \emph{a priori}. Given an arbitrary density $p_0$, we construct a multiscale representation of the density and propose data-driven methods for selecting representative models that capture meaningful discrete structure.
Our approach uses a nonparametric maximum likelihood estimator to estimate the latent structure at different scales and we further characterize their asymptotic limits. 
By carrying out such a multiscale analysis, we obtain coarse-to-fine structures inherent in the original distribution, which are integrated via a model selection procedure to yield an interpretable discrete representation of it. 
As an application, we design a clustering algorithm based on the proposed procedure and demonstrate its effectiveness in capturing a wide range of latent structures. %
\end{abstract}

\section{Introduction}

Multivariate distributions are known to exhibit exotic behaviour in high-dimensions, which 
makes density estimation difficult in higher and higher dimensions. 
At the same time, multivariate distributions often possess intrinsic structure that is useful for downstream tasks, and that does not require estimating the entire density at fine scales.
There are many known examples of this phenomenon: discrete structure in the form of clustering or mixtures \citep{wolfe1970pattern,mclachlan2004mixture,stahl2012model}, low-dimensional structure in the form of a manifold \citep{belkin2006manifold,lin2008riemannian} or sparsity \citep{hastie2015statistical}, and dependence structure in the form of a graph \citep{banerjee2015bayesian,lee2015learning,drton2017structure}.
\revise{Traditionally, these latent structures are assumed to exist and then learned from data. However, verifying whether these structural assumptions hold can be difficult in practice, and even reasonable assumptions apply only approximately in complex, high-dimensional datasets. These challenges are further amplified in nonparametric and infinite-dimensional models.}

\revise{Motivated by these observations, in this paper we revisit this problem from a different perspective:} To what extent can discrete latent structure be identified and estimated in general, high-dimensional densities, without necessarily assuming an analytic form of such a structure \emph{a priori}? 
In other words, is there a natural model-free notion of latent structure that is statistically meaningful and estimable for arbitrary densities?

\begin{figure}[t]
    \centering
    \minipage{0.35\textwidth}
    \includegraphics[width=\linewidth]{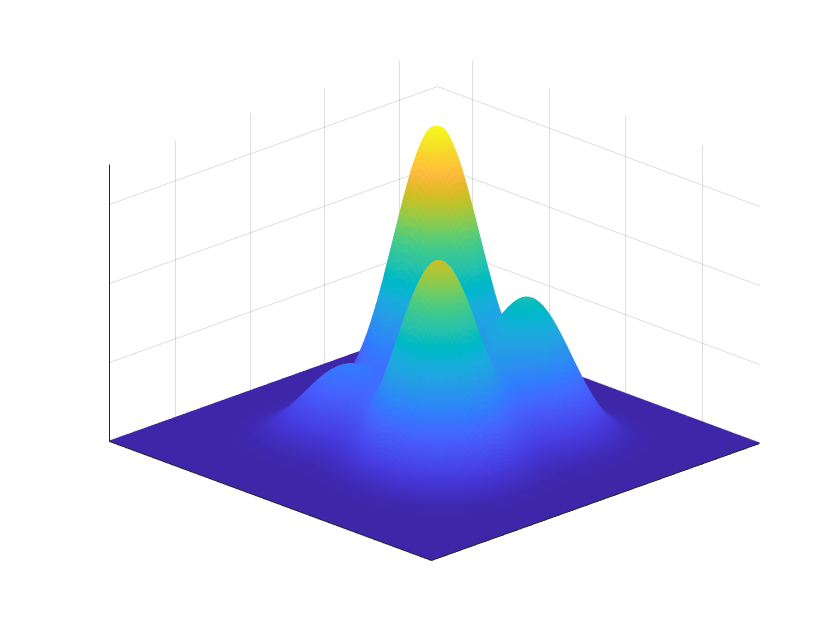}
    \subcaption{Density plot}
    \endminipage
    \minipage{0.25\textwidth}
    \includegraphics[width=\linewidth]{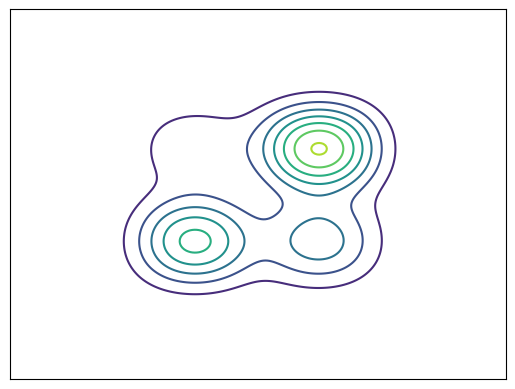}
    \subcaption{Contour plot}
    \endminipage
    \minipage{0.25\textwidth}
    \includegraphics[width=\linewidth]{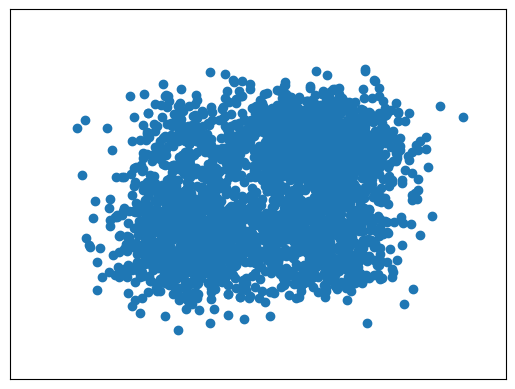}
    \subcaption{Samples}
    \endminipage\hfill

    \caption{A 2D example of a density with hidden structure. The density (a) has four modes, capturing four overlapping clusters that are not well-separated. This is more easily seen by the contour plot (b). This structure is not obviously present, however, in the raw data (c).}
    \label{fig:toy ex}
\end{figure}

We propose such a recipe for deciphering latent structures inherent in an \emph{arbitrary} $d$-dimensional density $p_0$ that leads to practical algorithms.
We do not impose any restrictions on the form of $p_0$; in particular, we do not assume any mixture, clustering, or latent class structure to begin with.
Our approach starts by treating and estimating the latent structure of $p_0$ as a nonparametric parameter, taking the form of a latent probability measure. 
This measure, computed via nonparametric maximum likelihood estimation \citep[NPMLE,][]{kiefer1956consistency}, can be interpreted as a compressed representation of $p_0$.
The crucial feature of this latent measure is that its support carries a rich geometry with a clearer structure than the original density $p_0$, as visualized in Figures~\ref{fig:toy ex}-\ref{fig:toy-npmle} and described in more detail in Section~\ref{sec:overview}.
A refined analysis on the support then reveals the hidden structures of $p_0$.  
The key aspect of our approach lies in how we construct this latent measure and its support, and its dependence on a hyperparameter $\sigma$ that controls the scale of the representation. 
In particular, by computing this across a range of $\sigma$'s, we reveal coarse-to-fine, ``multi-scale'' structures of $p_0$, which can be integrated to yield a most representative model $\widehat{G}$ that is useful for various downstream tasks. As an example application, we illustrate its use for clustering tasks.

\subsection{Overview}
\label{sec:overview}

To gain some intuition on what this latent structure looks like, we first present the main ideas at a high-level, deferring technical details to Sections~\ref{sec:main}-\ref{sec:est latent}.
Consider the two-dimensional example shown in Figure \ref{fig:toy ex}, which displays the density, its contour plot, and a scatter plot of samples generated from this density. 
\revise{This simulated nonparametric density has four modes with different heights that are not well separated, and whose clusters cannot be easily discerned from the raw samples. }
The latent structure, captured by the aforementioned latent measures and computed via the NPMLE, is pictured in
Figure \ref{fig:toy-npmle}. This shows the (evidently discrete) supports of different latent measures for increasing values of $\sigma$, which are seen to exhibit very different structures. 
In particular, they transition from a densely supported measure ($\sigma=0.4$) to a sparse one ($\sigma\ge1.2$). 
One can interpret such dynamics as the merging and separation of a collection of points representing the latent structure at different scales.

\begin{figure}[t]
    \centering
    \minipage{0.2\textwidth}
    \includegraphics[width=\linewidth]{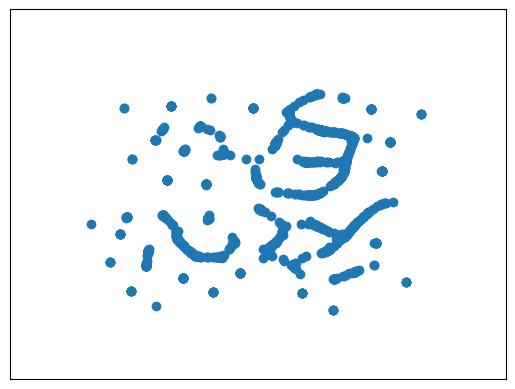}
    \subcaption{$\sigma=0.4$}\label{fig:toy-npmle-1}
    \endminipage
    \minipage{0.2\textwidth}
    \includegraphics[width=\linewidth]{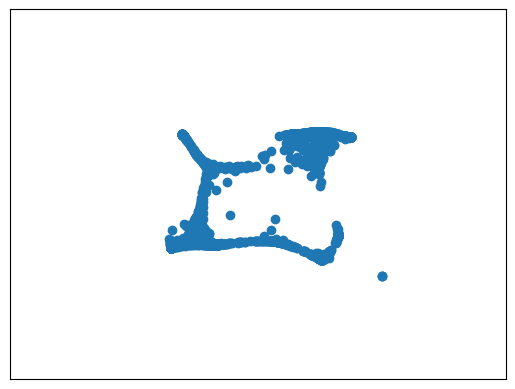}
    \subcaption{$\sigma=0.8$}\label{fig:toy-npmle-2}
    \endminipage
    \minipage{0.2\textwidth}
    \includegraphics[width=\linewidth]{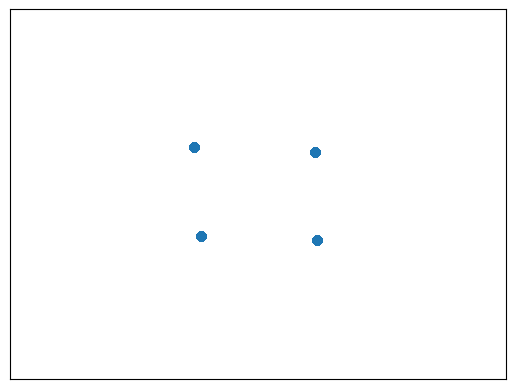}
    \subcaption{$\sigma=1.2$ ($\widehat{G}$)}\label{fig:toy-npmle-3}
    \endminipage
    \minipage{0.2\textwidth}
    \includegraphics[width=\linewidth]{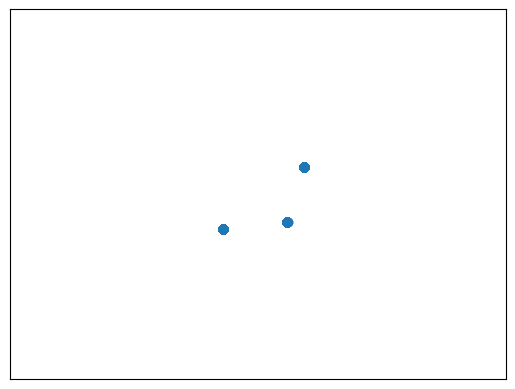}
    \subcaption{$\sigma=1.6$}\label{fig:toy-npmle-4}
    \endminipage
    \minipage{0.2\textwidth}
    \includegraphics[width=\linewidth]{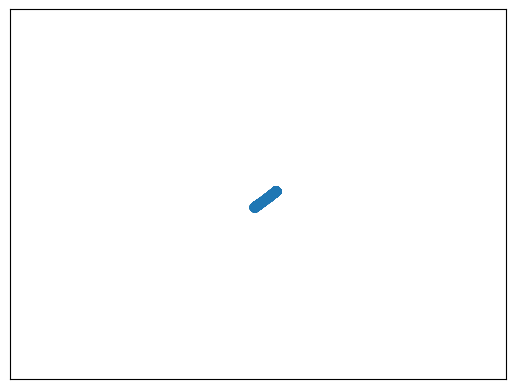}
    \subcaption{$\sigma=2.0$}\label{fig:toy-npmle-5}
    \endminipage\hfill

    \caption{The multiscale representation of the density in Figure~\ref{fig:toy ex} as the scale $\sigma$ increases from left to right, revealing discrete structure in $p_0$ across different scales. Each subfigure visualizes the discrete atoms of the latent probability measure which depends only on $p_0$ and $\sigma$. As an example, the third model $\widehat{G}$ with $\sigma=1.2$ is chosen by our method for clustering. The four atoms are actually clusters of atoms suggesting there are four distinct clusters.}
    \label{fig:toy-npmle}
\end{figure}

For each value of the hyperparameter $\sigma$, the latent measures in Figure~\ref{fig:toy-npmle} capture the intrinsic discrete structure of $p_0$ at different scales that we call a \emph{multiscale representation} of $p_0$. This is reminiscent of the bias-variance tradeoff in density estimation (e.g. in choosing the bandwidth or number of neighbours), however, our target is quite different: Instead of recovering the density $p_0$, we aim to recover \emph{discrete} structure such as intrinsic clusters or mixture components---\emph{without} assuming their existence \emph{a priori}. If density estimation were our goal, then choosing $\sigma=0.4$ or even smaller in Figure~\ref{fig:toy-npmle} would be preferable, but clearly this measure does not capture the four clusters in $p_0$ in any meaningful way.

After constructing the multiscale representation of $p_0$, the next step is to select a choice of $\sigma$ that best captures this discrete structure.
We will propose a model selection procedure that returns the third measure with $\sigma=1.2$ (Figure~\ref{fig:toy-npmle-3})---indicated by $\widehat{G}$---as the selected model in this example, which contains precisely four clusters of atoms representing the four high density regions of the original density $p_0$. 
Moreover, this latent structure will be captured much more concretely and rigourously than simply inspecting the support of $\widehat{G}$ by eye. 
The measure $\widehat{G}$ in fact defines several objects that are useful (Figure~\ref{fig:toy-features}):
\begin{enumerate}
    \item $\widehat{G}$ defines a \emph{dendrogram}, which provides a more nuanced and qualitative view of the latent geometry. This can be used to identify similar clusters and even subclusters, and also serves to help infer an approximate number of discrete states or latent classes. See Figure~\ref{fig:toy dg}.
    \item We can also use $\widehat{G}$ to define estimates of \emph{class conditional densities} over each cluster, as in Figure~\ref{fig:toy cd}.
    \item We can use $\widehat{G}$ to define a \emph{partition} of the input space as in Figure~\ref{fig:toy partition}, and thus also a \emph{clustering} of the original input data points as well as a \emph{classifier} for future unseen observations.
\end{enumerate}
Thus, $\widehat{G}$ not only captures a \emph{quantitative} notion of structure, to be introduced in the sequel, but also useful \emph{qualitative} notions of structure that can be used for model assessment, validation, and prediction.

\begin{figure}[t]
    \centering
    \minipage{0.25\textwidth}
    \includegraphics[width=\linewidth]{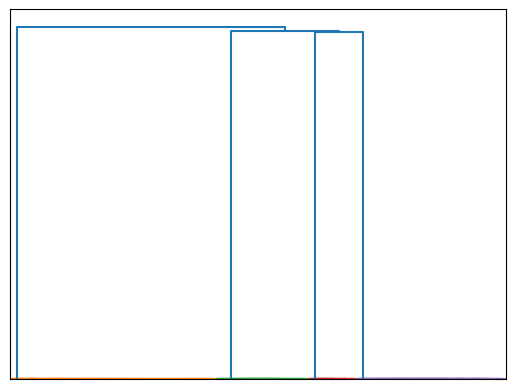}
    \subcaption{Dendrogram of $\supp(\widehat{G})$}\label{fig:toy dg}
    \endminipage   
    \minipage{0.35\textwidth}
    \includegraphics[width=\linewidth]{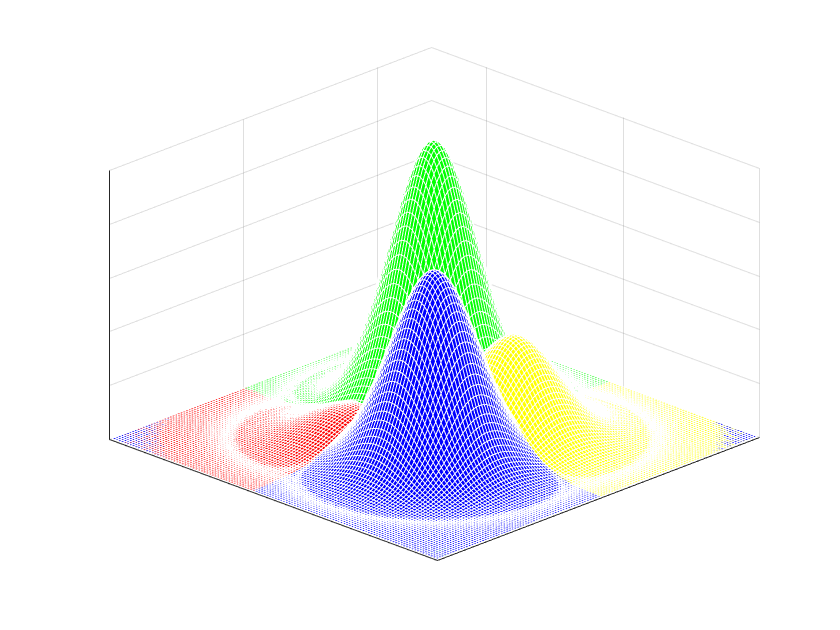}
    \subcaption{Class conditional densities}\label{fig:toy cd}
    \endminipage
    \minipage{0.25\textwidth}
    \includegraphics[width=\linewidth]{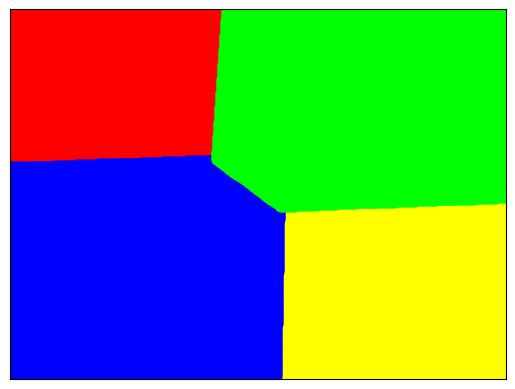}
    \subcaption{Partition of the space}\label{fig:toy partition}
    \endminipage\hfill

    \caption{Qualitative features extracted from $\widehat{G}$ for the example in Figure~\ref{fig:toy ex}.}
    \label{fig:toy-features}
\end{figure}

Crucially, this construction makes \emph{no} assumptions on the form of $p_0$, and applies to any density.
It appears as though the measure $\widehat{G}$ is capturing the intrinsic, discrete structure of $p_0$ without exploiting specific parametric or structural assumptions on it. 
While intuitively appealing, making this precise requires some effort. Moreover, there are several practical challenges to tackle including computation, high-dimensionality, and the selection of $\sigma$.
Below we shall briefly introduce the main tools of our approach, followed by a summary of our results and contributions.

\subsection{Main Ingredients}
A central tool that we shall employ in this paper for computing the probability measure $\widehat{G}$ is the \emph{nonparametric maximum likelihood estimator} \citep[NPMLE,][]{kiefer1956consistency}.
Formally speaking, this is an estimator defined as 
\begin{align}\label{eq:npmle formal}
    \widehat{p} = \underset{p\in \M }{\operatorname{arg\,max}}\,\, \sum_{i=1}^n \log p(Y_i),
\end{align}
where $\{Y_i\}_{i=1}^n$ are samples from $p_0$ and $\M$ is a suitable space of probability densities to be defined shortly. 
Instead of a finite-dimensional parameter space as in standard maximum likelihood estimation, \eqref{eq:npmle formal} searches over a potentially infinite-dimensional space.
Due to its nonparametric nature, some care is needed in specifying the search space $\M$.
For example, if we na\"ively choose the space of all continuous densities $\M_0$, we can see that the objective function value in \eqref{eq:npmle formal} would tend to infinity on a sequence of densities approximating the empirical measure $n^{-1}\sum_{i=1}^n \delta_{Y_i}$. 
The issue lies in the overly large parameter space $\M_0$ which allows $p$ to be arbitrarily narrow and spiked. 

Therefore, to make the problem well-posed, we shall restrict our attention to the set of all densities with a ``minimum length-scale" $\sigma$.
This is achieved by considering densities that can be written as a convolution of the form $p=\phi_\sigma\ast G$ for some kernel $\phi_\sigma$ with bandwidth $\sigma$, with $G$ ranging through the set of all probability measures. 
Denoting the collection of such densities as $\M_\sigma$, we can see that the maximization \eqref{eq:npmle formal} over $\M_\sigma$ would equivalently give a probability measure $\estmixing$, which will be shown to capture the latent structure in the original density $p_0$.
More formally, $\estmixing$ can and will be interpreted as a latent projection of $p_0$ onto $\M_{\sigma}$ that removes spurious effects from the data and retains only the latent structure.

This then brings up the key aspect of our approach, namely the idea of \emph{multiscale representation}. More precisely, 
by computing the NPMLE \eqref{eq:npmle formal} across a range of $\sigma$'s, we reveal latent structures of $p_0$ at different scales as we have seen in Figure \ref{fig:toy-npmle}.
\revise{In particular, as $\sigma\rightarrow 0$, the collection of densities $\M_\sigma$ becomes more and more expressive so that $\widehat{p}$ captures the fine scale structures of $p_0$.}
On the other hand, as $\sigma$ becomes larger, the kernel $\phi_\sigma$ becomes flatter so that the stronger projection effect would render a sparser structure in $\estmixing$. 
As in Figures \ref{fig:toy ex}-\ref{fig:toy-npmle}, the support of $\estmixing$ in this regime can inform us about the coarse scale information of $p_0$ such as the number and location of the high density regions. 
As we shall demonstrate later, the structures obtained across different $\sigma$'s can be leveraged to select ideal models that
best represent $p_0$.

\subsection{Summary of Results and Contributions}
The main contribution in this paper is a novel recipe for identifying and estimating the latent structure of general densities based on the idea of multiscale analysis introduced above, which is made practical via the NPMLE. 
We shall present results on both theoretical and practical aspects. 

On the theory side, we first give a characterization of the asymptotic limit of the NPMLE as a projection $p_\sigma$ of the original density $p_0$ (Proposition~\ref{prop:proj unique}).
We then proceed to study the discrete structures in $p_\sigma$ by identifying them with an \emph{intrinsic} notion of components from the perspective of mixture models, without constraining $p_0$ (e.g. without assuming a particular mixture structure, let alone an identifiable one). 
We propose an estimation procedure for these intrinsic components based on the NPMLE and establish its consistency (Theorem~\ref{thm:npmle component consistency}). 
These results justify our approach for estimating the latent structures of general densities that enjoys a rigourous interpretation, and which provides foundations for downstream applications such as cluster analysis.

On the practical side, we propose a model selection procedure for finding a ``most sparse model'' from the multiscale representation that explains the data well and turn this into a clustering algorithm (Algorithm~\ref{algo:ms npmle}). 
This model selection rule moreover gives a natural method to tune the scale parameter as well as determine the number of clusters.
We demonstrate through numerical experiments
that our algorithm is able to resolve a wider range of complex latent structures than standard ones such as $k$-means, spectral clustering, and HDBSCAN.

\subsection{Related Work}

A closely related but distinct line of work is cluster analysis, where two major approaches are \emph{density-based clustering} and \emph{model-based clustering}.
Our work shares many common features with both despite being intrinsically different. 
Most importantly, our objective is more general than clustering: We want to detect discrete latent structure in a density, of which a clustering (i.e. a partition of the data) would be one special case but not the only possibility (specifically, we also discuss component density estimation and hierarchical structure in the form of a dendrogram). In Section~\ref{sec:clustering}, we illustrate a concrete application to the task of clustering.

In density-based clustering, one is often either interested in estimating the connected components of the level sets \citep{hartigan1981consistency,ester1996density,steinwart2011adaptive,sriperumbudur2012consistency,steinwart2015fully,jang2019dbscan++} or the cluster tree of the density \citep{chaudhuri2010rates,stuetzle2010generalized,chaudhuri2014consistent}.  
The (empirical) level sets are similar to the NPMLE as both tend to capture high density regions of the original density $p_0$. 
However, the NPMLE is defined through a more involved optimization scheme and appears to give a more sparse summary of the high density regions as can be seen from Figure \ref{fig:toy-npmle}, where only a few clusters of points are visually present for $\sigma\geq 1.2$. 
A conceptually similar algorithm is the hierarchical DBSCAN \citep{campello2013density} which performs DBSCAN \citep{ester1996density} for range of connectivities $\epsilon$'s and returns a clustering with the best stability over $\epsilon$.
The algorithm searches for clustering structures on high density regions of the data samples whereas our approach works with the NPMLEs instead, and are not necessarily supported on subsets of the data points. 
On the other hand, the sequence of NPMLEs computed for different $\sigma$'s is reminiscent of the cluster tree. 
However, we point out that the nodes of a cluster tree are nested sets whereas the support atoms of $\widehat{G}$ are not; see Remark~\ref{rem:dbc:compare}.
Furthermore, the overall trend for the NPMLE may not even be monotone with respect to $\sigma$ as can be seen from Figure \ref{fig:toy-npmle}.

In model-based clustering, one typically assumes the existence of a mixture decomposition $p_0 = \sum_k\lambda_k p_k$ and attempts to identify and estimate it.
There is a vast literature on model-based clustering \citep{fraley2002} and here we shall refer to the review papers \citet{melnykov2010finite,mcnicholas2016model} and the references therein. 
In the recent work \cite{coretto2023nonparametric}, the authors consider maximum likelihood estimation of mixtures of elliptically symmetric distributions and apply such models for fitting general nonparametric mixtures. 
Another recent work is \cite{do2024dendrogram}, where the authors first estimate an overfitted mixture model, followed by a refinement process based on a dendrogram of the estimated parameters.   
However, we mention that the authors still assume the data to be generated by a finite mixture model and focus only on parametric cases. 
A major difference between our approach and model-based clustering is that we do not assume any model for $p_0$, and allow $p_0$ to be an arbitrary density.

Another closely related line of work is the study of nonparametric mixture models, which play a role in our technical development. 
Instead of specifying an explicit parametric form of the mixture components, the nonparametric approach imposes assumptions on the component densities such as symmetry \citep{bordes2006semiparametric,hunter2007inference}, Markov assumptions \citep{allman2009identifiability,gassiat2016nonparametric}, product structures \citep{hall2003nonparametric,hall2005nonparametric,elmore2005application}, and separation \citep{aragam2020identifiability,aragam2023uniform,tai2023tight}.
Related results on the identifiability and estimation of nonparametric mixtures can also be found in \citet{nguyen2013}. 
A recent Bayesian clustering approach is proposed in \cite{dombowsky2024bayesian} by merging an overfitted mixture under a novel loss function, which is shown to mitigate the effects of model misspecification. The idea of merging overfitted mixtures has recently generated some attention \citep[e.g.][]{aragam2020identifiability,guha2021posterior,aragam2023uniform,dombowsky2024bayesian,do2024dendrogram}.
We emphasize that unlike this line of work, we do \emph{not} assume a    mixture representation---not even a nonparametric one---for the data generating mechanism.

Finally, we return to the key technical device employed, namely the NPMLE, which has attracted increasing attention in recent years.
One of the earlier uses of the NPMLE lies in estimating the mixing measures of mixture models \citep{lindsay1995mixture} and finding superclusters in a galaxy \citep{roeder1990density}. 
In a recent line of work, several authors have continued this endeavor with a focus on establishing convergence rates for density estimation \citep{genovese2000rates,ghosal2001entropies,zhang2009generalized,saha2020nonparametric}, and its application in Gaussian denoising \citep{saha2020nonparametric, soloff2024multivariate}. 
Apart from the various applications, the NPMLE is itself  a mathematically intriguing estimator that carries interesting geometric structures \citep{lindsay1983geometry, lindsay1983geometry2}.
In particular, it can be shown that the NPMLE is a discrete measure supported on at most $n$ atoms where $n$ is the number of observations \citep{soloff2024multivariate}. 
The recent work \cite{polyanskiy2020self} has improved this bound to $O(\log n)$ for one-dimensional Gaussian mixtures with a sub-Gaussian mixing measure, matching the conventional wisdom that usually many fewer support atoms are present. 
Our empirical observation (such as those in Figure \ref{fig:toy ex}) also suggests a tendency for the NPMLE to be supported only on a few atoms when $\sigma$ is large. 
In Proposition \ref{prop:num of atoms sigma} we establish an upper bound in terms of $\sigma$.
Lastly, although the NPMLE is computationally intractable as it is posed as an infinite-dimensional optimization problem, there has been progress on computational aspects of the NPMLE, including a convex approximation \citep{feng2018approximate} and gradient flow-based methods \citep{yan2024learning,yao2024minimizing}, making it practical for high-dimensional problems.

\subsection{Notation}\label{sec:notation}
For a set $\Theta \subset \R^d$, we shall denote $\P(\Theta)$ the set of all probability measures supported on $\Theta$. 
For a function $f\in L^1(\R^d)$ and $G\in \P(\Theta)$ a probability measure over $\R^d$, we denote $f\ast G = \int_{\Theta} f(\cdot-\theta)dG(\theta)$ as the convolution of $f$ with $G$. 
We also denote 
\begin{align}\label{def:supp}
    \supp(G) = \bigcap\{B: B \text{ is closed and } G(B)=1\}.  
\end{align}
For two subsets $A,B\subset \R^d$,  we denote $\dist(A,B)=\inf_{x\in A,y\in B} |x-y|$. If $A=\{x\}$ is a singleton, we shall simply denote $\dist(A,B)$ as $\dist(x,B)$. 
For $\eta>0$, we denote $A(\eta)=\{x:\dist(x,A)\leq \eta\}$.
For $P,Q\in \P(\Theta)$,  the $r$-Wasserstein distance for $r\in [1,\infty)$ is defined as 
\begin{align}\label{def:Wr}
	W_r(P,Q)=\bigg(\underset{\gamma\in \Gamma(P,Q)}{\operatorname{inf}} \int_{\Theta\times \Theta} |x-y|^rd\gamma(x,y)\bigg)^{1/r}, 
\end{align}
where $\Gamma(P,Q)$ is the set of couplings between $P$ and $Q$.

The rest of the paper is organized as follows. In Section \ref{sec:main}, we formalize our setting and present preliminary results that serve as foundations for later development. 
In Section \ref{sec:est latent}, we present our main results on identifying and estimating the latent structures of general densities. 
Section \ref{sec:clustering} describes a clustering algorithm that arises from an application of our estimation procedure, followed by numerical experiments in Section \ref{sec:numerical}.
All proofs are deferred to the Appendix.

\section{Background}\label{sec:main}

In this section we shall formalize our setting by making precise the various concepts mentioned above.  
We shall present preliminary results on the characterization of the asymptotic limit of the NPMLE and discuss its interesting geometric structure in the context of multiscale analysis. 
To start with, we shall introduce an important family of densities, alluded above as the set of densities with minimum length-scale $\sigma$,  which will play a crucial role in our later development.  

\paragraph{Assumptions on $p_0$.}
Here and throughout the rest of the paper, the multivariate density $p_0$ that generates our data is allowed to be arbitrary: We will not need to impose any further regularity conditions on $p_0$ besides the fact that it is a density on $\R^d$. In fact, a major contribution of our framework is to define---for any multivariate $p_0$---an appropriate multiscale latent representation $p_\sigma$ ($\sigma>0$) that will be the target of estimation.

\subsection{Convolutional Gaussian Mixtures}

We begin by formalizing the family of probability measures $\M_\sigma$ used in the sequel for defining the NPMLE.
\revise{Let $\Theta \subset \R^d$ and}
\begin{align}\label{eq:cvg}
    \M_\sigma(\Theta) 
    = \bigg\{ \int_{\R^d} \phi_\sigma(\cdot-\theta) dG(\theta): G\in\P(\Theta) \bigg\}
    = \bigg\{ \phi_\sigma\ast G : G\in\P(\Theta) \bigg\},
\end{align}
where $\phi_\sigma$ is the probability density function of the multivariate Gaussian $\mathcal{N}(0,\sigma^2 I_d)$, and $\P(\Theta)$ is the set of all probability measures supported on $\Theta$. 
Densities in $\M_\sigma(\Theta)$ can be interpreted as a continuous mixture of Gaussians $\mathcal{N}(\theta,\sigma^2 I_d)$, where the center $\theta$'s are encoded in the mixing measure $G$.
The set $\M_\sigma(\Theta)$ includes finite Gaussian mixtures as a special case  when the mixing measure $G=\sum_{i=1}^K w_k\delta_{\theta_i}$ is a sum of Dirac delta functions. 
As $G$ varies over $\P(\Theta)$, \eqref{eq:cvg} forms a rich nonparametric family of densities that will allow us to extract latent information through the mixing measure $G$.
Furthermore, given any $p\in\M_\sigma(\Theta)$ the underlying mixing measure $G$ is identifiable so that its estimation is possible.

\subsection{Nonparametric Maximum Likelihood Estimator}
Now that we have defined the space $\M_\sigma(\Theta)$ of candidate densities, we are ready to make precise our definition of the NPMLE. 
Let $\{Y_i\}_{i=1}^n$ be i.i.d. samples from $p_0$. For $\sigma>0$, define
\begin{align*}
    \widehat{p}_{n,\sigma}:=\underset{p\in \M_\sigma(\Theta)}{\operatorname{arg\,max}}\,\, \sum_{i=1}^n \log p(Y_i).
\end{align*}
Since densities in $\M_\sigma(\Theta)$ are of the form $p=\phi_\sigma\ast G$ for $G\in \P(\Theta)$, 
the above maximization can be equivalently cast as
\begin{align}\label{eq:NPMLE}
    \estmixing := \underset{G\in\mathcal{P}(\Theta)}{\operatorname{arg\, max}} \,\, \sum_{i=1}^n \log (\phi_\sigma \ast G) (Y_i),
\end{align}
where we take any maximizer if there are multiple. 
As mentioned in the introduction, the measure $\estmixing$ is the latent probability measure of interest and for this reason we shall work with the definition \eqref{eq:NPMLE} for the rest of the paper.

A natural question that follows is whether the NPMLE defined in \eqref{eq:NPMLE} has a valid asymptotic limit as $n\rightarrow \infty$. 
To answer this, we interpret $\estmixing$ as approximating the KL projection of $p_0$ onto the space $\M_\sigma(\Theta)$. 
To make this precise, let's define 
\begin{align}\label{eq:projection}
    p_\sigma :=\underset{p\in \M_\sigma(\Theta)}{\operatorname{arg\,min}} \,\,\KL{p_0}{p} = \underset{p\in \M_\sigma(\Theta)}{\operatorname{arg\,min}} \,\, \int_{\R^d} p_0 \log \frac{p_0}{p}. 
\end{align}
In other words, $p_\sigma$ is the density in $\M_\sigma(\Theta)$ that is closest to $p_0$ in KL-divergence. 
Notice that the KL-divergence between $p_0$ and any element $p\in \M_\sigma(\Theta)$ is always well-defined since the later is non-vanishing. 
The following result relates the NPMLE $\estmixing$ to the projection $p_\sigma$, which generalizes \cite{kiefer1956consistency}.

\begin{proposition}\label{prop:proj unique}
\revise{Let $\Theta$ be a compact set and $p_0$ be any density. 
There exists a unique $G_\sigma\in \P(\Theta)$ such that $p_\sigma = \phi_\sigma \ast G_\sigma$ solves \eqref{eq:projection}. 
Furthermore, for each $r\in[1,\infty)$ we have 
\begin{align*}
    W_r\big(\estmixing,\projmixing) \xrightarrow{n\rightarrow \infty} 0
\end{align*}
almost surely, where $W_r$ is the Wasserstein distance defined in \eqref{def:Wr}.
}
\end{proposition}

\begin{proof}
The proof can be found in Appendix \ref{sec:proof of proj uniqueness}. 
\end{proof}

\noindent
The mixing measure $\projmixing$ established in the previous proposition will play an important role in the sequel.

\begin{remark}\label{remark:NPMLE=projection}
Proposition \ref{prop:proj unique} confirms the fact that $\estmixing$ is a consistent estimator of the mixing measure $\projmixing$ of the projection $p_\sigma$ in $W_r$ distance. 
The result then suggests that in the asymptotic regime, the NPMLE resembles the projection $p_\sigma$ so that they should share similar structures. 
For this reason, we shall base our discussion below on either the NPMLE or $p_\sigma$ interchangeably since sometimes the population limit $p_\sigma$ provides more insight without the distraction of finite sample effects.
\end{remark}

\subsection{Geometry of NPMLE}

In the multivariate case, it was recently shown  \citep[Lemma 1]{soloff2024multivariate} that $\estmixing$ defined in \eqref{eq:NPMLE} is a discrete measure supported only on at most $n$ atoms despite the search space being the space of all probability measures on $\Theta$. 
A crucial observation that will lead to the idea of multiscale analysis is that the NPMLE could exhibit very different structures for different choices of $\sigma$'s.   
In particular, as $\sigma$ increases, the number of support atoms present in $\estmixing$ tends to decrease and concentrate towards the high density regions of $p_0$. 
This can be already seen in Figure \ref{fig:toy-npmle}. 
The following result then makes this rigorous in the one-dimensional case.

\begin{proposition}\label{prop:num of atoms sigma}
Suppose $d=1$ and the $Y_i$'s are ordered increasingly. Define $r=\frac{Y_n-Y_1}{2}$. We have 
\begin{align}\label{eq:no of atoms}
    \textup{Number of atoms in } \estmixing 
    \leq 1.90 + \frac{(Y_n+10)r}{0.85\sigma^2}.
\end{align}    
If $\sigma>r$, $\estmixing$ is a Dirac delta measure located at the mean $\bar{Y}=n^{-1}\sum_{i=1}^n Y_i$. 
\end{proposition}
\begin{proof}
    The proof can be found in Appendix \ref{sec:proof of no of atoms}.
\end{proof}

Notice that the right hand side of \eqref{eq:no of atoms} is a decreasing function of $\sigma$ and if $\sigma$ exceeds certain threshold,  only one flat Gaussian remains. 
However, the result does not imply that the number of atoms in $\estmixing$ decreases monotonically with respect to $\sigma$, but only an overall decreasing trend as can be seen also in Figure \ref{fig:toy-npmle}. 
A typical observation is that if $p_0$ has $K$ well-separated modes, then for certain ranges of $\sigma$'s $\estmixing$ will have (close to) $K$ atoms located around these modes. 
If one keeps increasing $\sigma$, those $K$ atoms will merge further but many more atoms could be present during such transition. 
A more refined characterization of the atom locations in $\estmixing$ is an interesting theoretical question to be investigated in the future. 
Theoretical analysis of NPMLE is still emerging and we mention the recent works \cite{polyanskiy2020self}, which gives an $O(\log n)$ upper bound on the number of support atoms for a one-dimensional sub-Gaussian model and \cite{soloff2024multivariate,yan2024learning} which establish the existence of the NPMLE in general dimensions. Evidently, even basic questions about the NPMLE remain unresolved.

\begin{remark}\label{remark:npmle is discrete}
There is an important qualitative difference between the NPMLE $\estmixing$ and its limit $\projmixing$, namely that $\estmixing$ is always a discrete measure regardless of whether its limit $\projmixing$ is continuous or not \citep[Lemma 1]{soloff2024multivariate}.
Therefore, the geometric structures in $\estmixing$ and $\projmixing$ should be analyzed using different approaches as for instance the notion of connected components in the support of $\estmixing$ does not immediately make sense. 
We shall address this issue in more detail in Section \ref{sec:clustering}. 
\end{remark}

\section{Multiscale Representation}\label{sec:multiscale}

With this preparation, we are now ready to make precise the idea of multiscale analysis. 
Following Remark \ref{remark:NPMLE=projection}, we shall focus the discussion below on the projections $p_\sigma$'s. 
By Proposition~\ref{prop:proj unique},
we have $p_\sigma=\phi_\sigma\ast G_\sigma$ for some $\projmixing\in \P(\Theta)$. 
Intuitively, we can view $p_\sigma$ as a limiting kernel density estimator with bandwidth $\sigma$, with $\projmixing$ containing possibly infinitely many centers.  
The choice of $\sigma$ determines the bandwidth used to approximate $p_0$ and affects the structure of the surrogate $p_\sigma$.
\revise{Unlike a kernel density estimator, which has the same centers regardless of the bandwidth, $\projmixing$ has centers (atoms) that vary depending on $\sigma$.}
The central idea underlying the multiscale representation is that for different choices of $\sigma$'s, the measures $\projmixing$'s will exhibit structures of $p_0$ at different scales, which we illustrate next.

\begin{figure}[t]
    \centering
    \minipage{0.3\textwidth}
    \includegraphics[width=\linewidth]{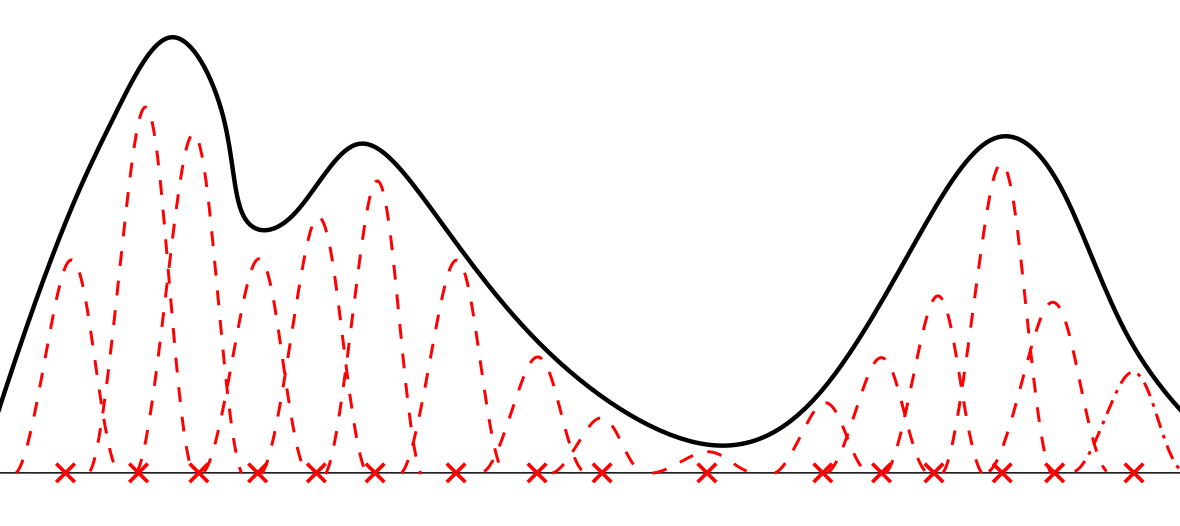}
    \subcaption{}\label{fig:ms:a}\vspace{-5pt}  
    \endminipage\hfill
    \minipage{0.3\textwidth}
    \includegraphics[width=\linewidth]{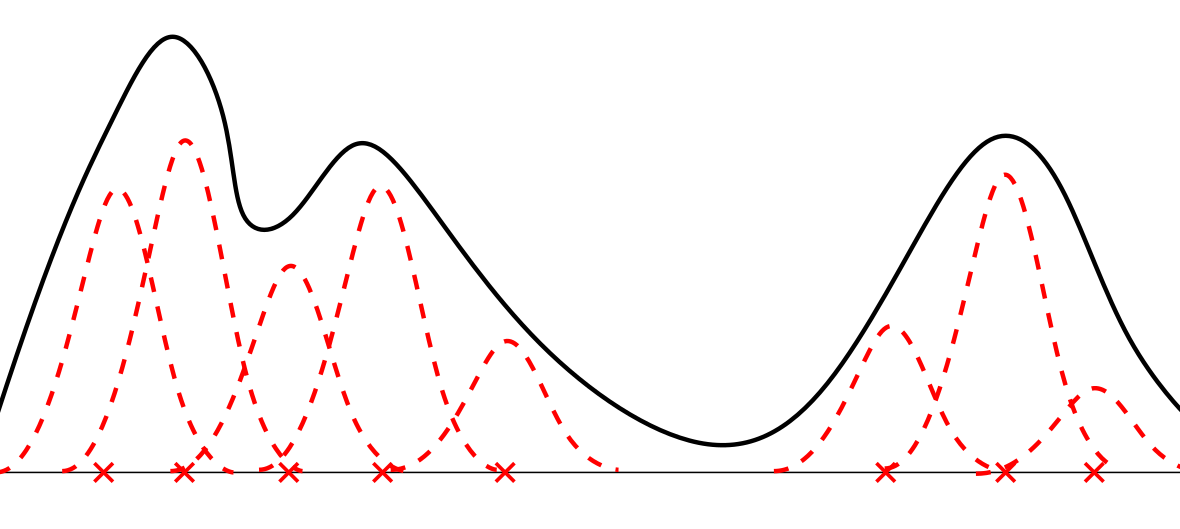}
    \subcaption{}\label{fig:ms:b}\vspace{-5pt}  
    \endminipage\hfill
    \minipage{0.3\textwidth}
    \includegraphics[width=\linewidth]{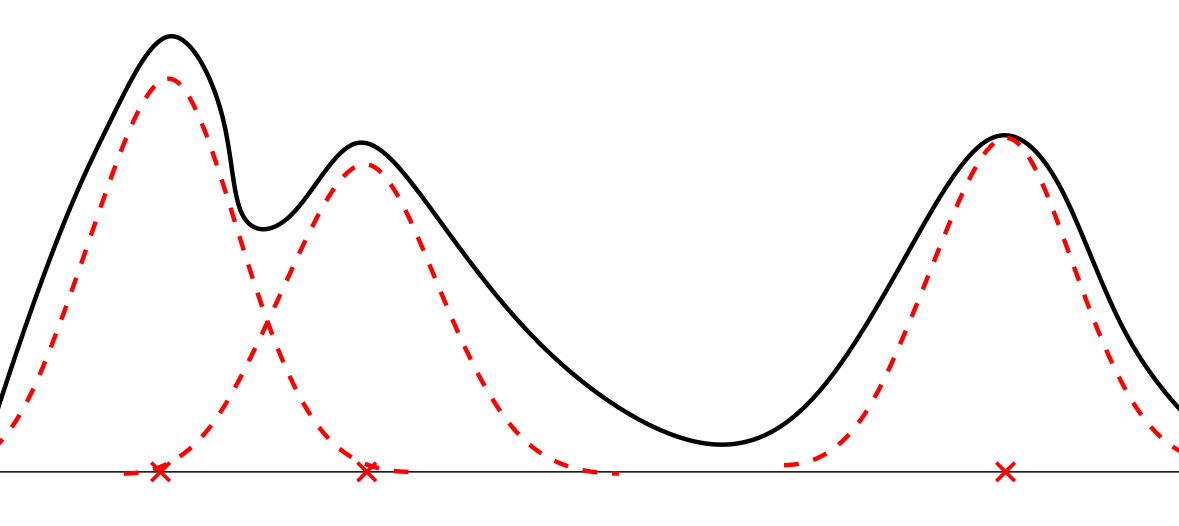}
    \subcaption{}\label{fig:ms:c}\vspace{-5pt}  
    \endminipage\hfill
    \caption{Visualization of the multiscale representation as $\sigma\rightarrow \infty.$ The crosses represent the support atoms of the $\projmixing$'s with the dashed curves representing the associated Gaussians. }
    \label{fig:ms}
\end{figure}

Figure \ref{fig:ms} gives a visualization for the multiscale representation of a density $p_0$.  
For general compact $\Theta$, as $\sigma\to0$,
$\M_\sigma(\Theta)$ becomes more expressive and its distance to $p_0$ decreases so that $p_\sigma$ gives a better and better approximation to $p_0$. 
In this case, the associated measure $\projmixing$ tends to be ``dense'', meaning that substantially many Gaussians whose centers lie close to each other are needed for fitting $p_0$ due to the small bandwidth $\sigma$ (Figure~\ref{fig:ms:a}). 
On the other hand, as $\sigma$ gets larger,  flatter Gaussians are used instead to approximate $p_0$, in which case one expects only a few of them to be present in $p_\sigma$ (Figure~\ref{fig:ms:c}). 
The underlying mixing measure $\projmixing$ then tends to be ``sparse'' and have smaller support sizes compared to the case of small $\sigma$'s. 
In the extreme case of $\sigma \rightarrow \infty$, the projection will approximate a single flat Gaussian centered at the mean of $p_0$ (cf. Proposition~\ref{prop:num of atoms sigma}).

Of particular interest is the intermediate regime where $\sigma$ is moderately large, as in Figure~\ref{fig:ms:c}. 
In this case, the projection $p_\sigma$ still gives a reasonable approximation to $p_0$, while at the same time yields a sparse mixing measure $\projmixing$ whose atoms are centered near the high density regions of $p_0$. 
Due to sparsity, the atoms are likely to be well-separated and a simple clustering of them would allow us to extract the number and locations of the different high densities regions of $p_0$. This is precisely the type of structure captured by the motivating example in Figure \ref{fig:toy-npmle-3}.

\revise{It is helpful to consider the special case where $p_0\in\M_\sigma(\Theta)$: In this case, there is some $\sigma>0$ where this projection is exact, i.e. $p_0=p_\sigma$ and thus $G_\sigma$ exactly captures the latent structure. In the general case with $p_0\notin\M_\sigma(\Theta)$, it helps to think of $p_0$ as being ``close'' to some $p_\sigma=\phi_\sigma \ast G_\sigma\in\M_\sigma(\Theta)$, with the idea being that $G_\sigma$ is then a useful approximation to the latent structure of $p_0$. Of course, a crucial point is that this approximation is not mathematically necessary in the analysis, although it provides useful insight here into the types of structure that the NPMLE picks up on.}
In the next section, we shall discuss how to identify and estimate the finer structures within each of the $\projmixing$'s or their empirical counterparts $\estmixing$'s.

\section{Estimation of Latent Structures}\label{sec:est latent}

With these preliminaries out of the way, we now continue to discuss estimation of the latent structure as represented by the multiscale representation $p_\sigma$.
Our goal is to use this representation to construct a well-defined notion of latent structure for $p_{0}$ across multiple scales $\sigma$, and then to give a procedure for estimating this latent structure. The resulting structure will be represented by a latent mixing measure for $p_{\sigma}$ that approximately captures the structure in $p_{0}$. The main result of this section (Theorem~\ref{thm:npmle component consistency}) will then construct strongly consistent estimators of this multiscale structure.

As before, let's start the discussion from the projections by defining a notion of components for the projections as representing the latent structures of $p_0$ and then propose an estimation procedure for recovering them. 

\begin{remark}
\label{rem:trueparam}
Throughout this section, we remind the reader that there are no ``true'' values of the scale parameter or the number of components.
This is because we allow $p_{0}$ to be an arbitrary, unstructured density. The goal is to find representative values that reflect useful latent structure, for instance in downstream tasks such as clustering (Sections~\ref{sec:clustering}-\ref{sec:numerical}).
\end{remark}

\subsection{Components of the Projections}

We are especially interested in the case where $p_0$ appears to be loosely comprised of multiple subpopulations, or unobserved latent classes, but for which we lack identifying assumptions such as Gaussianity for these classes. 
See Figure~\ref{fig:toy ex}. Bearing this in mind, we will not \emph{assume} such a structure explicitly exists, and instead will use the latent projections $p_\sigma$ to locate approximations to this latent structure, and to analyze them from a mixture modeling perspective.
\revise{
The idea is that if $p_0$ has latent structure, then the latent projection $p_\sigma$
should inherit this structure to some degree. 
Furthermore, although $p_0$ may not be a mixture, the latent projection $p_\sigma$ decomposes into a mixture model in a canonical way.
Therefore, a natural target would be to estimate the components of $p_\sigma$, with the hope of them revealing the latent structure of the original density $p_0$.}

\revise{To define such a mixture structure, we shall exploit the connected components of $\supp(G_\sigma)$.} We begin with the following decomposition of $\projmixing$:
\begin{definition}\label{def:proj components}
Let $\projmixing$ be the projection mixing measure defined in \eqref{eq:NPMLE}. 
Let $\{\suppk\}_{k=1}^{N_\sigma}$ be the connected components of $\supp(\projmixing)$ (cf. \eqref{def:supp}) so that  $\supp(\projmixing)=\bigcup_{k=1}^{N_\sigma}\suppk$. 
Define the following decomposition of $\projmixing$ into its connected components: 
\revise{
\begin{align}\label{eq:decomposition of G sigma}
    \projmixing 
    =\sum_{k=1}^{N_\sigma} \underbrace{\projmixing(\sp{S}{k})}_{\lambda_{\sigma,k}} \underbrace{\projmixing(\cdot \,|\, \sp{S}{k})}_{G_{\sigma,k}}
    =:\sum_{k=1}^{N_\sigma} \lambda_{\sigma,k} G_{\sigma,k}.
\end{align}
\revise{Owing to Proposition~\ref{prop:proj unique}, this decomposition always exists for any $p_0$ and is well-defined.}
Accordingly, $p_\sigma$ can be decomposed as 
\begin{align}\label{eq:decomposition of p sigma}
    p_\sigma=\sum_{k=1}^{N_\sigma} \lambda_{\sigma,k} \underbrace{\phi_\sigma\ast G_{\sigma,k}}_{f_{\sigma,k}} =: \sum_{k=1}^{N_\sigma} \lambda_{\sigma,k}f_{\sigma,k}.
\end{align}}
\end{definition}

\noindent
Recall that $\supp(G_\sigma)$ is the smallest closed set $E$ such that $\projmixing(E)=1$.
If $\projmixing$ has a density $g_\sigma$, then this is simply the closure of the set $\{g_\sigma>0\}$.  
The definition also includes the case where the support of $\projmixing$ is a lower-dimensional set in $\R^d$.

Definition \ref{def:proj components} decomposes $\projmixing$ into a mixture of $N_\sigma$ components, where each $G_{\sigma,k}$ is simply the restriction of $\projmixing$ onto its $k$-th connected component $\sp{S}{k}$. 
\revise{We shall treat each $G_{\sigma,k}$ as a component of the projection $G_\sigma$ and this defines the object of interest $(\lambda_{\sigma,k},f_{\sigma,k})$ that we will estimate in the next subsection.
In particular, the $(\lambda_{\sigma,k},f_{\sigma,k})$'s across different $\sigma$'s represent a \emph{multiscale structure} of the original density $p_0$, the examination of which reveals rich information on the latent structure of $p_0$.
}

\begin{remark}
\label{rem:dbc:compare}
    The collection of sets $\{\suppk\}_{k=1}^{N_\sigma}$ in Definition \ref{def:proj components} for a fixed $\sigma$ bear a superficial resemblance to the level sets of density-based clustering for some fixed threshold $\lambda$. 
    The $\suppk$'s are indeed high density regions, but as $\sigma$ varies, they do not necessarily form a tree structure like the level sets do because $\supp(\projmixing)$ may not be nested.  
    Nevertheless, this helps to provide a complementary interpretation of our approach from the lens of density-based clustering. See Section~\ref{sec:numerical} for numerical comparisons. %
\end{remark}

\subsection{Estimating the Components}\label{sec:component est}

With the decomposition \eqref{eq:decomposition of p sigma} in mind, we shall now focus on estimating 
\revise{this multiscale structure, i.e.}
the weights $\lambda_{\sigma,k}$ and 
the densities $f_{\sigma,k}$.
The procedure is almost ready as Proposition \ref{prop:proj unique} suggests that $\estmixing$ is a consistent estimator of $\projmixing$ so that asymptotically it should also satisfy a discrete version of \eqref{eq:decomposition of G sigma}, where the support atoms of $\estmixing$ have $N_\sigma$ ``connected components'' that are further separated by a positive distance. %
Therefore an appropriate clustering step 
should return these $N_\sigma$ components.
We remark that this intuition is often correct in practice, especially with moderately large $\sigma$'s where these ``connected components'' are well-separated. 
Examples include those in Figures \ref{fig:simulated NPMLE best} and \ref{fig:benchmark NPMLE best} where we can clearly see a clustering structure. 

However, since $\estmixing$ is always a discrete measure, this intuition is not precise, and we must be careful when comparing the decomposition of the continuous measure $G_\sigma$ with its discrete approximation $\estmixing$; recall Remark~\ref{remark:npmle is discrete}.
Especially from a theoretical perspective,  the $W_r$ convergence established in Proposition \ref{prop:proj unique} is insufficient to guarantee this kind of ``nice'' structure: There are many more possibilities for the configurations of the atoms in $\estmixing$ that are consistent with Wasserstein (i.e. weak) convergence. 
In particular, there could be atoms with vanishingly small weights lying in between the different connected components that would ruin their well-separatedness. %

Therefore, we need to cluster the atoms of $\estmixing$ with care.
Our approach is to first employ a preprocessing step that identifies the high density regions of $\estmixing$, which are well-separated and easy to cluster. 
This is accomplished by convolving $\estmixing$ with a multivariate box kernel $I_{\delta_n}=(2\delta_n)^{-d}\mathbf{1}_{[-\delta_n,\delta_n]^d}$ for a suitable bandwidth $\delta_n>0$, and then applying a clustering algorithm to the level sets of the resulting density. 
More precisely, applying single-linkage clustering to the collection of open sets in $\{\estmixing\ast I_{\delta_n}>t_n\}$ for some suitable $\delta_n,t_n>0$, we obtain a collection of sets $\{\widehat{S}_{\sigma,k}\}_{k=1}^{N_\sigma}$ that almost recover $\{S_{\sigma,k}\}_{k=1}^{N_\sigma}$.
The final step is to partition the parameter space $\Theta=\bigcup_{k=1}^{N_\sigma}E_k$ with the Voronoi partition induced by these sets:
\begin{align}\label{eq:def Ek}
    E_k= \{x\in\Theta:\dist(x,\widehat{S}_{\sigma,k}) \leq \dist(x,\widehat{S}_{\sigma,j})\quad \forall j\neq k\}.
\end{align}
Using this partition, we define our estimators as  
\begin{align}
\label{eq:cd}
    \left\{
    \begin{aligned}
        \widehat{\lambda}_{n,\sigma,k} &= \estmixing(\voronoik), \\
        \widehat{f}_{n,\sigma,k} 
        &= \phi_\sigma\ast \widehat{G}_{n,\sigma,k}, 
        \quad\text{where}\quad
        \widehat{G}_{n,\sigma,k}
        = \estmixing(\,\cdot\,|\voronoik).
    \end{aligned}
    \right.
\end{align}
Our main result below states that $\widehat{\lambda}_{n,\sigma,k}$ and $\widehat{f}_{n,\sigma,k}$ are strongly consistent estimators of the multiscale structure of $p_{0}$, i.e. $(\lambda_{\sigma,k},f_{\sigma,k})$. 
\begin{theorem}\label{thm:npmle component consistency}
Let $p_{0}$ be any density.
For each $\sigma>0$, define $\widehat{\lambda}_{n,\sigma,k}$ and $\widehat{f}_{n,\sigma,k}$ as in \eqref{eq:cd}, where 
\begin{align*}
    \delta_n\rightarrow 0, \quad t_n\rightarrow 0, \quad t_n\geq 2^{-d}\delta_n^{-(d+1)}d^{-1/2}W_1(\estmixing,\projmixing).
\end{align*}
For instance, $\delta_n$ and $t_n$ can be chosen as two slowly decaying sequences. 
Then we have
\begin{align*}
    \underset{k}{\operatorname{max}}\, \Big[|\widehat{\lambda}_{n,\sigma,k}-\lambda_{\sigma,k}| \vee \|\widehat{f}_{n,\sigma,k}-f_{\sigma,k}\|_1\Big] \to 0 
    \quad\text{a.s. as $n\to\infty$.}
\end{align*}
\end{theorem}
\begin{proof}
    The proof can be found in Appendix \ref{sec:proof of main theorem}.
\end{proof}

\noindent
This crucial result shows that the latent structure is not only well-defined for any $p_0$, but estimable directly via the NPMLE.

\subsection{Qualitative features of the latent structure}

The NPMLE $\estmixing$ captures the latent structure of $p_{0}$ at different scales in three ways (cf. Figure~\ref{fig:toy-features}): 1) The dendrogram of its atoms, 2) The class conditional densities $\widehat{f}_{n,\sigma,k}$, and 3) The Bayes partition (defined below). These objects offer a \emph{qualitative} glimpse into the latent structure of $p_{0}$---at finite samples and at multiple resolutions $\sigma$---above and beyond the consistency guarantees provided by Theorem~\ref{thm:npmle component consistency}.
As with the rest of this section, everything here is valid for any $\sigma>0$.
An obvious caveat is that we have not demonstrated how to choose a ``good'' scale $\sigma$ in practice: This will be the subject of the next section.

\paragraph{Dendrogram}
Since $\estmixing$ is a discrete measure, we can compute a dendrogram over its atoms (Figure~\ref{fig:toy dg}) using standard hierarchical clustering techniques. Throughout this paper, we have chosen to focus on single-linkage clustering, however, other approaches can be used. This dendrogram is a significant advantage of our NPMLE approach: It provides an interpretable, qualitative certificate of the latent structure along with a means for model assessment and validation. 
This allows the statistician to inspect the output, look for nearby subclusters, and assess the relationship between clusters.
\revise{The practical value of this dendrogram will be demonstrated in our experiments; %
Figure~\ref{fig:simulated NPMLE best} (Section~\ref{sec:numerical:sim}) finds well-separated clusters whereas Figure~\ref{fig:benchmark NPMLE best} (Section~\ref{sec:numerical:benchmark}) conveys that some clusters possess subclusters that might warrant further investigation. }
This is in contrast to many clustering methods that simply output a black-box partition of the space. 

As we will discuss in Section~\ref{sec:estimate K}, the dendrogram can also be used to determine the number of clusters: In Figure~\ref{fig:toy dg}, it is clear that there are four well-separated clusters. In datasets where clusters are not well-separated, the dendrogram will indicate this.

\paragraph{Class conditional densities}
The class conditional densities (Figure~\ref{fig:toy cd}), defined in \eqref{eq:cd}, provide both qualitative and quantitative measures of latent class membership. 
For example, for a given $x$, $\widehat{f}_{n,\sigma,k}(x)$ provides a quantitative estimate of class membership, providing a level of confidence in classifying this observation into a particular class. The class conditional densities are also used to define the Bayes classifier and partition---and hence a clustering---as described below.
Importantly, through Theorem~\ref{thm:npmle component consistency}, the $\widehat{f}_{n,\sigma,k}$ provide consistent estimates of the latent structure, providing assurances that these estimates are meaningful and targeting a well-defined estimand.

\paragraph{Bayes partition and clustering}
We can use the class conditional densities to define the Bayes partition (Figure~\ref{fig:toy partition}), which gives a clustering of the original data as well as a way to classify new observations. Compared to standard black-box clustering algorithms which only provide an \emph{in-sample} clustering, the Bayes partition provides a way to cluster new points \emph{out-of-sample}, in addition to the usual in-sample clustering. 

Formally, the Bayes partition is defined for any $x\in\R^{d}$ by (breaking ties arbitrarily)
\begin{align}\label{eq:bayes optimal partition}
c(x)
:= \underset{k}{\operatorname{arg\,max}}\,\, \widehat{\lambda}_{n,\sigma,k}\widehat{f}_{n,\sigma,k}(x).
\end{align}
Then $c(x)\in[K]$ defines a classifier that is used to classify each point $x\in\R^{d}$ into one of $K$ classes, thereby providing a partition called the Bayes partition.

\section{A Clustering Algorithm}\label{sec:clustering}

The procedure in Section \ref{sec:est latent} gives a recipe for consistently estimating the \revise{multiscale structure of the original density $p_0$.} %
In this section, we shall illustrate an application of this procedure on the task of clustering.
The overall idea is to look for an appropriate $\sigma$ so that the corresponding NPMLE reveals the clustering structure of the data set. 
By examining a sequence of surrogate models, we propose a model selection criterion that picks the most representative $\sigma$ and exploits this to obtain a clustering rule.
The procedure also leads to a selection of the number of clusters by examining the dendrogram of this most representative model.

\subsection{Model Selection}

Ideally, we would like to select a $\sigma$ so that the projection $p_\sigma$ is reasonably close to the original density $p_0$ while at the same time has an intepretable clustering structure (e.g. Figure~\ref{fig:toy-npmle-3}).
From the perspective of approximation, a smaller $\sigma$ would lead to a projection $p_\sigma$ that better fits $p_0$ as the space $\M_\sigma(\Theta)$ becomes richer. 
Therefore if we only focus on how well the projection fits the data, we will end up with choosing a vanishingly small $\sigma.$
However, the resulting NPMLE $\estmixing$ would be densely packed so that there is only one sensible giant component (cf. Figure~\ref{fig:toy-npmle-1}-\ref{fig:toy-npmle-2}), where no useful latent structure can be inferred. 
We can think of this regime as overfitting the data with an extremely complex model. 
Therefore a natural route is to incorporate a penalty term that discourages overly complex models.

To define an appropriate notion of complexity, we shall again rely on the idea of counting the number of connected components as in Definition \ref{def:proj components}, but in a slightly different way as foreshadowed in Remark \ref{remark:npmle is discrete} since we are working with the discrete measure $\estmixing$, where additional care is needed in making the notion of connected component precise. 
To this end, we shall consider the maximal $\epsilon$-connected subsets of $\supp(\estmixing)$ as the discrete analog of connected components, where $\epsilon$ is a parameter to be chosen later, and count the number of such discrete connected components.
For practical computation, we employ the DBSCAN algorithm \citep{ester1996density}, which precisely achieves this goal along with an additional denoising effect. 
More precisely, let's make the following definition.

\begin{definition}\label{def:dbscan K}
Let $\widehat{K}(\sigma)$ be the number of (non-noise) clusters returned by the DBSCAN algorithm applied to the atoms of $\estmixing$ with $\epsilon=2\sigma$ and minPts=1. 
\end{definition}

In particular, we are setting the connectivity parameter $\epsilon$ to be proportional to $\sigma$.  
The rationale of such a choice is that
to decipher the connected components of $\estmixing$, one needs to set $\epsilon$ so that the (normalized) radial kernel $\mathbf{1}_{|x|\leq \epsilon}$ approximates $\phi_\sigma$ reasonably well, leading to the choice $\epsilon=2\sigma$. %
The choice minPts=1 ensures that every point is treated as a core point in DBSCAN and allows any single atom to be a cluster on its own. 
In particular, the $\widehat{K}(\sigma)$ in Definition \ref{def:dbscan K} decreases monotonically as $\sigma$ increases, reflecting the model complexity. 

With this preparation, we shall employ the Bayesian information criterion \citep[BIC,][]{schwarz1978} for selecting $\sigma$.
Treating $\estmixing$ as a mixture of $\widehat{K}(\sigma)$ components in $d$-dimensions, its model complexity is $d\widehat{K}(\sigma)$. 
Precisely, let $\Sigma \subset \R_+$ be a set of candidate $\sigma$'s. 
We consider
\begin{align}\label{eq:bic}
    \sigmahat = \underset{\sigma\in \Sigma}{\operatorname{arg\,min}}\,\, \big[ -2\widehat{\ell}_n(\sigma) + d\widehat{K}(\sigma)  \log n \big],
\end{align}
where 
\begin{align}\label{eq:llg}
    \widehat{\ell}_n(\sigma) = \sum_{i=1}^n \log \widehat{p}_{\sigma}(Y_i)
\end{align}
is the log-likelihood of the projection $\widehat{p}_{\sigma}=\estmixing\ast \phi_\sigma$ and $\widehat{K}(\sigma)$ is defined in Definition \ref{def:dbscan K}.

\begin{algorithm}[t]
\caption{Multiscale NPMLE Clustering}\label{algo:ms npmle}
\begin{algorithmic}
\Require A set $\Sigma$ of candidate $\sigma$'s, number of clusters $K$ (optional). 
\For{$\sigma\in \Sigma$}:
\State Compute NPMLE $\estmixing$ and the log-likelihood $\widehat{\ell}_n(\sigma)$
\State Apply DBSCAN to the atoms of $\estmixing$ with $\epsilon=2\sigma$ and minPts=1 to get the number of clusters $\widehat{K}(\sigma)$
\EndFor
\State Find $\sigmahat=\operatorname{arg\,min}_{\sigma\in \Sigma}[ -2\widehat{\ell}_n(\sigma)+d\widehat{K}(\sigma)\log(n)].$
\State If $K$ is not provided, find a good candidate by examining the dendrogram of $\supp(\selected).$ 
\State Apply single-linkage clustering to $\supp(\selected)$ and obtain the weighted densities $\widehat{p}_k$'s defined in \eqref{eq:BC component}.
\State Return the Bayes classifier $c(x)=\operatorname{arg\,max}_{k=1,\ldots,K} \widehat{p}_k(x)$.
\end{algorithmic}
\end{algorithm}

\subsection{Over-smoothing} 
Intuitively speaking, \eqref{eq:bic} attempts to find the model that can be explained by fewest number of components yet still fits the data well. 
However, our empirical observation is that the NPMLE selected is often too ``noisy'' to reveal a clear clustering structure. 
This is also related to our discussion on the need of a preprocessing step in our estimation procedure in Section \ref{sec:component est}. 
Here we propose a simple remedy via the use of an over-smoothed projection, namely the projection computed with bandwidth $2\sigmahat$:
\begin{align}\label{eq:selected p}
    \widehat{p}_{2\sigmahat} = \phi_{2\sigmahat}\ast \widehat{G}_{n,2\sigmahat}.
\end{align}
The motivation comes from the observation in Proposition \ref{prop:num of atoms sigma} and in Figure \ref{fig:toy-npmle} that the NPMLEs computed with larger $\sigma$'s tend to have fewer support atoms. 
By choosing a larger $\sigma$ than that returned by \eqref{eq:bic}, we get a denoising effect. 
Returning to the example in Figure \ref{fig:toy-npmle}, the selected model is shown in the third figure, which has a much clearer structure than in the original data samples. Observe that this constant-order multiplicative correction still satisfies the assumptions of Theorem~\ref{thm:npmle component consistency}.

\subsection{Selection of Number of Components}\label{sec:estimate K}

Now with the proxy model $\widehat{p}_{2\sigmahat}$ \eqref{eq:selected p}, the remaining step is to cluster the support of its mixing measure $\widehat{G}_{n,2\sigmahat}$ to form a partition of the parameter space, based on which we cluster the original data points. 
To start with, we need to determine an appropriate number of components $K$ to decompose $\widehat{G}_{n,2\sigmahat}$ since such knowledge is usually not available or there is even no ground truth for $K$ (cf. Remark~\ref{rem:trueparam}). 

A natural choice would be to use $\widehat{K}(2\sigmahat)$ defined in Definition~\ref{def:dbscan K}, which we recall should be interpreted as the number of connected components of the discrete set $\supp(\widehat{G}_{n,2\sigmahat})$. 
We remark that this often gives a good estimate of the genuine number of clusters present, but in some cases it could lead to an overestimation due to finite sample issues.  
Figure~\ref{fig:toy dendro} shows three illustrative examples from the experiments in Section~\ref{sec:numerical}, with the supports of the selected model $\selected$'s (top row) and the associated dendrograms (bottom row), along with the estimates $\widehat{K}(2\sigmahat)$'s.
The three examples consist of 4, 3, and 4 clusters respectively, which is clearly observed from the dendrograms.
However, for the third example which comes from a benchmark dataset, $\widehat{K}(2\sigmahat)$ overestimates this by breaking up some clusters. 
Therefore we propose instead to use the dendrogram as a model diagnostic to select the number of clusters. The idea is that when $\widehat{K}(2\sigmahat)$ is a good estimate, this will also be clear from the dendrogram, and if not, the dendrogram will provide qualitative guidance on how to select $K$. See Remark~\ref{sec:selection of K} for details and Section~\ref{sec:numerical} for more examples and discussion on the selection of $K$ in practice.

\begin{figure}[t]
    \centering
    \minipage{0.25\textwidth}
    \includegraphics[width=\linewidth]{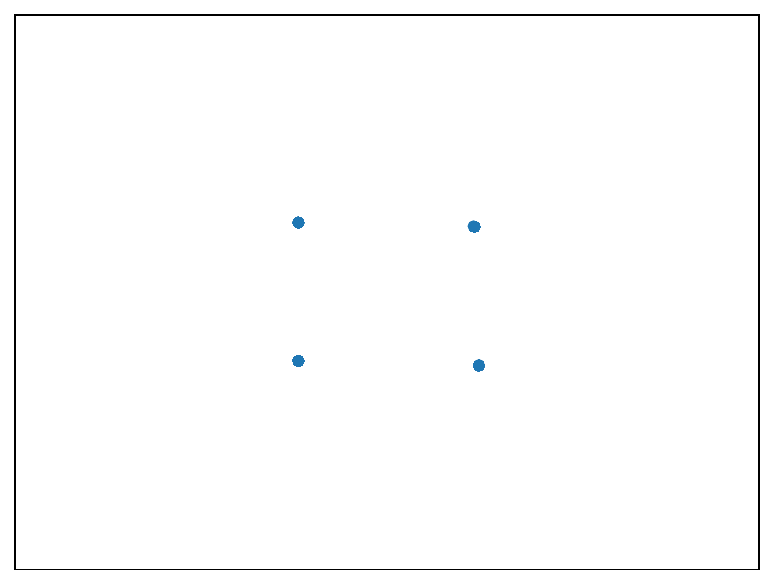}
    \endminipage
    \minipage{0.25\textwidth}
    \includegraphics[width=\linewidth]{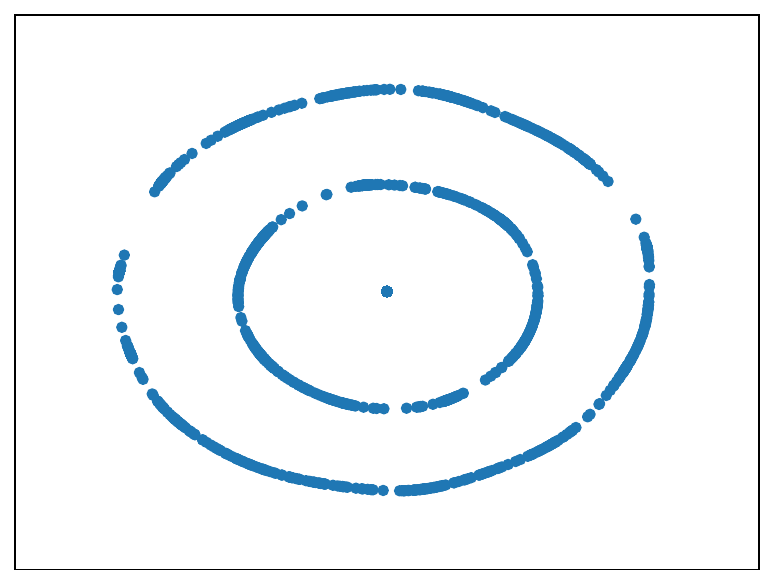}
    \endminipage
    \minipage{0.25\textwidth}
    \includegraphics[width=\linewidth]{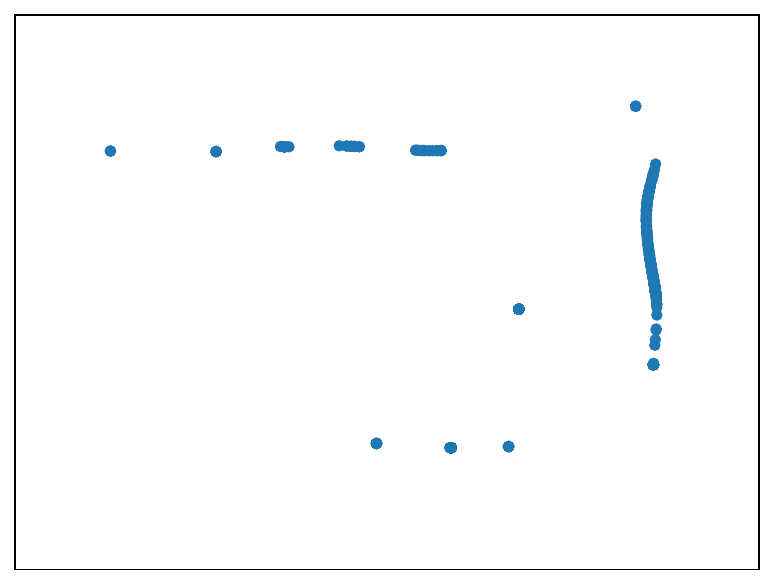}
    \endminipage\hfill 

    \minipage{0.25\textwidth}
    \includegraphics[width=\linewidth]{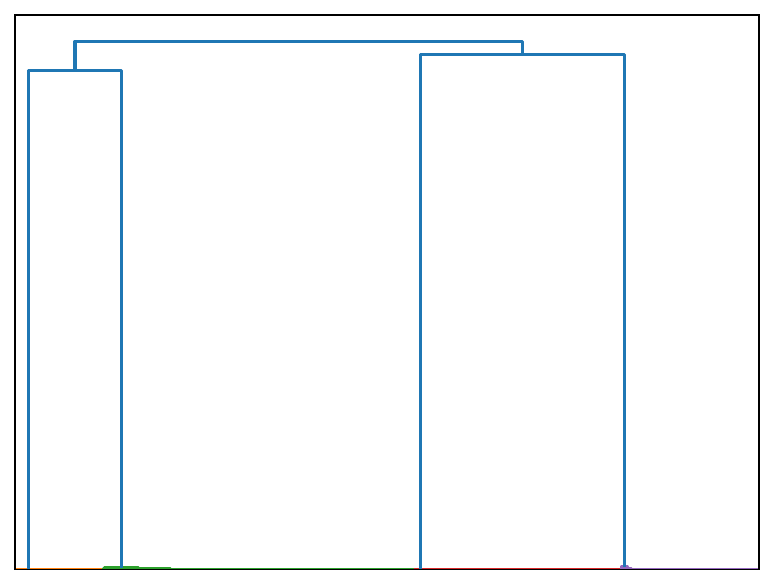}
    \subcaption{$\widehat{K}(2\sigmahat)$=4\\Figure~\ref{fig:four squares}}%
    \endminipage
    \minipage{0.25\textwidth}
    \includegraphics[width=\linewidth]{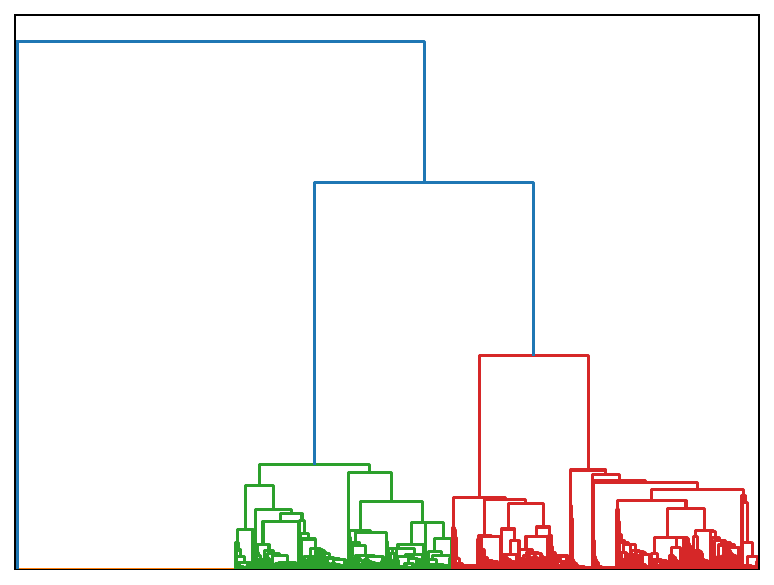}
    \subcaption{$\widehat{K}(2\sigmahat)$=3\\Figure~\ref{fig:three-circles-sample}}%
    \endminipage
    \minipage{0.25\textwidth}
    \includegraphics[width=\linewidth]{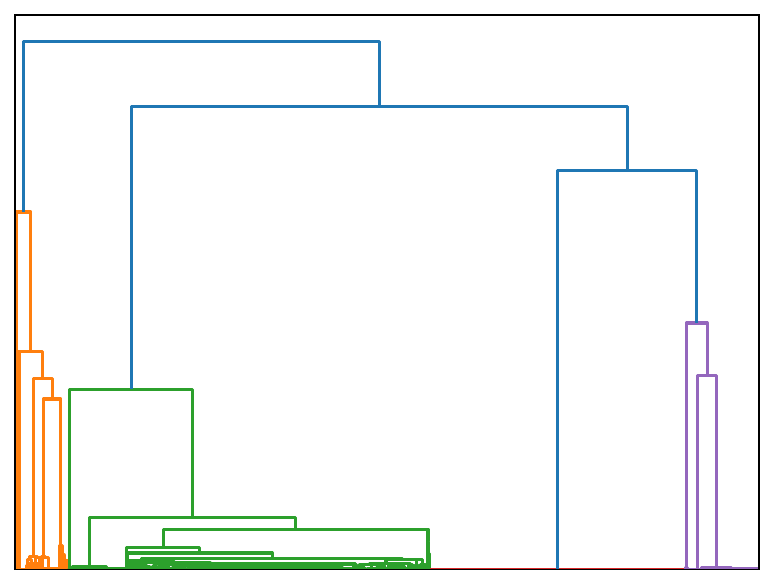}
    \subcaption{$\widehat{K}(2\sigmahat)$=6\\Figure~\ref{fig:benchmark sample:4leaves}}%
    \endminipage\hfill

    \caption{Examples of $\widehat{K}(2\sigmahat)$ from the datasets in Section~\ref{sec:numerical}.
    (top) Support of the selected model $\selected$ and (bottom) its dendrogram for (a) the simulated dataset in Figure~\ref{fig:four squares} (b) the simulated dataset in Figure~\ref{fig:three-circles-sample} and (c) the benchmark dataset in Figure~\ref{fig:benchmark sample:4leaves}. }\label{fig:toy dendro}
\end{figure}

Once $K$ has been selected, we cut the dendrogram to return $K$ clusters.
Denoting the resulting clusters of support atoms as$\{\npmleclusterk\}_{k=1}^K$, we obtain $K$ weighted component densities
\begin{align}\label{eq:BC component}
    \widehat{p}_k = \underbrace{\widehat{G}_{n,2\sigmahat}(\npmleclusterk)}_{\widehat{\lambda}_{n,2\sigmahat,k}} \,\underbrace{\phi_{2\sigmahat}\ast \widehat{G}_{n,2\sigmahat} (\cdot\, |\, \npmleclusterk)}_{\widehat{f}_{n,2\sigmahat,k}}.
\end{align}
The final clustering rule is then defined as the Bayes optimal partition \eqref{eq:bayes optimal partition}.
\noindent
A complete description of our clustering algorithm can be found in Algorithm \ref{algo:ms npmle}. 

A distinctive feature of our approach is that we never assume that $p_0$ is a mixture of $K$ (potentially nonparametric) components: If $p_0$ has such a structure (approximately), it will be revealed by the steps described above as can be clearly seen in the examples above in Figure \ref{fig:toy dendro}.
A more striking illustration of this can be seen in Figure~\ref{fig:simulated dendro}: A hierarchical clustering dendrogram of the raw samples are shown on the left, with no evident clustering structure. 
Indeed, standard clustering metrics would assign a majority of the samples to a single cluster indicated in orange. 
However, after running our procedure, we can clearly see that there are four clusters in the data as indicated in the dendrogram plot of selected model $\selected$ on the right. 
In particular, we can view the atoms of $\selected$ as a ``denoised'' version of the data that more clearly captures the latent structure.

\begin{figure}[t]
    \centering
    \minipage{0.333\textwidth}
    \includegraphics[width=\linewidth]{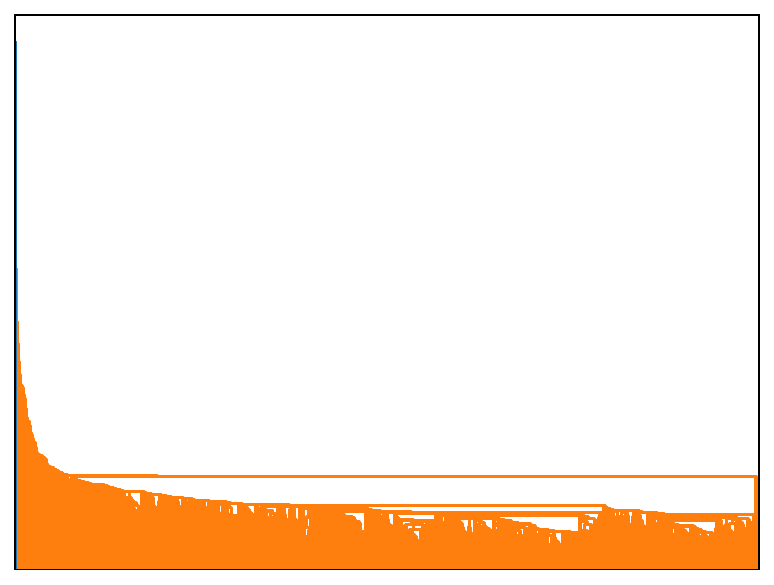}
    \subcaption{}
    \endminipage
    \minipage{0.333\textwidth}
    \includegraphics[width=\linewidth]{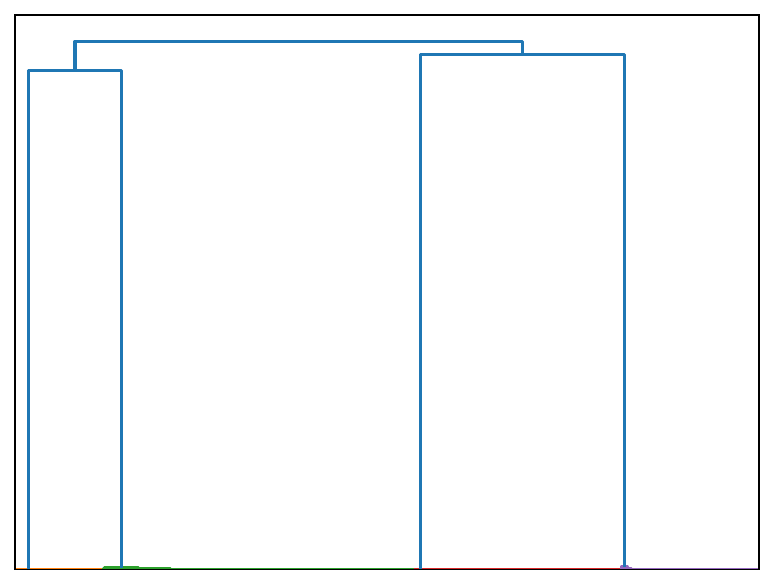}
    \subcaption{}
    \endminipage\hfill

    \caption{Dendrograms of (a) the raw samples and (b) the support of the selected model $\selected$ for the dataset in Figure~\ref{fig:four squares}.}\label{fig:simulated dendro}
\end{figure}

\begin{remark}
It is clear that Algorithm \ref{algo:ms npmle} works for any choice of $K$, whether obtained through $\widehat{K}(2\sigmahat)$, the dendrogram, prior knowledge, or some other means.
\end{remark}

\begin{figure}[t]
    \centering
    \minipage{0.25\textwidth}
    \includegraphics[width=\linewidth]{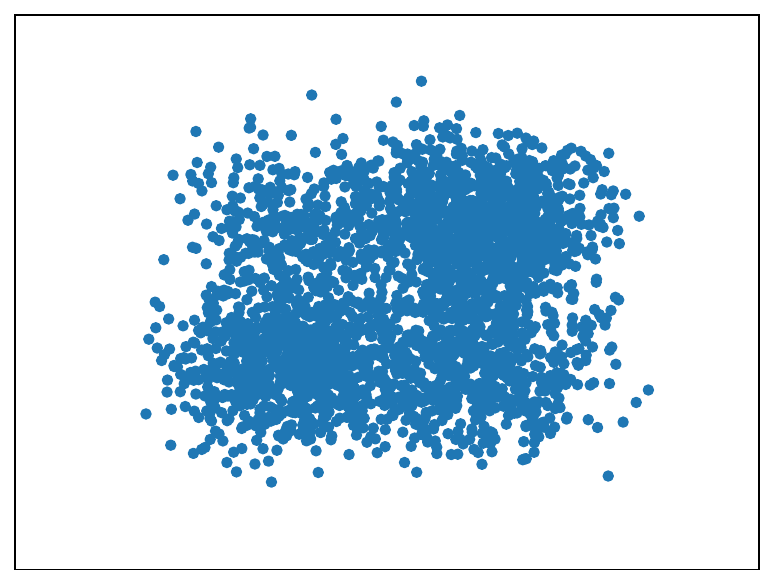}
    \endminipage
    \minipage{0.25\textwidth}
    \includegraphics[width=\linewidth]{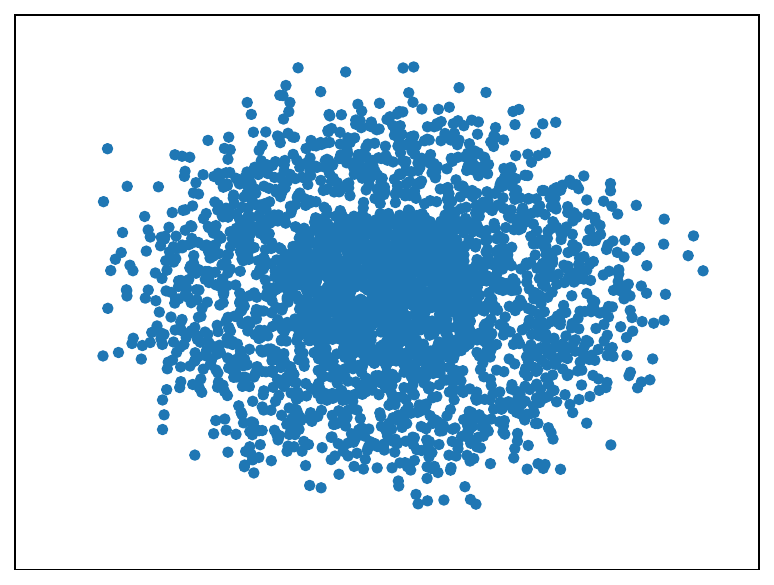}
    \endminipage
    \minipage{0.25\textwidth}
    \includegraphics[width=\linewidth]{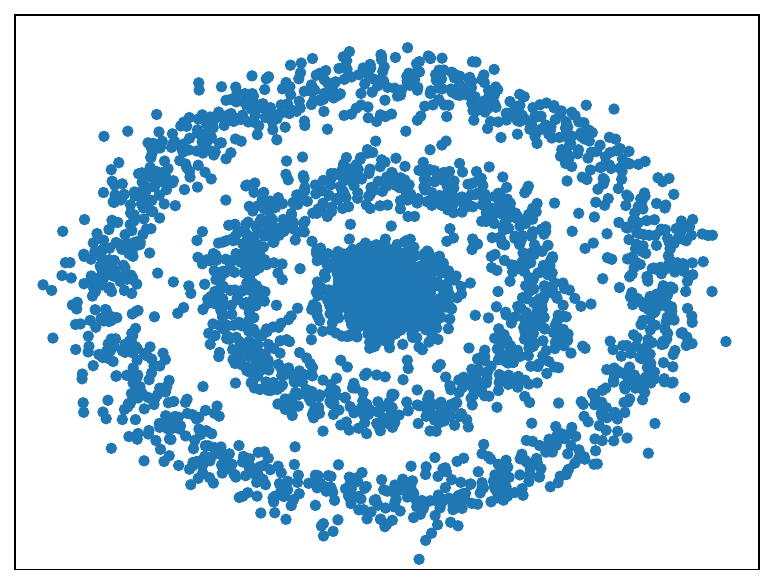}
    \endminipage
    \minipage{0.25\textwidth}
    \includegraphics[width=\linewidth]{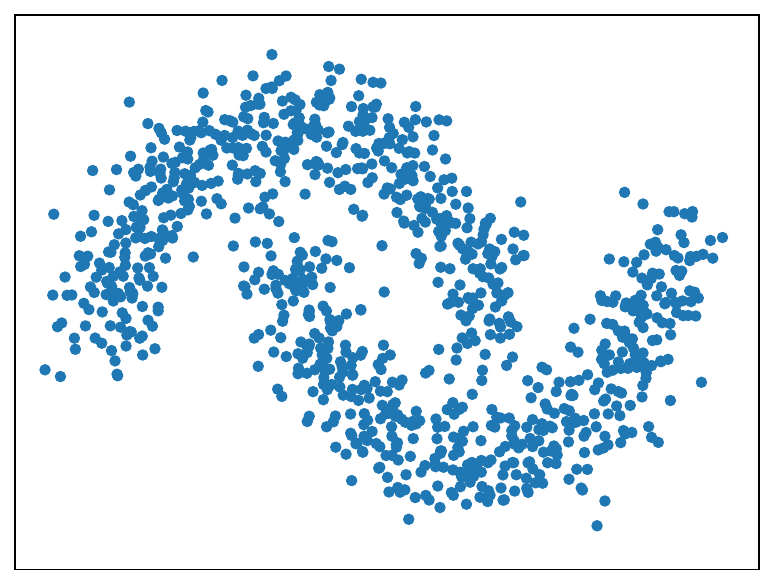}
    \endminipage

    \minipage{0.25\textwidth}
    \includegraphics[width=\linewidth]{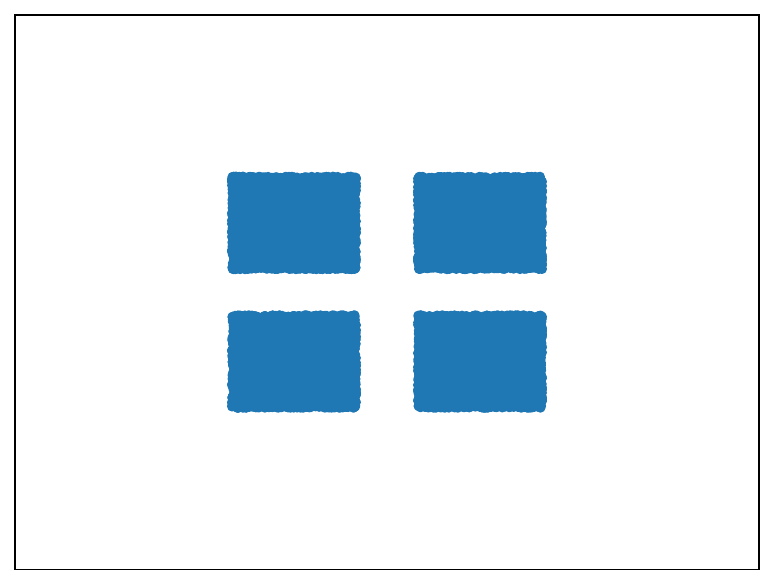}
    \subcaption{Four squares}\label{fig:four squares}
    \endminipage
    \minipage{0.25\textwidth}
    \includegraphics[width=\linewidth]{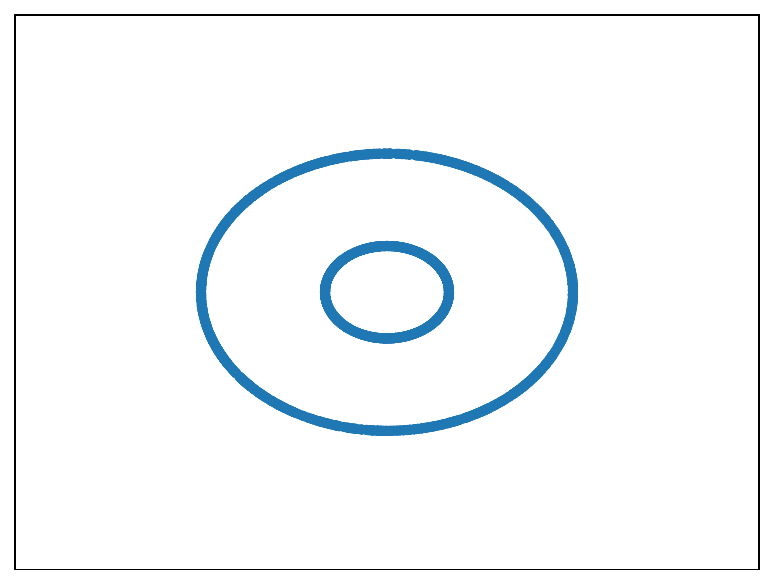}
    \subcaption{Two concentric circles}
    \endminipage
    \minipage{0.25\textwidth}
    \includegraphics[width=\linewidth]{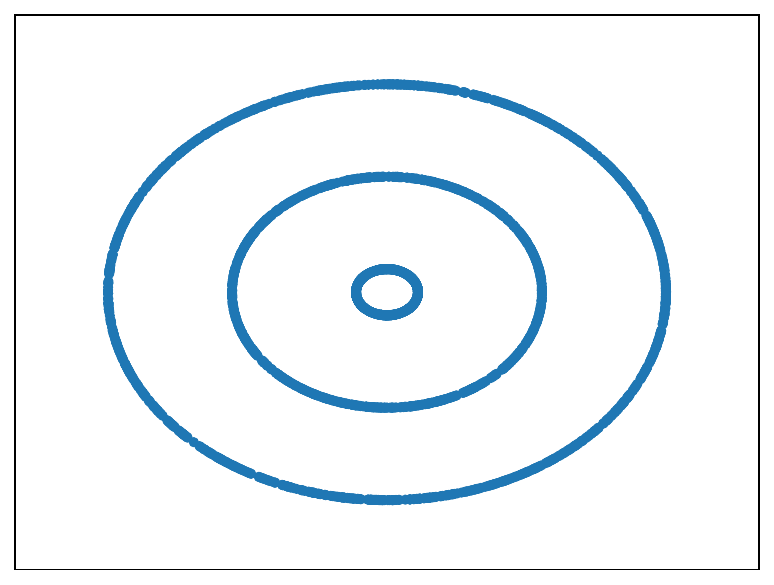}
    \subcaption{Three concentric circles}\label{fig:three-circles-sample}
    \endminipage
    \minipage{0.25\textwidth}
    \includegraphics[width=\linewidth]{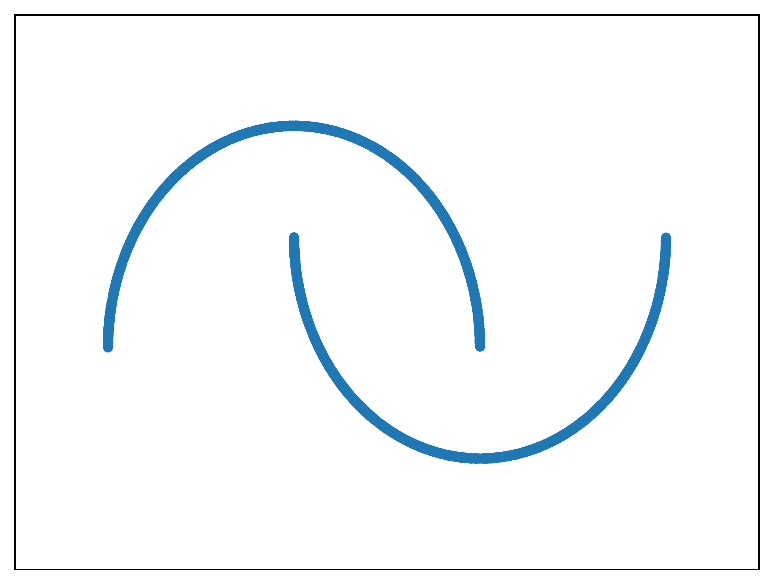}
    \subcaption{Two moons}
    \endminipage
    \caption{Samples (top) and the underlying mixing measure supports (bottom) for the four simulated examples.}
    \label{fig:simulate samples}
\end{figure}

\begin{figure}[t]
\centering
    \minipage{0.25\textwidth}
    \includegraphics[width=\linewidth]{figures_pdf/four-squares-npmle.pdf}
    \endminipage
    \minipage{0.25\textwidth}
    \includegraphics[width=\linewidth]{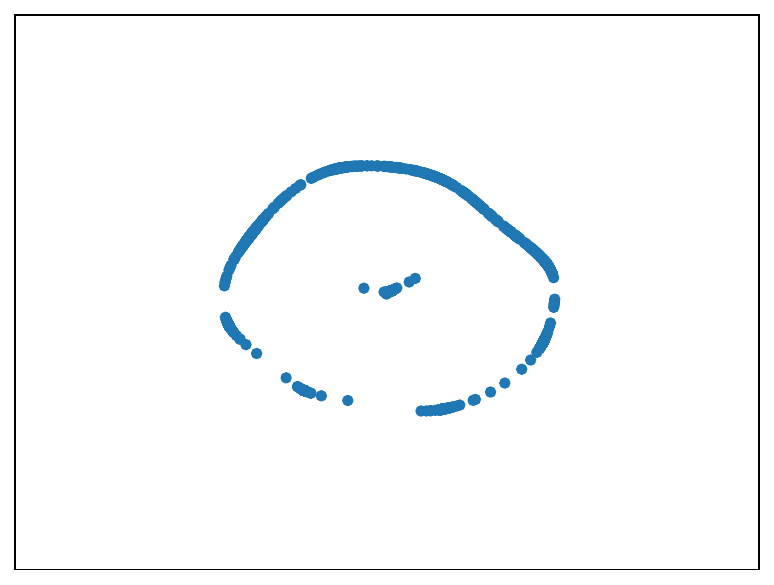}
    \endminipage
    \minipage{0.25\textwidth}
    \includegraphics[width=\linewidth]{figures_pdf/concentric-circles-3-npmle.pdf}
    \endminipage
    \minipage{0.25\textwidth}
    \includegraphics[width=\linewidth]{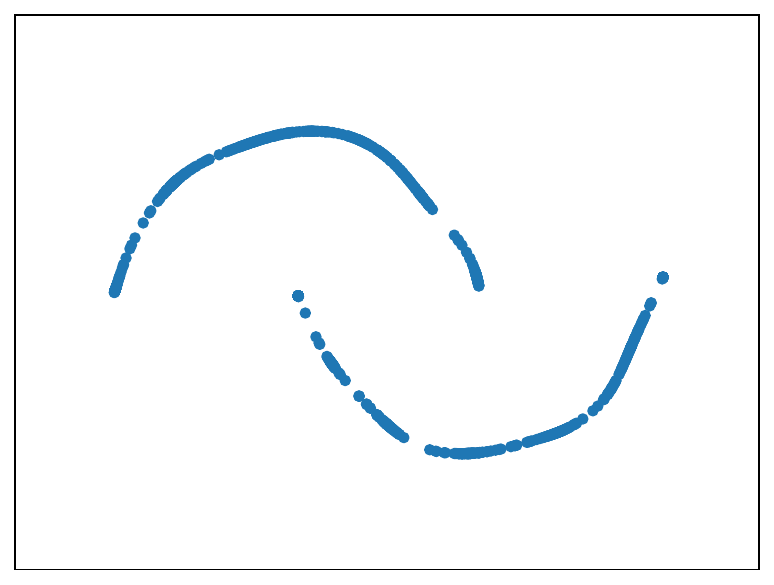}
    \endminipage

    \minipage{0.25\textwidth}
    \includegraphics[width=\linewidth]{figures_pdf/four-squares-npmle-DG.pdf}
    \subcaption{Four squares}\label{fig:four squares DG}
    \endminipage
    \minipage{0.25\textwidth}
    \includegraphics[width=\linewidth]{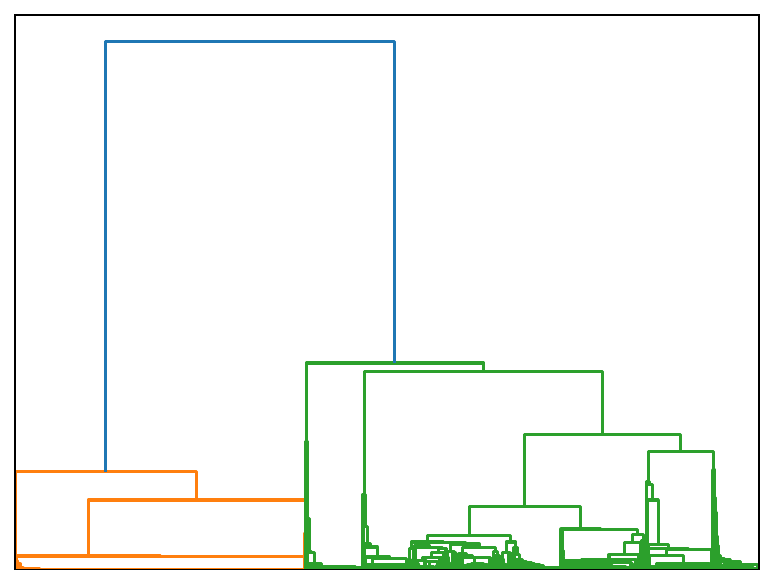}
    \subcaption{Two concentric circles}
    \endminipage
    \minipage{0.25\textwidth}
    \includegraphics[width=\linewidth]{figures_pdf/concentric-circles-3-npmle-DG.pdf}
    \subcaption{Three concentric circles}\label{fig:three circles DG}
    \endminipage
    \minipage{0.25\textwidth}
    \includegraphics[width=\linewidth]{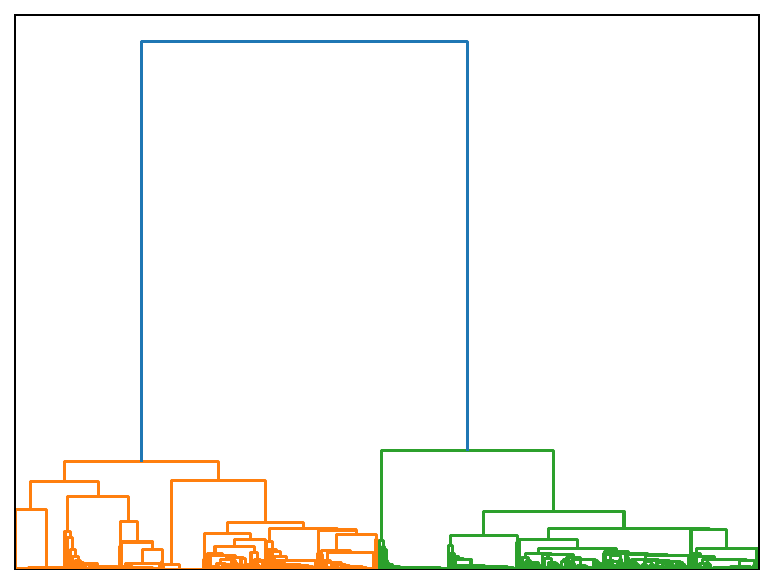}
    \subcaption{Two moons}
    \endminipage
    \caption{The selected NPMLE $\widehat{G}_{n,2\sigmahat}$ (top row) and the associated dendrograms of its support atoms (bottom row).}
    \label{fig:simulated NPMLE best}
\end{figure}

\section{Numerical Experiments}\label{sec:numerical}

We investigate numerically the idea of multiscale representation and demonstrate the wide applicability of our clustering algorithm.
Before presenting the results, let's discuss one missing piece from our estimation procedure, namely how to compute the NPMLE. 
The original definition \eqref{eq:NPMLE} gives an infinite-dimensional optimization problem and is not directly solvable.
Earlier works such as \cite{feng2018approximate} propose to first construct a grid over the parameter space and restrict search to probability measures supported on this grid. 
This has the advantage of reducing \eqref{eq:NPMLE} to a convex problem since only the weights of the probability measure needs to be computed, but suffers from the curse of dimensionality as the grid size would scale exponentially with respect to the dimension. 
Since then, many recent works have been carried out on advancing computational tools for NPMLEs \citep{zhang2024efficient,yan2024learning,yao2024minimizing}.
In this paper, we shall employ the Wasserstein-Fisher-Rao gradient flow algorithm proposed by \cite[Algorithm 1]{yan2024learning}. 
However, we remark this step of computing the NPMLE should be treated as a black-box and any one of the above mentioned algorithms is applicable.

\begin{figure}[t]
\centering
    \minipage{0.25\textwidth}
    \includegraphics[width=\linewidth]{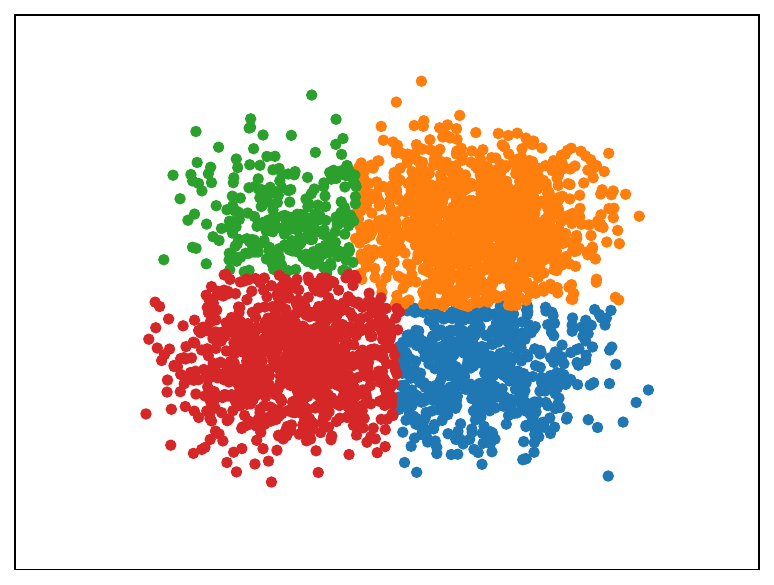}
    \endminipage
    \minipage{0.25\textwidth}
    \includegraphics[width=\linewidth]{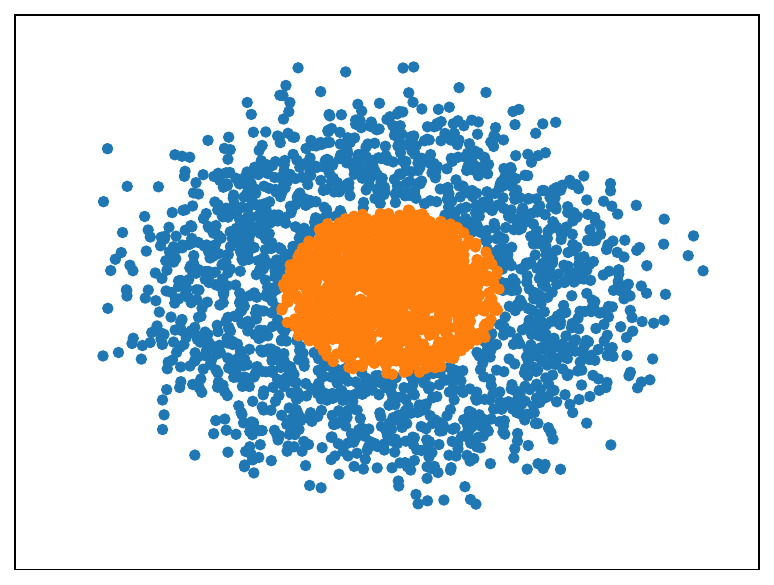}
    \endminipage
    \minipage{0.25\textwidth}
    \includegraphics[width=\linewidth]{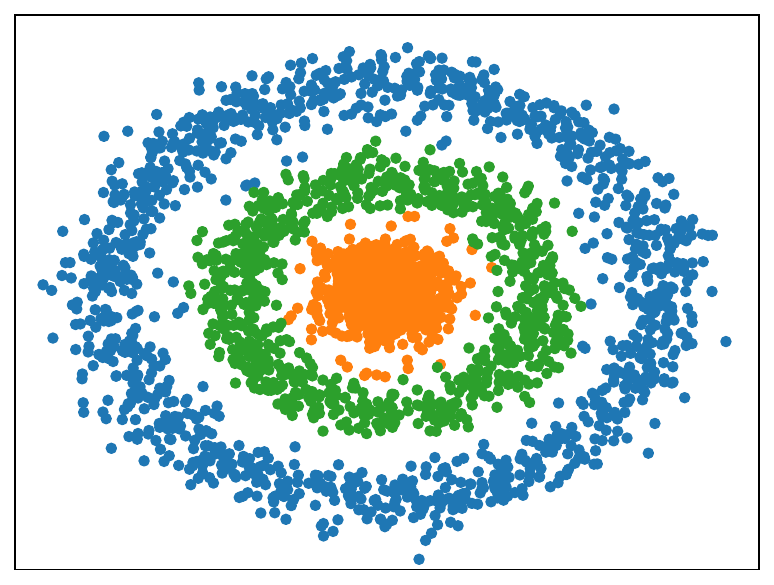}
    \endminipage
    \minipage{0.25\textwidth}
    \includegraphics[width=\linewidth]{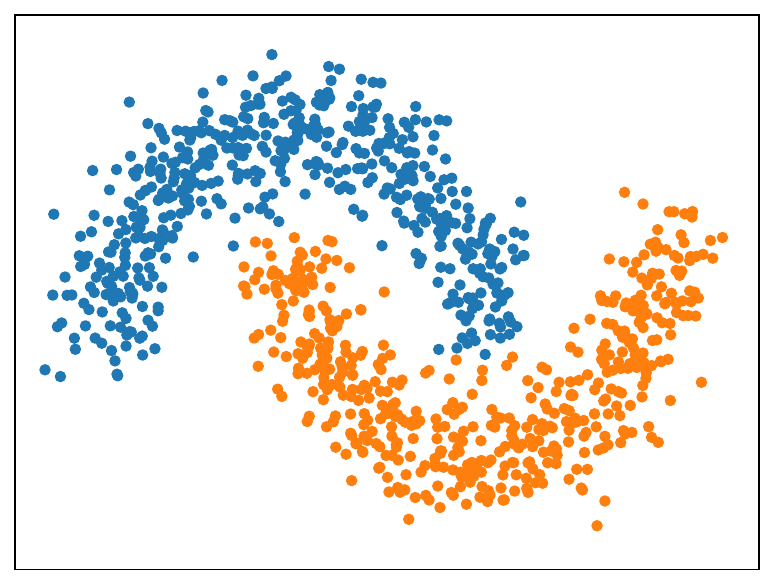}
    \endminipage

    \minipage{0.25\textwidth}
    \includegraphics[width=\linewidth]{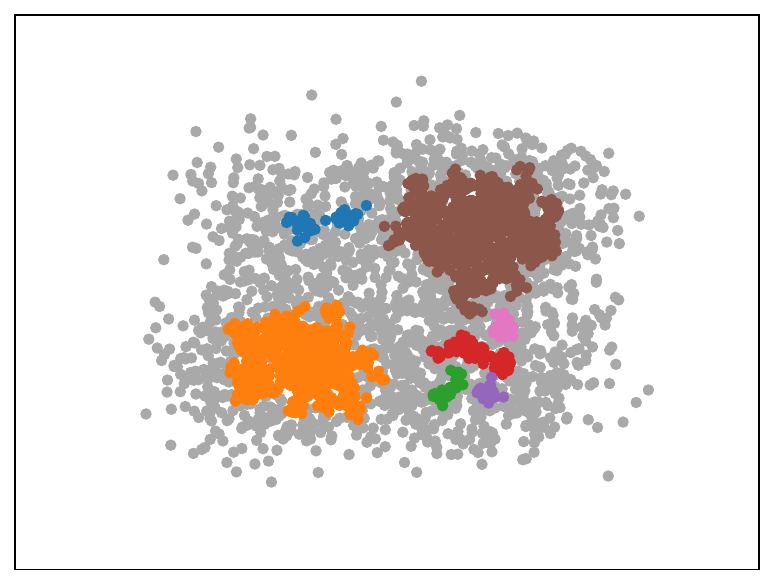}
    \subcaption{Four squares}
    \endminipage
    \minipage{0.25\textwidth}
    \includegraphics[width=\linewidth]{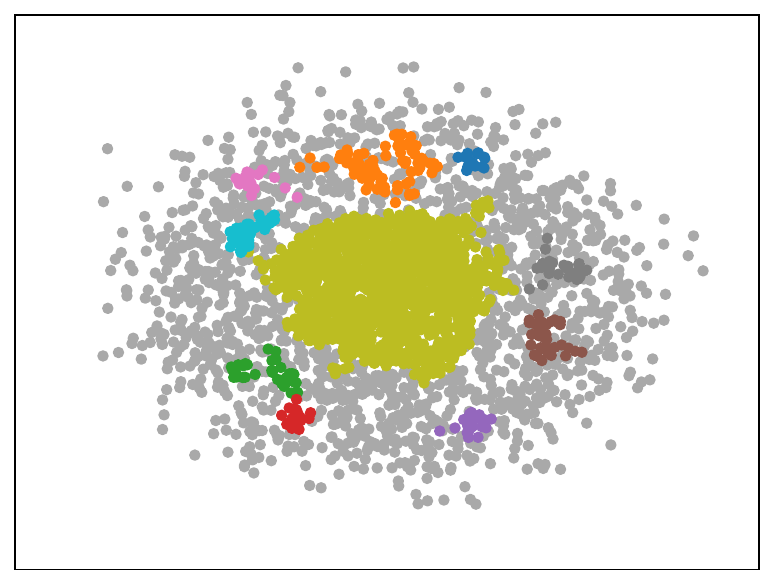}
    \subcaption{Two concentric circles}
    \endminipage
    \minipage{0.25\textwidth}
    \includegraphics[width=\linewidth]{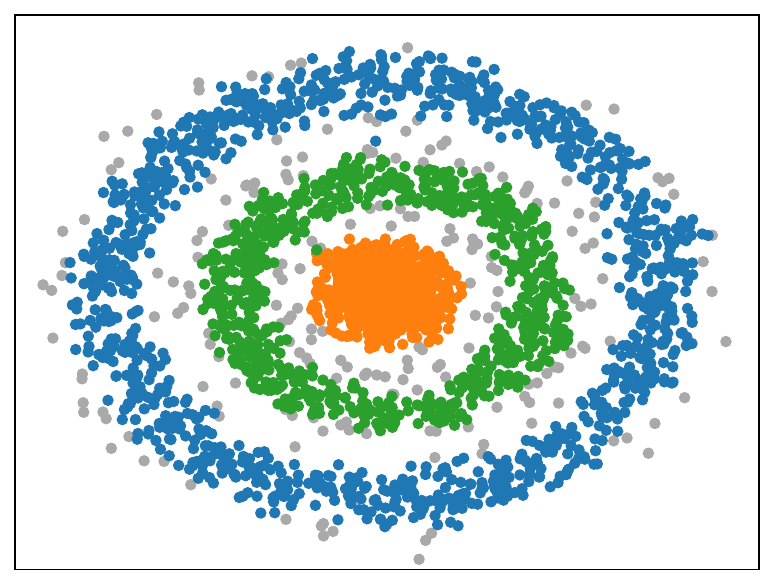}
    \subcaption{Three concentric circles}
    \endminipage
    \minipage{0.25\textwidth}
    \includegraphics[width=\linewidth]{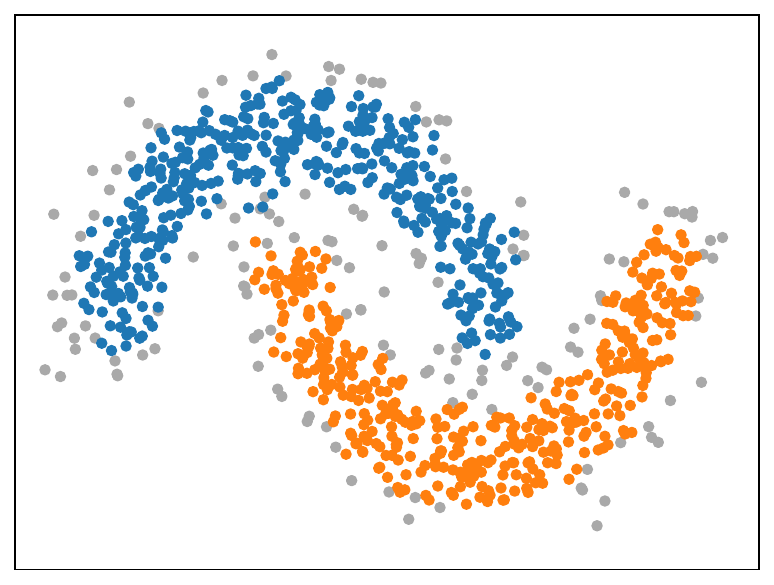}
    \subcaption{Two moons}
    \endminipage
    \caption{Clustering of the data points given by Algorithm \ref{algo:ms npmle} (top row) and the HDBSCAN algorithm (bottom row).}
    \label{fig:simulated clustering}
\end{figure}

\begin{figure}[t]
\centering
    \minipage{0.25\textwidth}
    \includegraphics[width=\linewidth]{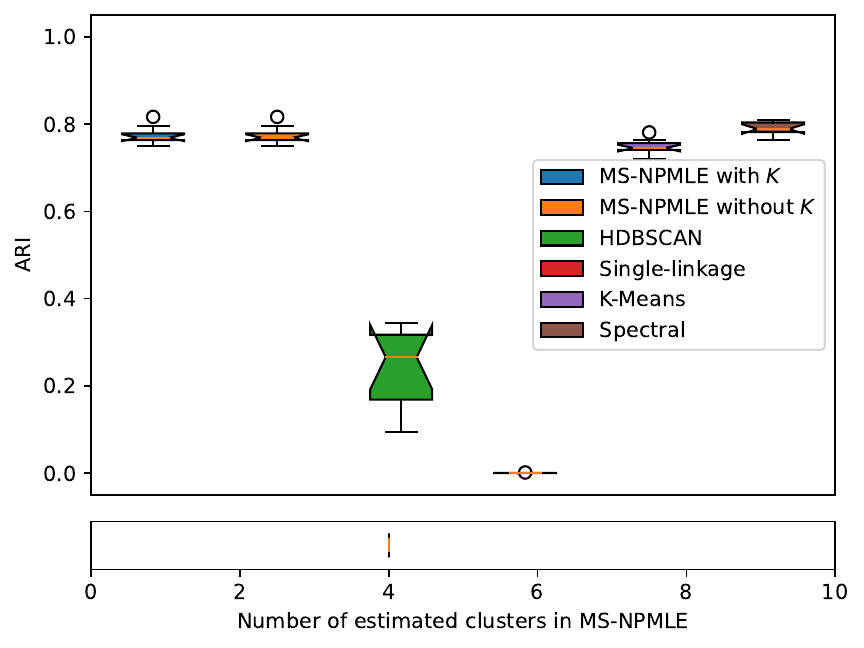} 
    \subcaption{Four squares}
    \endminipage
    \minipage{0.25\textwidth}
    \includegraphics[width=\linewidth]{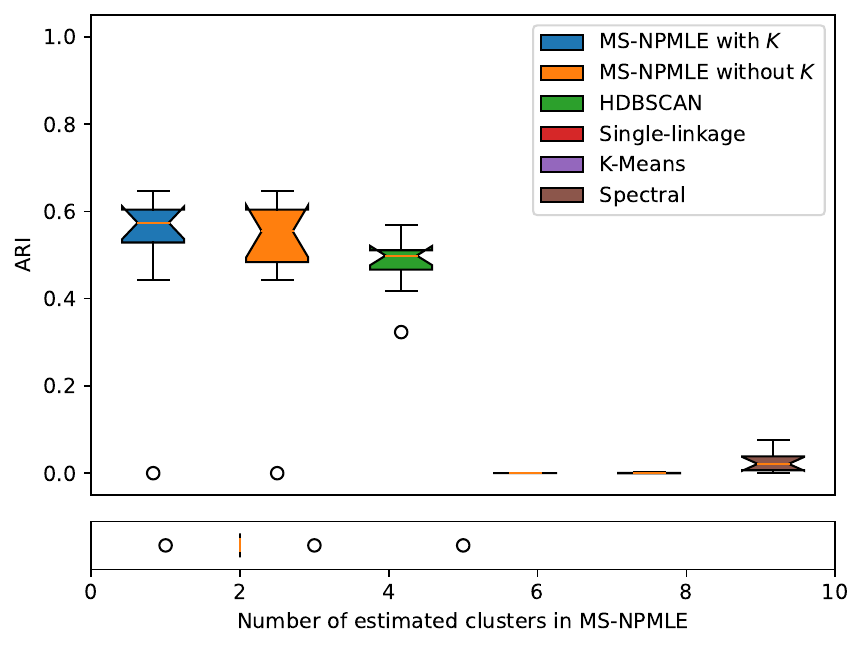}
    \subcaption{Two concentric circles}
    \endminipage
    \minipage{0.25\textwidth}
    \includegraphics[width=\linewidth]{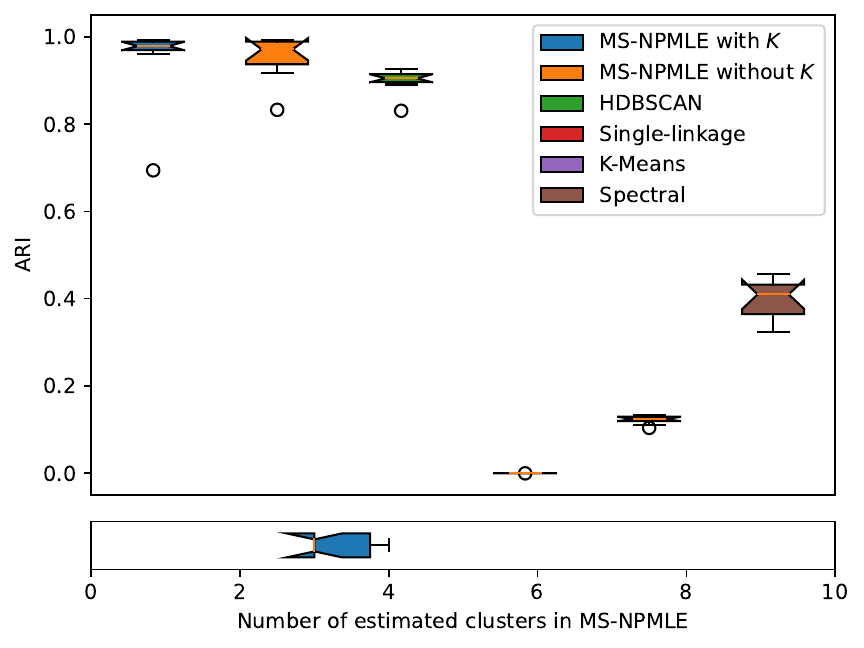}
    \subcaption{Three concentric circles}
    \endminipage
    \minipage{0.25\textwidth}
    \includegraphics[width=\linewidth]{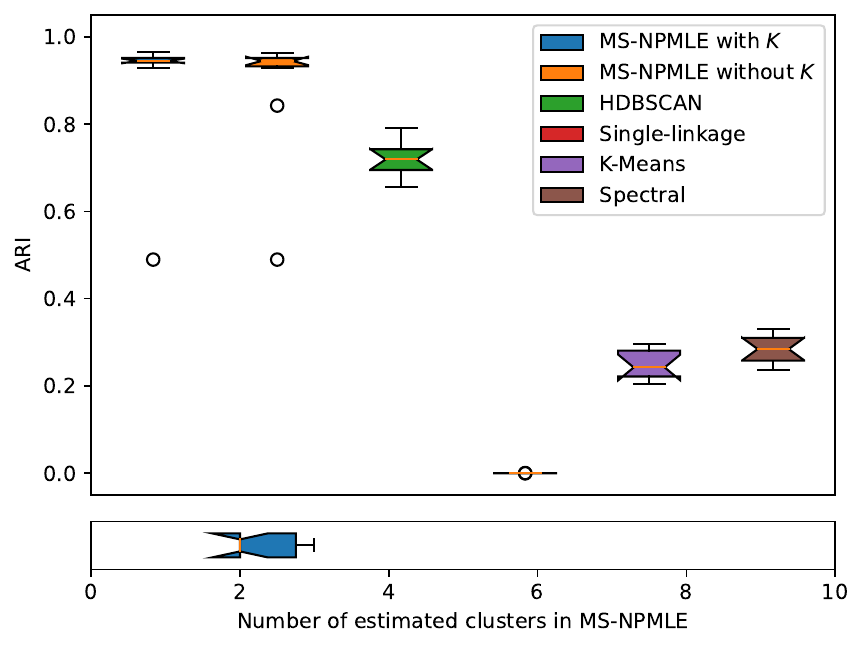}
    \subcaption{Two moons}
    \endminipage
    \caption{Comparison of accuracy for the four simulated examples. The experiments are repeated 10 times. Below each of the ARI comparison is the boxplot for the estimated $K$ obtained by examining the dendrogram of the selected NPMLE.}
    \label{fig:simulated accuracy}
\end{figure}

\subsection{Simulation Studies}\label{sec:numerical:sim}
To start with, we consider four simulated examples where the density $p_0$ indeed comes from the family of convolutional Gaussian mixtures $p_0=\phi_\sigma \ast G_0$. These simulations are appealing since there is an easily computed (i.e. known) ground truth to compare with. 
In subsequent sections, we consider benchmarks where the underlying structure is implicit and not simulated from a particular model.

We consider the following four examples:
\begin{enumerate}
    \item Four squares: $G_0=0.4\cdot\unif\big((1.5,1.5)+\square_1\big)+0.3\cdot\unif\big((-1.5,-1.5)+\square_1\big)+0.2\cdot\unif\big((1.5,-1.5)+\square_1\big)+0.1\cdot\unif\big((-1.5,1.5)+\square_1\big)$, where $\square_1=[-1,1]\times [-1,1]$ is a square. 

    \item Two concentric circles: $G_0=0.5\cdot\unif\big(\circle_1\big)+0.5\cdot \unif\big(3\circle_1\big)$, where $\circle_1=(\cos(\theta),\sin(\theta))$, $\theta\in[0,2\pi]$ is the unit circle. 
        \item Three concentric circles: $G_0=0.3\cdot\unif\big(\circle_1\big)+0.3\cdot \unif\big(5\circle_1\big)+0.4\cdot\unif\big(9\circle_1\big)$, where $\circle_1=(\cos(\theta),\sin(\theta))$, $\theta\in[0,2\pi]$ is the unit circle. 
        
    \item Two moons: $G_0=0.3\cdot \unif\big(\arc_1\big)+0.6\cdot \unif\big((1,0.5)-\arc_1\big)$, where $\arc_1=(\cos(\theta),\sin(\theta))$, $\theta\in[0,\pi]$ is a semi-circle arc. 
\end{enumerate}
Figure \ref{fig:simulate samples} shows the support of these mixing measures as well as samples from them.

We shall apply Algorithm \ref{algo:ms npmle} to these four datasets and compare its performance with standard clustering algorithms.
Firstly, we shall examine the clusters returned by Algorithm \ref{algo:ms npmle} without specifying the true number of clusters $K$ and compare them with those found by the HDBSCAN algorithm \citep{campello2013density} since the latter does not require knowledge of $K$ either.
We recall that our estimate for $K$ is based on the dendrograms of the selected NPMLEs, which are shown in Figure \ref{fig:simulated NPMLE best}.

We can see that each one of them has a clear clustering structure as suggested by the dendrograms, based on which we set the number of clusters to be 4, 2, 3, 2 respectively. 
Figure \ref{fig:simulated clustering} (top row) shows the resulting clustering of the samples computed as in Algorithm \ref{algo:ms npmle}, compared with those returned by the HDBSCAN algorithm (bottom row). 
The gray points in the figures of the bottom row are ``noise points'' labeled by the HDBSCAN algorithm.
We can see that our method (Algorithm \ref{algo:ms npmle}) successfully captures the latent clustering structures in all cases, whereas HDBSCAN does a poorer job for the first two datasets, where the cluster structures are less clear from the raw samples.

Next, to get a quantitative sense of the performance of Algorithm \ref{algo:ms npmle}, we shall compare the adjusted Rand index (ARI) with other standard algorithms such as $k$-means, single-linkage clustering, spectral clustering. 
We repeat the experiments 10 times by re-generating the samples and results are shown in Figure \ref{fig:simulated accuracy}.
\revise{Here and below, we shall denote our proposed approach as MS-NPMLE (which stands for multiscale NPMLE) and we present its performance in both cases when the true number of clusters $K$ is either provided or not. All the other methods except HDBSCAN are supplemented with the true $K$.}
We see that our method provides uniformly good performance across all examples, even when $K$ is not provided, whereas the other methods give a poor clustering in multiple instances. 
This suggests that our approach is able to handle a wider range of geometric structures in $p_0$.

\begin{figure}[t]
\centering
    \minipage{0.25\textwidth}
    \includegraphics[width=\linewidth]{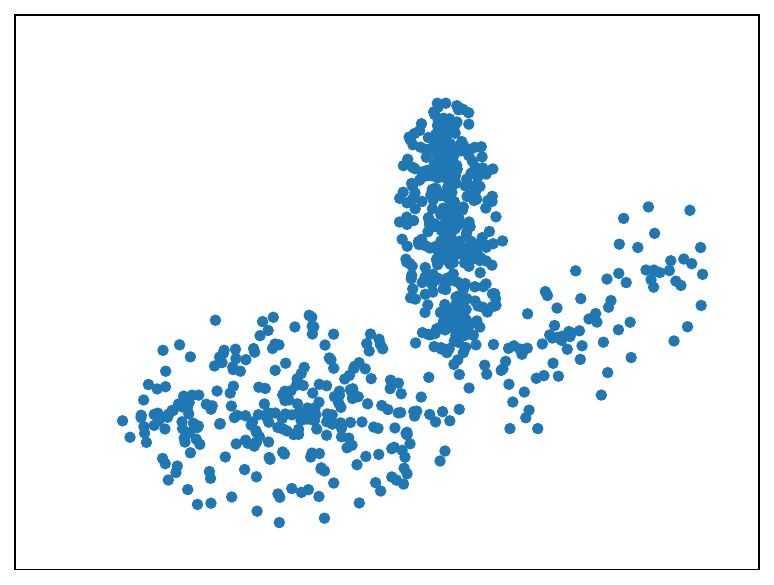} 
    \subcaption{Three leaves}
    \endminipage
    \minipage{0.25\textwidth}
    \includegraphics[width=\linewidth]{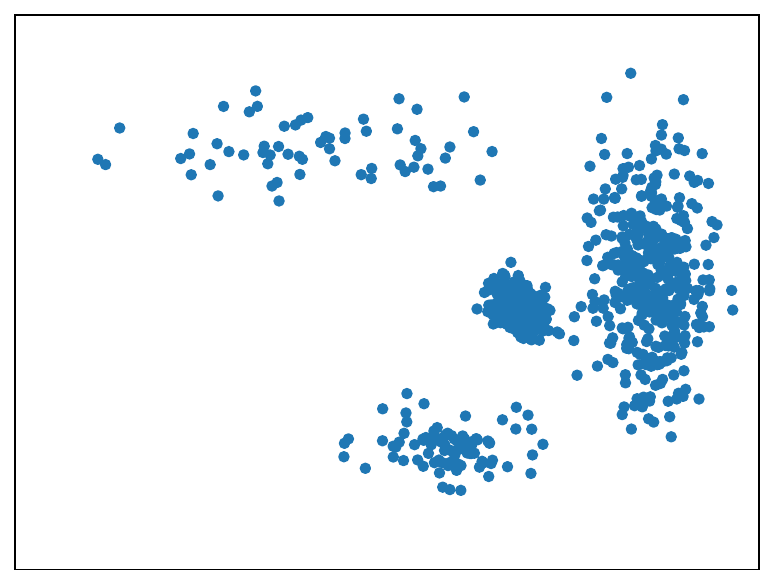}
    \subcaption{Four leaves}
    \label{fig:benchmark sample:4leaves}
    \endminipage
    \minipage{0.25\textwidth}
    \includegraphics[width=\linewidth]{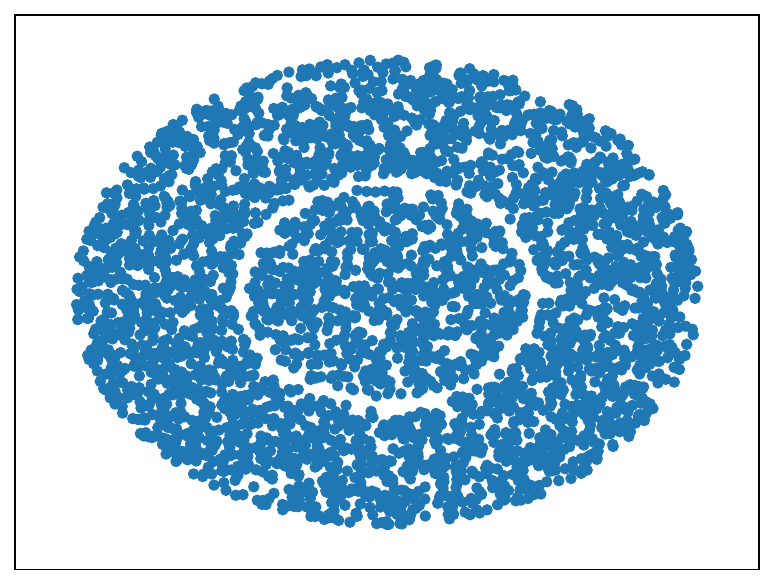}
    \subcaption{Two discs}
    \endminipage
    \minipage{0.25\textwidth}
    \includegraphics[width=\linewidth]{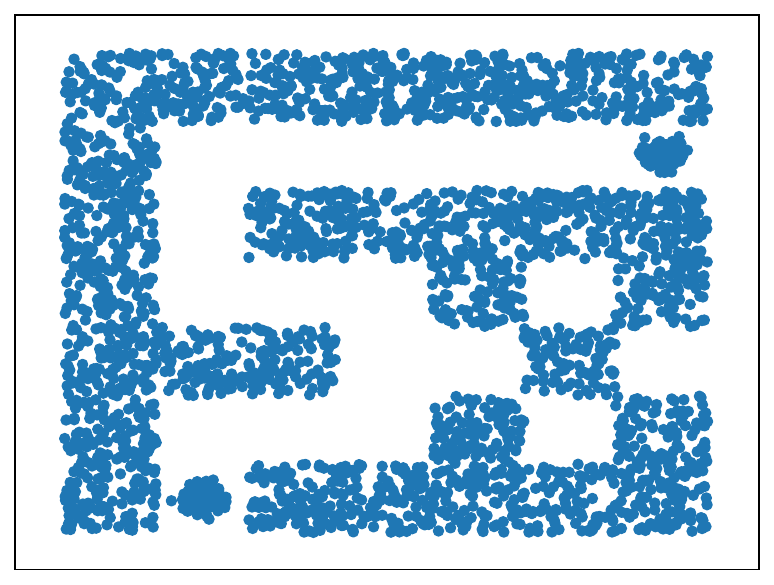}
    \subcaption{Labyrinth}
    \endminipage
    \caption{Samples from the four benchmark examples. }
    \label{fig:benchmark sample}
\end{figure}

\begin{figure}[t]
\centering
    \minipage{0.25\textwidth}
    \includegraphics[width=\linewidth]{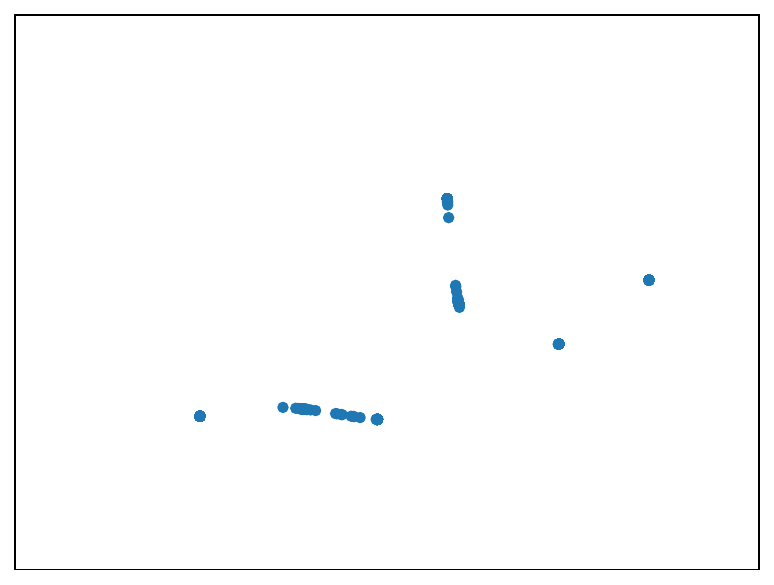}
    \endminipage
    \minipage{0.25\textwidth}
    \includegraphics[width=\linewidth]{figures_pdf/four-leaves-npmle.pdf}
    \endminipage
    \minipage{0.25\textwidth}
    \includegraphics[width=\linewidth]{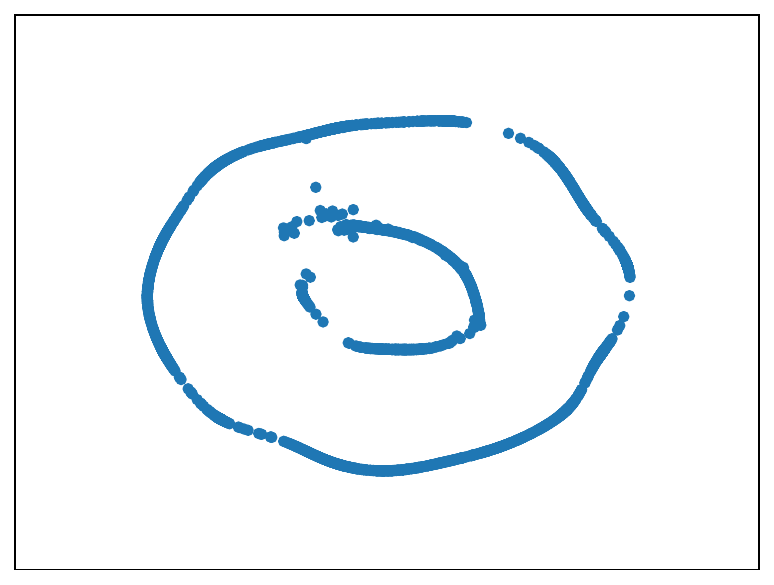}
    \endminipage
    \minipage{0.25\textwidth}
    \includegraphics[width=\linewidth]{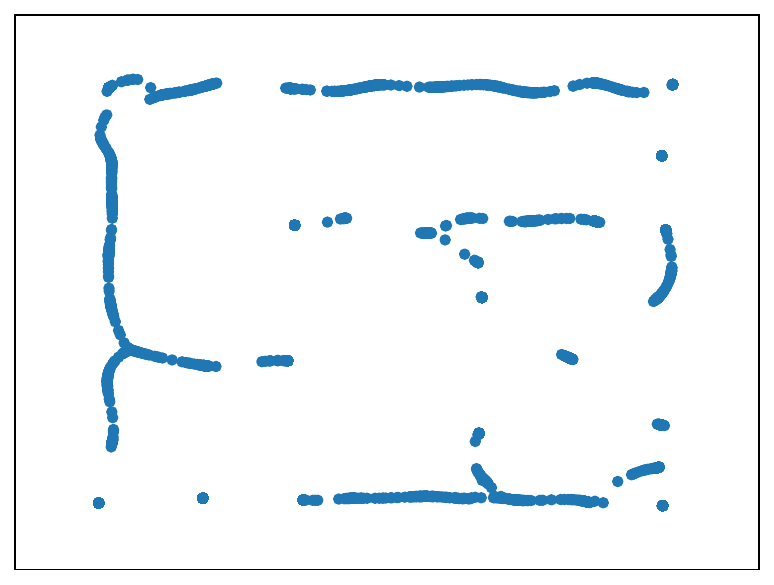}
    \endminipage

     \minipage{0.25\textwidth}
    \includegraphics[width=\linewidth]{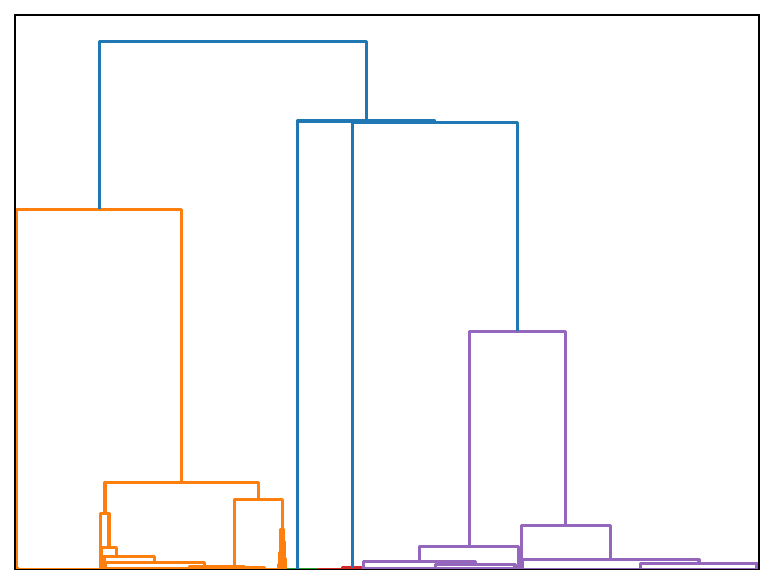}
    \subcaption{Three leaves}
    \endminipage
    \minipage{0.25\textwidth}
    \includegraphics[width=\linewidth]{figures_pdf/four-leaves-npmle-DG.pdf}
    \subcaption{Four leaves}
    \label{fig:benchmark NPMLE best:4leaves}
    \endminipage
    \minipage{0.25\textwidth}
    \includegraphics[width=\linewidth]{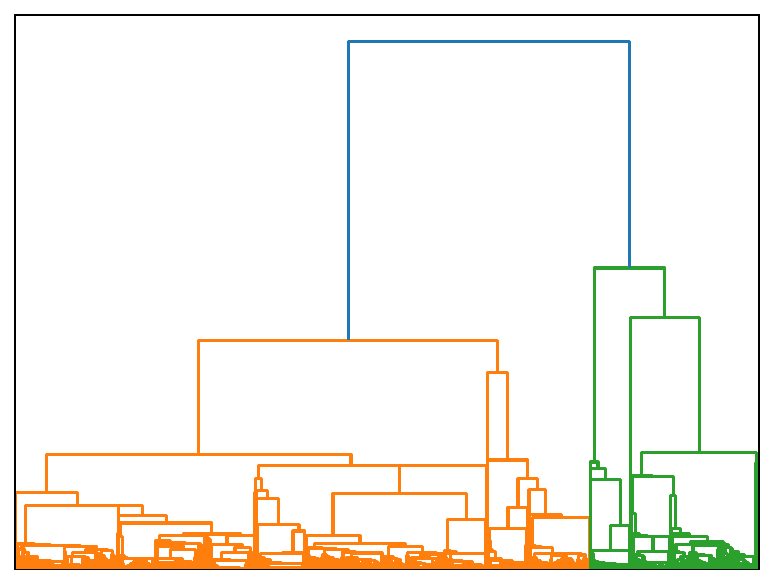}
    \subcaption{Two discs}
    \endminipage
    \minipage{0.25\textwidth}
    \includegraphics[width=\linewidth]{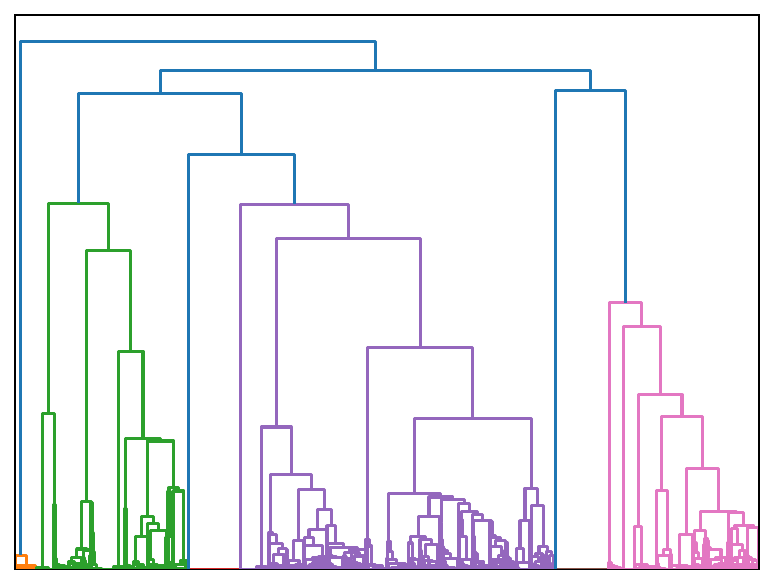}
    \subcaption{Labyrinth}
    \endminipage
    \caption{The selected NPMLE $\widehat{G}_{n,2\sigmahat}$ (top row) and the associated dendrograms of its support atoms (bottom row).}
    \label{fig:benchmark NPMLE best}
\end{figure}

\subsection{Benchmark datasets}\label{sec:numerical:benchmark}

After testing our algorithm on simulated datasets with known ground truth, we proceed to consider several benchmark datasets. 
Figure \ref{fig:benchmark sample} shows four of the most difficult benchmark datasets taken from \cite{benchmark_dataset1} and \cite{Gagolewski_A_framework_for_2022}.

Again, we shall apply Algorithm \ref{algo:ms npmle} and compare its performance with the other standard clustering algorithms. 
Figure \ref{fig:benchmark NPMLE best} shows the selected NPMLEs and the associated dendrograms. 
We see that there is again clear clustering structures in each of them, and by inspecting the dendrograms, we shall set the number of clusters to be 4, 4, 2, 6 respectively, where the last one comes from the three colored branches and the three single-leaf ones.

Figure \ref{fig:benchmark clustering} then compares the resulting clustering of the datasets with those found by HDBSCAN.
A notable observation is that our approach again captures the major clusters within the datasets, and outperforms HDBSCAN on the Labyrinth example. 
The only exception is in the three leaves example, where both methods struggle. Our method splits rightmost cluster into two, which is likely due to the imbalanced weights of the three clusters, where only two isolated atoms remain in our selected NPMLE for representing the rightmost cluster. 
A similar issue is present in HDBSCAN, where the rightmost portion of the samples are labeled as noise and not clustered correctly. 

Finally, we give a quantitative measure of the clustering accuracy by comparing further with $k$-means, single-linkage clustering and spectral clustering in Figure \ref{fig:benchmark accuracy}. 
As before, we repeat the experiment 10 times by subsampling 90\% of the original datasets. 
As shown in Figure \ref{fig:benchmark accuracy}, our proposed method achieves uniformly good performance and outperforms the other algorithms for the latter two examples.

\begin{figure}[t]
\centering
    \minipage{0.25\textwidth}
    \includegraphics[width=\linewidth]{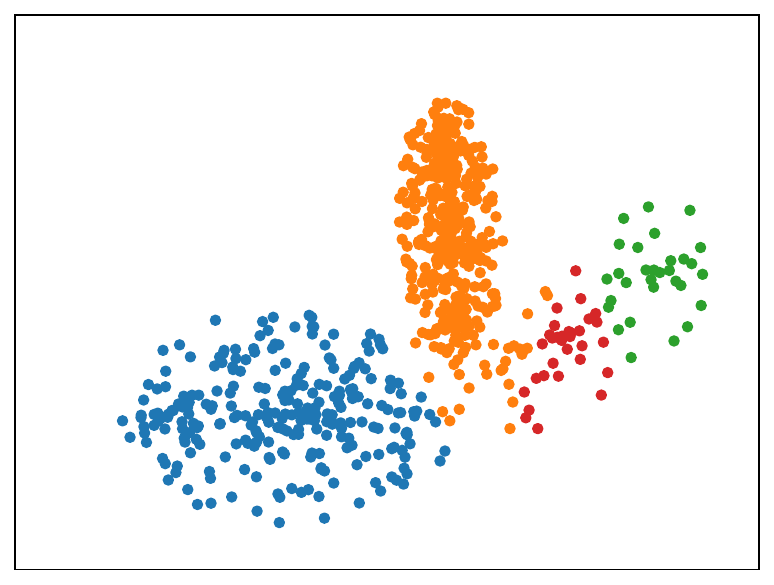}
    \endminipage
    \minipage{0.25\textwidth}
    \includegraphics[width=\linewidth]{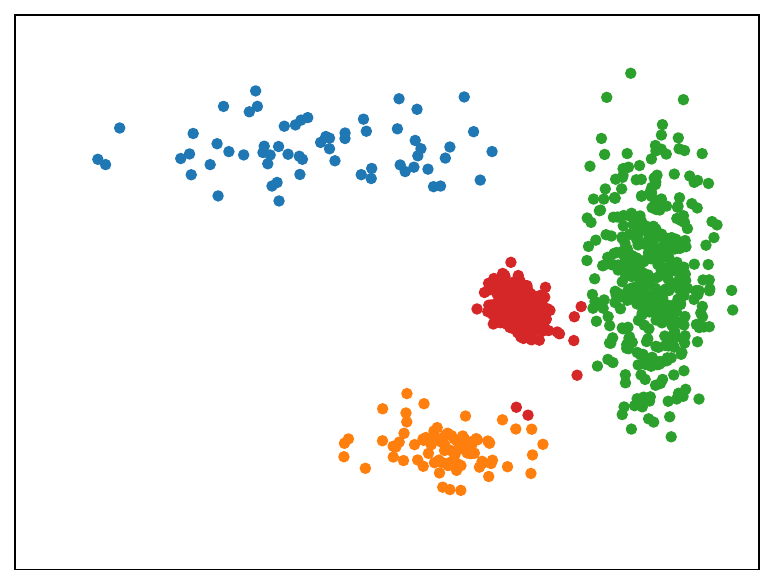}
    \endminipage
    \minipage{0.25\textwidth}
    \includegraphics[width=\linewidth]{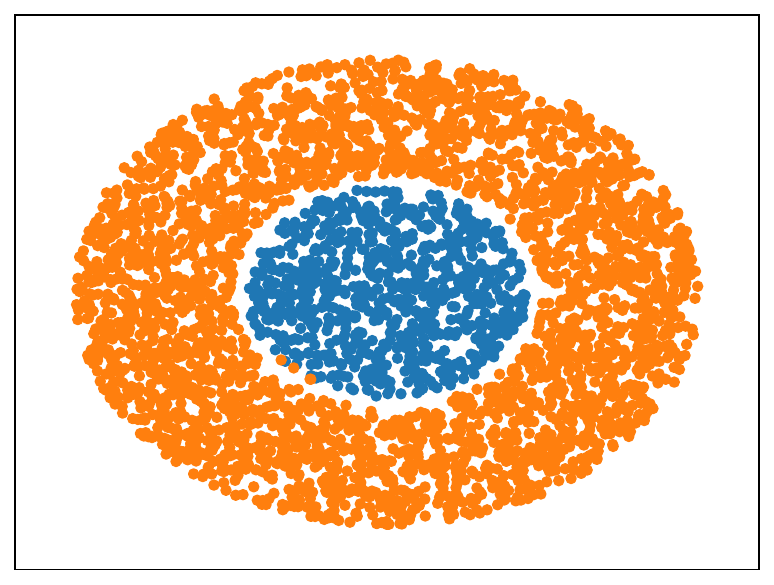}
    \endminipage
    \minipage{0.25\textwidth}
    \includegraphics[width=\linewidth]{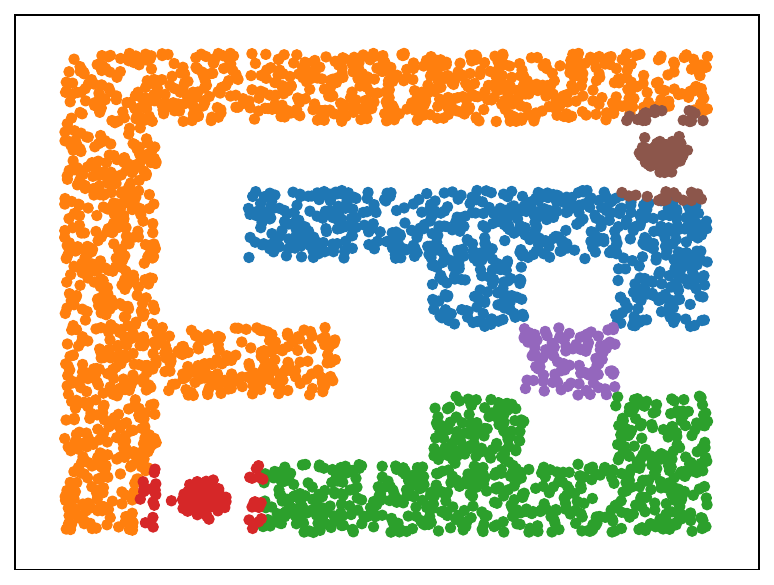}
    \endminipage

     \minipage{0.25\textwidth}
    \includegraphics[width=\linewidth]{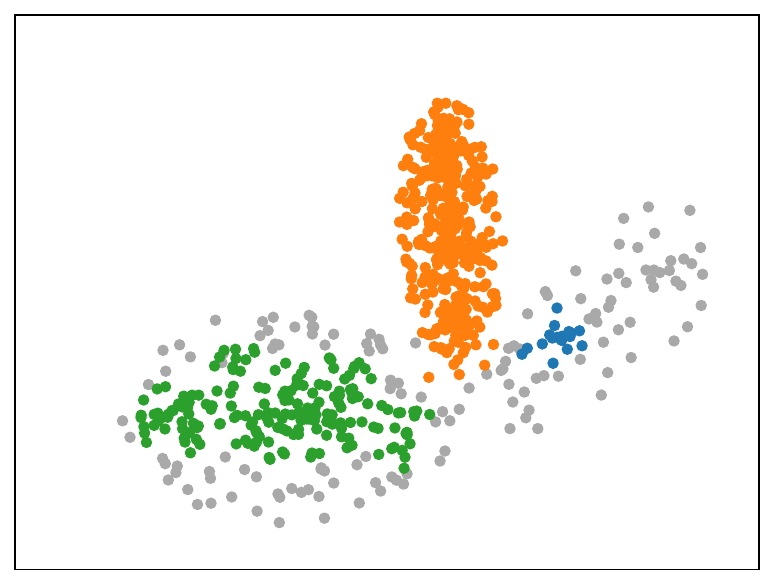}
    \subcaption{Three leaves}
    \endminipage
    \minipage{0.25\textwidth}
    \includegraphics[width=\linewidth]{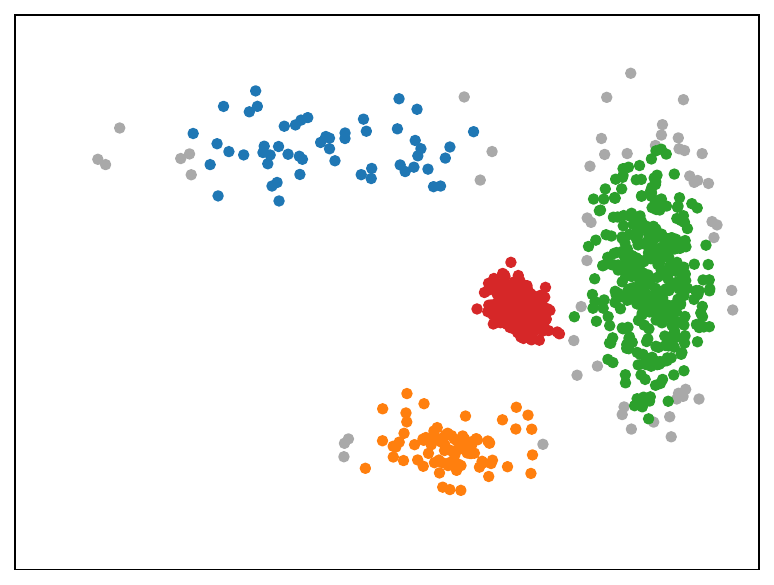}
    \subcaption{Four leaves}
    \endminipage
    \minipage{0.25\textwidth}
    \includegraphics[width=\linewidth]{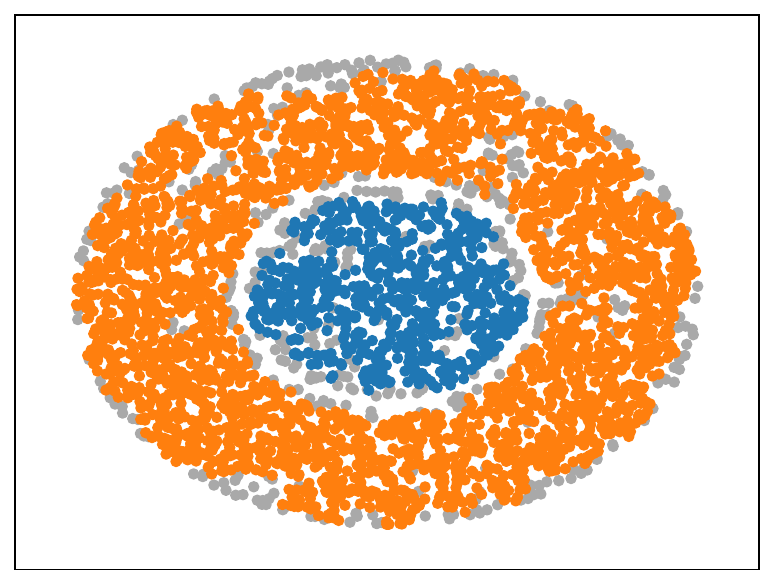}
    \subcaption{Two discs}
    \endminipage
    \minipage{0.25\textwidth}
    \includegraphics[width=\linewidth]{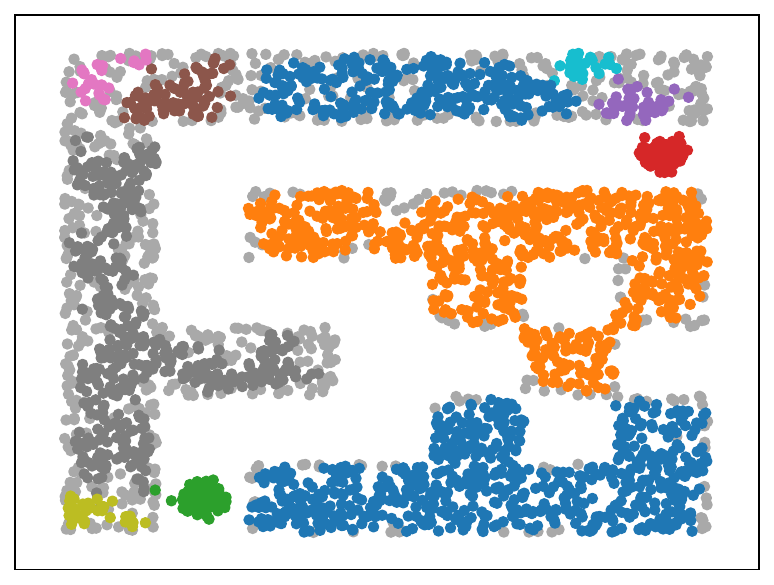}
    \subcaption{Labyrinth}
    \endminipage
    \caption{Clustering of the data points given by Algorithm \ref{algo:ms npmle} (top row) and the HDBSCAN algorithm (bottom row).}
    \label{fig:benchmark clustering}
\end{figure}

\begin{figure}[t]
\centering
    \minipage{0.25\textwidth}
    \includegraphics[width=\linewidth]{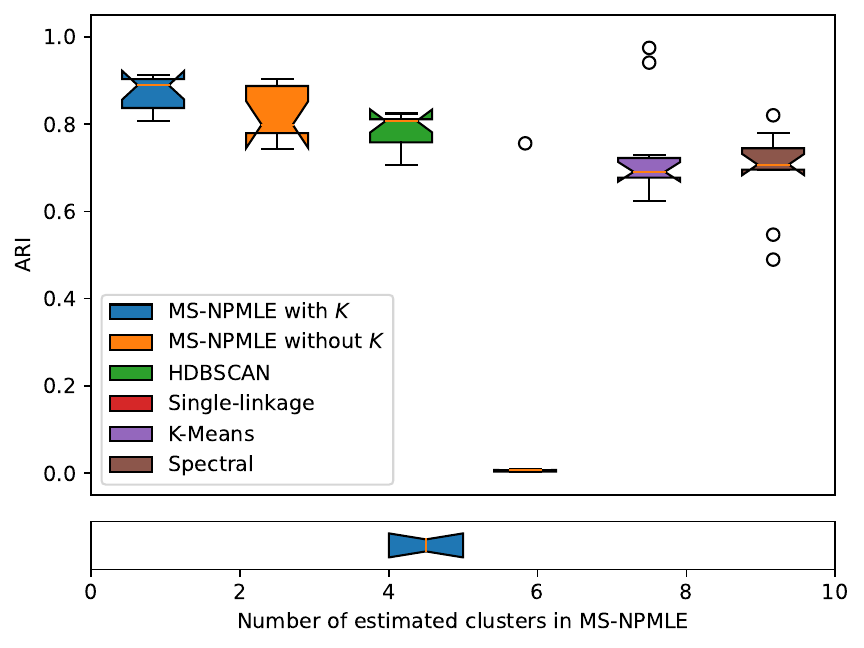} 
    \subcaption{Three leaves}
    \endminipage
    \minipage{0.25\textwidth}
    \includegraphics[width=\linewidth]{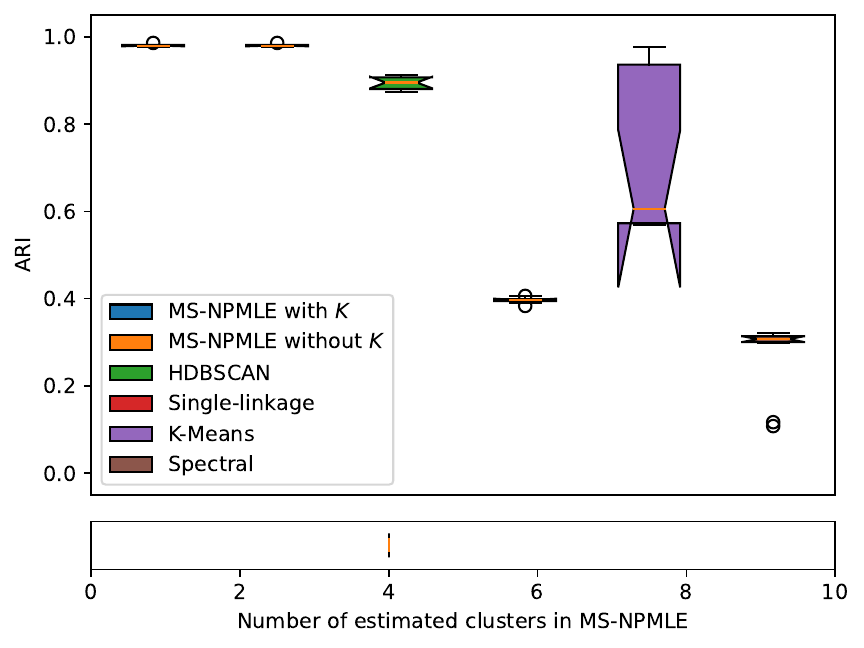}
    \subcaption{Four leaves}
    \endminipage
    \minipage{0.25\textwidth}
    \includegraphics[width=\linewidth]{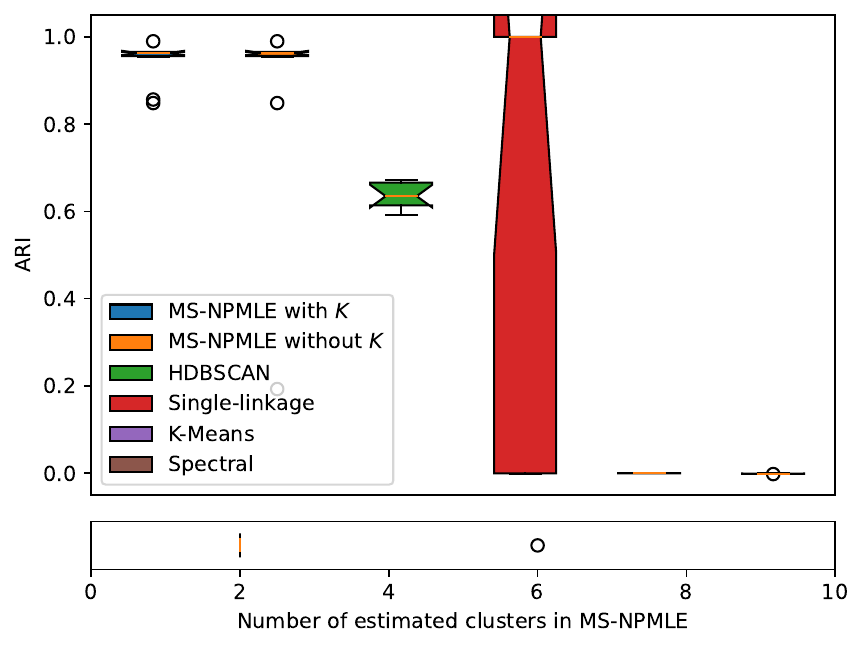}
    \subcaption{Two discs}
    \endminipage
    \minipage{0.25\textwidth}
    \includegraphics[width=\linewidth]{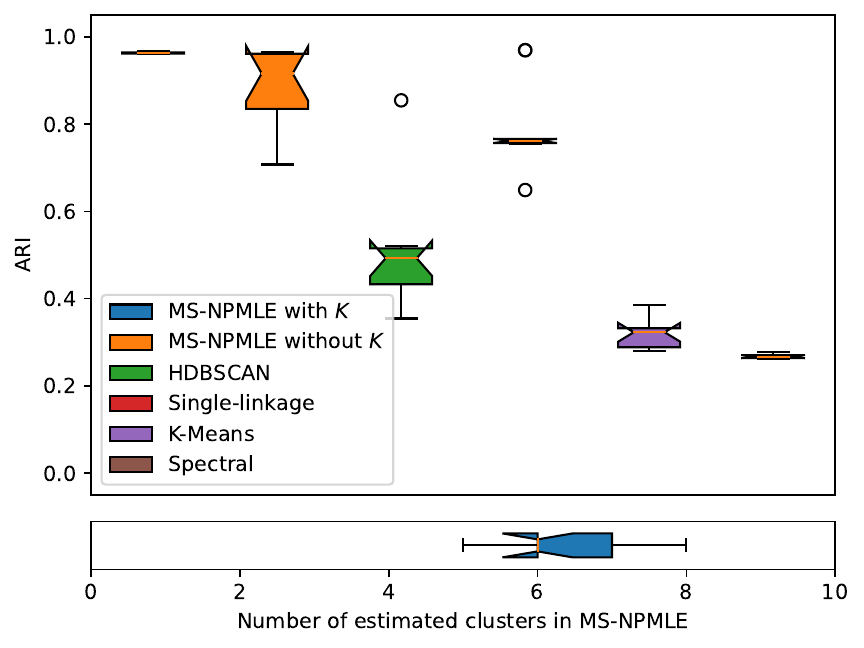}
    \subcaption{Labyrinth}
    \endminipage
    \caption{Comparison of accuracy for the four benchmark examples. The experiments are repeated 10 times. Below each of the ARI comparison is the boxplot for the estimated $K$ obtained by examining the dendrogram of the selected NPMLE.}
    \label{fig:benchmark accuracy}
\end{figure}

\begin{figure}[t]
    \centering
    \minipage{0.333\textwidth}
        \includegraphics[width=\textwidth]{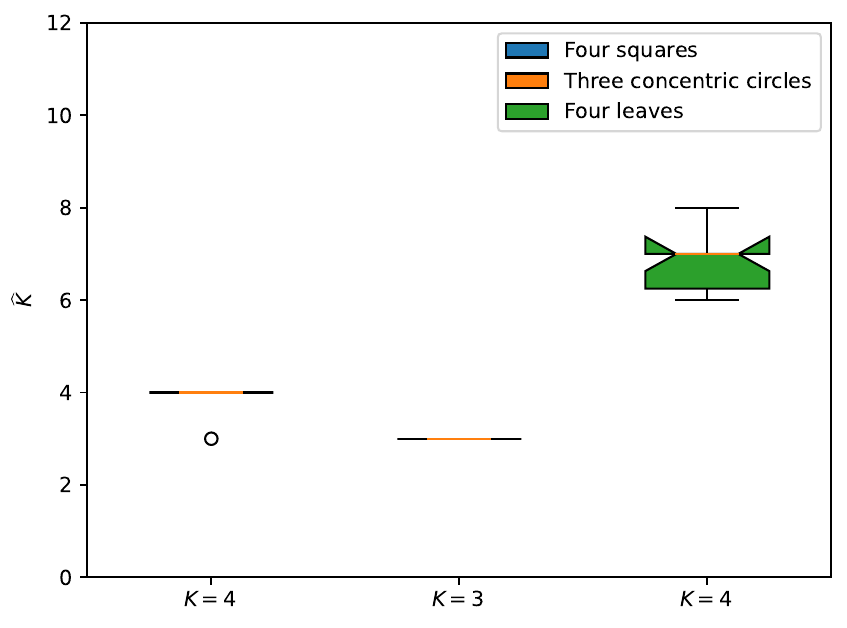}
        \subcaption{}\label{fig:Khat for three examples}
    \endminipage
    \minipage{0.333\textwidth}
        \includegraphics[width=\textwidth]{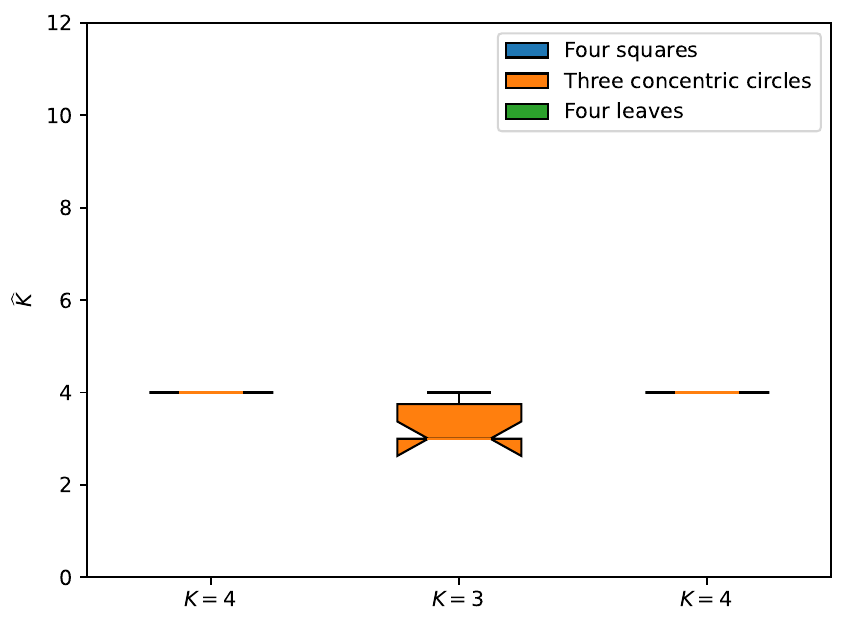}\subcaption{}\label{fig:Khat-DG for three examples}
    \endminipage
\caption{Estimates of $K$ for three examples using (a) $\widehat{K}(2\sigmahat)$ of the selected NPMLE and (b) by inspecting its dendrogram. The x-ticks record the true number of clusters $K$.}
\end{figure}

\begin{remark}[Selection of $K$ revisited]\label{sec:selection of K}
As we have seen from the experiments above, an important problem in practice is the selection of the number of components $K$, for which we have advocated inspecting the dendrogram. 
Another natural estimate for $K$ is the number $\widehat{K}(2\sigmahat)$ (cf. Definition \ref{def:dbscan K}) obtained from the selected NPMLE $\selected$, as it can be interpreted as the discrete analog of number of components within the support atoms of $\selected$.
We remark that this can indeed give the correct number of clusters in certain cases, but is less robust in practice than directly examining the dendrogram. See Figure~\ref{fig:toy dendro} for examples.

To illustrate this point in more detail, we show in Figure \ref{fig:Khat for three examples} the estimates $\widehat{K}(2\sigmahat)$'s for the three concentric circle (Figure \ref{fig:three-circles-sample}), four squares (Figure \ref{fig:four squares}), and the four leaves (Figure \ref{fig:benchmark sample:4leaves}) examples.
We see that for the first two simulated datasets, $\widehat{K}(2\sigmahat)$ correctly estimates the true number of clusters 90\% of the time, whereas for the last benchmark example, $\widehat{K}(2\sigmahat)$ is consistently greater than or equal to 5 and can be as large as 10. 
This is potentially due to finite sample effects where a single connected component can be broken up to multiple ones due to insufficient samples in that cluster. 
On the other hand, the dendrograms of the selected NPMLE can be seen to give more robust estimates of $K$, which we summarize in Figure \ref{fig:Khat-DG for three examples}.
For this reason, we emphasize the practical importance of examining the dendrogram for selecting number of clusters as a model checking and diagnostic step.

\end{remark}

\section{Discussion}
We have proposed a recipe for identifying and estimating the latent structures of a general density whenever such structures are present. 
Our approach is model-free and does not rely on any structural assumptions on the original density.
The central idea lies in extracting the latent structures at different scales, encoded as a collection of probability measures computed via nonparametric maximum likelihood estimation. 
We rigorously characterized the asymptotic limit of such estimators from the perspective of mixture models, and showed that they are strongly consistent. 
By further incorporating the structural information across different scales, we proposed a model selection procedure that returns a most representative model that turns out to be useful for clustering purposes. 
The resulting clustering algorithm is able to decipher a wide range of hidden structures in practice that we have demonstrated on benchmark datasets. 
The investigation brings up interesting theoretical questions regarding the geometry of the NPMLE in high-dimensions that present intriguing directions for future work.

\bibliographystyle{abbrvnat} 
\bibliography{bib}

\clearpage
\appendix

\section{Proof of Proposition \ref{prop:proj unique}} \label{sec:proof of proj uniqueness}

\paragraph{First part:}
\revise{First of all, let's note that each density $p_\sigma \in  \M_\sigma(\Theta)$ has a unique representation as $p_\sigma=\phi_\sigma \ast G_\sigma$ for some $G_\sigma\in \P(\Theta)$ \citep[see e.g.][Theorem 2]{nguyen2013}. Therefore it suffices to show that there exists a unique $p_\sigma$ that solves \eqref{eq:projection}.}

To start with, let's point out that the minimization in \eqref{eq:projection} is equivalent to the following maximization problem:
\begin{align}\label{eq:population mle}
    \underset{p\in \M_\sigma(\Theta)}{\operatorname{max}}\,\, \int_{\R^d} p_0(x) \log p(x) dx =: \underset{p\in \M_\sigma(\Theta)}{\operatorname{max}}\,\, F(p).
\end{align}
Notice that densities in $\M_\sigma(\Theta)$ are uniformly bounded above, so that the maximum \eqref{eq:population mle} exists.
To establish existence of the maximizer, we claim that $\M_\sigma(\Theta)$ is compact when equipped with the $L^1(\R^d)$ metric. 
We need the following result. 

\begin{proposition}\label{prop:compactness of parameter space}
Let $\Theta$ be compact. The space $(\P(\Theta),W_r)$ is compact for $r\in[1,\infty).$
\end{proposition}
\begin{proof}
Since $\Theta$ is compact, it is known that $\P(\Theta)$ is weakly compact. By Corollary 6.13 in \cite{villani2009optimal}, the weak convergence is equivalent to convergence in $W_r$ since the Euclidean distance on $\Theta$ is a bounded metric. Therefore the result follows.  
\end{proof}

To show $(\M_\sigma(\Theta),\|\cdot\|_1)$ is compact, let $p_n=\phi_\sigma\ast G_n\in \M_\sigma(\Theta)$ be a sequence. 
By compactness of $(\P(\Theta),W_1)$ established in Proposition \ref{prop:compactness of parameter space}, the sequence $\{G_n\}_{n=1}^\infty$ converges in $W_1$ along a subsequence to a point $G^*\in \P(\Theta)$. 
By Lemma \ref{lemma:TV < W1} below, since $\|p_n-p^*\|_1=\|\phi_\sigma\ast G_n -\phi_\sigma\ast G^*\|_1 \leq C W_1(G_n,G^*)$, this implies that $\{p_n\}_{n=1}^\infty$ converges in $L^1(\R^d)$ along the same subsequence towards $\phi_\sigma\ast G^*\in \M_\sigma(\Theta)$, establishing compactness. 

Now let $p_n$ be a sequence in $\M_\sigma(\Theta)$ such that $F(p_n)\xrightarrow{n\rightarrow\infty} \operatorname{max}_{p\in\M_\sigma(\Theta)} F(p)$. 
By compactness of $(\M_\sigma(\Theta),\|\cdot\|_1)$, there is a subsequence $p_{n_k}$ such that $p_{n_k}\rightarrow p^*\in \M_\sigma(\Theta)$ in $L^1(\R^d)$ norm, so that along a further subsequence (still denoted as $p_{n_k}$) $p_{n_k}\rightarrow p^*$ pointwise almost everywhere. 
Now consider the sequence of non-positive functions $q_{n_k}=\log p_{n_k}-\log \phi_\sigma(0)$, which converges pointwise almost everywhere to $q^*=\log p^*-\log \phi_\sigma(0)$.
By Fatou's lemma for non-positive functions, we obtain 
\begin{align*}
    F(p^*)-\log \phi_\sigma(0)=\int p_0 q^*& =\int \underset{k\rightarrow\infty}{\operatorname{lim}} p_0 q_{n_k}\\
    &\geq \underset{k\rightarrow\infty}{\operatorname{lim\,sup}} \int p_0 q_{n_k} = \underset{k\rightarrow\infty}{\operatorname{lim\,sup}} \,F(p_{n_k}) - \log \phi_\sigma(0) = \underset{p\in\M_\sigma(\Theta)}{\operatorname{max}} F(p)- \log \phi_\sigma(0). 
\end{align*}
Therefore $p^*\in \M_\sigma(\Theta)$ attains the maximum. 

Finally, uniqueness follows from the strict concavity of $\log(x)$. Indeed, let $p_1$ and $p_2$ be two distinct maximizers of $F$. 
Then by convexity of $\M_\sigma(\Theta)$, $\frac{p_1+p_2}{2}\in \M_\sigma(\Theta)$ and we have
\begin{align*}
	F\Big(\frac{p_1+p_2}{2}\Big) &= \int_{\R^d} p_0 \log \Big(\frac{p_1+p_2}{2}\Big) \\
	&> \int_{\R^d} p_0 \bigg(\frac{\log(p_1)+\log(p_2)}{2}\bigg) = \frac{F(p_1)+F(p_2)}{2} = \underset{p\in \M_\sigma(\Theta)}{\operatorname{max}}\,\, F(p),
\end{align*} 
a contradiction.

\paragraph{Second part:}
Let $\varepsilon>0$ be fixed. 
We shall show that $\widehat{G}_n(\sigma)$ belongs to the $W_p$ ball $B(\projmixing,\varepsilon)$ for all $n$ large, thereby establishing the result. 

Fix $u\in(0,1)$. For any $G\in \P(\Theta)\backslash B(\projmixing,\varepsilon)$, we claim that there exists $\delta=\delta_G$ such that 
\begin{align*}
    \mathbb{E}_{p_0} \log \bigg[1-u+u\frac{(\phi_{\sigma}\ast \projmixing)(X)}{\underset{H\in B(G,\delta)}{\operatorname{sup}}(\phi_{\sigma}\ast H)(X)}\bigg]>0.
\end{align*}
Indeed, we have 
\begin{align*}
    \underset{\delta \rightarrow 0}{\operatorname{lim}}\, \mathbb{E}_{p_0}  \log \bigg[1-u+u\frac{(\phi_{\sigma}\ast \projmixing)(X)}{\underset{H\in B(G,\delta)}{\operatorname{sup}}(\phi_{\sigma}\ast H)(X)}\bigg]
    &\overset{(\text{S1})}{=}
    \underset{\delta \rightarrow 0}{\operatorname{lim\,inf}}\, \mathbb{E}_{p_0}  \log \bigg[1-u+u\frac{(\phi_{\sigma}\ast \projmixing)(X)}{\underset{H\in B(G,\delta)}{\operatorname{sup}}(\phi_{\sigma}\ast H)(X)}\bigg] \\
    &\overset{(\text{S2})}{\geq} \mathbb{E}_{p_0} \underset{\delta \rightarrow 0}{\operatorname{lim\,inf}} \log \bigg[1-u+u\frac{(\phi_{\sigma}\ast \projmixing)(X)}{\underset{H\in B(G,\delta)}{\operatorname{sup}}(\phi_{\sigma}\ast H)(X)}\bigg]\\
    &\overset{(\text{S3})}{=}\mathbb{E}_{p_0}\log \bigg[1-u+u\frac{(\phi_{\sigma}\ast \projmixing)(X)}{(\phi_{\sigma}\ast G)(X)}\bigg]\\
    & \overset{(\text{S4})}{\geq} u\mathbb{E}_{p_0}\log \frac{(\phi_{\sigma}\ast \projmixing)(X)}{(\phi_{\sigma}\ast G)(X)}\overset{(\text{S5})}{>}0.
\end{align*}
Here the step (S1) follows from the monotonicity of the sequence. Step (S2) uses Fatou's lemma, which is indeed applicable here since the integrand is lower bounded by $\log (1-u)>-\infty$. 
Step (S3) follows from the fact that $(\phi_{\sigma}\ast H_m)(x)\xrightarrow{m\rightarrow\infty} (\phi_{\sigma}\ast G)(x)$ for all $x$ if $W_p(H_m,G)\xrightarrow{m\rightarrow\infty} 0$, which is because convergence in $W_p$ metric implies weak convergence and the function $\phi_{\sigma}(x-\theta)$ as a function of $\theta$ is bounded continuous. 
Step (S4) uses the concavity of $\log(x)$. 
The step (S5) follows from the observation that 
\begin{align*}
    \KL{p_0}{ \phi_{\sigma}\ast \projmixing} < \KL{p_0}{\phi_{\sigma}\ast G} \quad \quad \forall G\neq \projmixing,
\end{align*}
which is equivalent to 
\begin{align*}
    \mathbb{E}_{p_0} \log (\phi_{\sigma}\ast \projmixing)(X) > \mathbb{E}_{p_0}\log (\phi_{\sigma}\ast G)(X)\quad \quad \forall G\neq \projmixing.
\end{align*}

Now with the claim, we can form an open cover of the set $\P(\Theta)\backslash B(\projmixing,\varepsilon)$ by using the open balls $B(G,\delta_G)$ with $G$ ranging over $\P(\Theta)\backslash B(\projmixing,\varepsilon)$. Since $\P(\Theta)\backslash B(\projmixing,\varepsilon)$ is closed and hence compact by Proposition \ref{prop:compactness of parameter space}, we obtain a finite open cover $\P(\Theta)\backslash B(\projmixing,\varepsilon)\subset \cup_{j=1}^J B_j$ so that 
\begin{align*}
    \underset{j=1,\ldots,J}{\operatorname{min}}\,\mathbb{E}_{p_0} \log \bigg[1-u+u\frac{(\phi_{\sigma}\ast \projmixing)(X)}{\underset{H\in B_j}{\operatorname{sup}}(\phi_{\sigma}\ast H)(X)}\bigg]>0.
\end{align*}
The by law of large numbers, we have that almost surely 
\begin{align*}
    0<\underset{j=1,\ldots,J}{\operatorname{min}}\,\sum_{i=1}^n \log \bigg[1-u+u\frac{(\phi_{\sigma}\ast \projmixing)(X_i)}{\underset{H\in B_j}{\operatorname{sup}}(\phi_{\sigma}\ast H)(X_i)}\bigg]
    &=\underset{j=1,\ldots,J}{\operatorname{min}}\,\underset{H\in B_j}{\operatorname{inf}}\,\sum_{i=1}^n \log \bigg[1-u+u\frac{(\phi_{\sigma}\ast \projmixing)(X_i)}{(\phi_{\sigma}\ast H)(X_i)}\bigg]\\
    &\leq \underset{H\in\P(\Theta)\backslash B(\projmixing,\varepsilon)}{\operatorname{inf}}\, \sum_{i=1}^n \log \bigg[1-u+u\frac{(\phi_{\sigma}\ast \projmixing)(X_i)}{(\phi_{\sigma}\ast H)(X_i)}\bigg]
\end{align*}
for all large enough $n$'s. 
In other words, for all $H\in \P(\Theta)\backslash B(\projmixing,\varepsilon)$, we have 
\begin{align*}
    \sum_{i=1}^n \log (\phi_{\sigma}\ast H)(X_i)
    &<\sum_{i=1}^n \log \Big[(1-u)(\phi_{\sigma}\ast H)(X_i)+u(\phi_{\sigma}\ast \projmixing)(X_i)\Big]\\
    &=\sum_{i=1}^n \log \Big[\phi_{\sigma}\ast \Big((1-u)H+u\projmixing)\Big)\Big](X_i).
\end{align*}
Since $(1-u)H+u\projmixing\in \P(\Theta)$, the above inequality implies that any $H\in \P(\Theta)\backslash B(\projmixing,\varepsilon)$ cannot attain the maximum likelihood over $\P(\Theta)$. Therefore $\widehat{G}_n(\sigma)$ must be an element of $B(\projmixing,\varepsilon)$.

\section{Proof of Proposition \ref{prop:num of atoms sigma}}\label{sec:proof of no of atoms}

Let's first review a characterization of the number of support atoms for the NPMLE, where we recall that 
\begin{align*}
    \widehat{G}:=\widehat{G}_{n}(\sigma) =\underset{G\in\mathcal{P}(\Theta)}{\operatorname{arg\,max}}\, \sum_{i=1}^n \log (\phi_{\sigma}\ast G)(Y_i),
\end{align*}
where $\{Y_i\}_{i=1}^n$ are i.i.d. samples from the true density $p_0$. It can be shown that \citep[see e.g.][Chapter 5]{lindsay1995mixture} the maximizer $\widehat{G}$ is a discrete measure with at most $n$ atoms and 
\begin{align*}
    \text{supp}(\widehat{G})\subset \{\text{Global maximizers of } D_{\widehat{G}}\},
\end{align*}
where 
\begin{align*}
    D_{\widehat{G}}(\theta) = \frac{1}{n}\sum_{i=1}^n \frac{(\phi_{\sigma}\ast \delta_{\theta})(Y_i)}{(\phi_{\sigma}\ast \widehat{G})(Y_i)}.
\end{align*}
Therefore it suffices to characterize the critical points of $D_{\widehat{G}}$, which reads in our case as 
\begin{align}
    D_{\widehat{G}}(\theta)=\frac{1}{n}\sum_{i=1}^n \frac{\phi_{\sigma}(Y_i-\theta)}{L_i}, \label{eq:the D function}
\end{align}
where $L_i=\phi_{\sigma}\ast \widehat{G}(Y_i)$. 
The following lemma is adapted from Theorem 3 in \cite{polyanskiy2020self} by keeping track of $\sigma$ and then proves Proposition \ref{prop:num of atoms sigma}.

\begin{lemma}
Let $\Ymax=\operatorname{max}_i Y_i$ and $\Ymin=\operatorname{min}_i Y_i$. Define $r=\frac{\Ymax-\Ymin}{2}$. We have 
\begin{align*}
    \text{Number of modes of $D_{\widehat{G}}$} \leq 1.90 + \frac{(\Ymax+10)r}{0.85\sigma^2}.
\end{align*}
\end{lemma}
\begin{proof}
We need to study the zeros of the gradient of \eqref{eq:the D function}, which takes the form 
\begin{align*}
    \frac{1}{n}\sum_{i=1}^n \frac{\phi_{\sigma}(Y_i-\theta)}{L_i}\left(\frac{Y_i-\theta}{\sigma^2}\right)=\frac{1}{n\sigma^2\sqrt{2\pi\sigma^2}}\exp\left(-\frac{\theta^2}{2\sigma^2}\right)\sum_{i=1}^n \left[\frac{1}{ L_i}\exp\left(-\frac{Y_i^2}{2\sigma^2}\right)\right] \exp\left(\frac{Y_i\theta}{\sigma^2}\right)(Y_i-\theta).
\end{align*}
Therefore it suffices to study the zeros of the function
\begin{align*}
    F(\theta)=\sum_{i=1}^n w_i \exp\left(\frac{Y_i\theta}{\sigma^2}\right)(Y_i-\theta),
\end{align*}
where the $w_i=cL_i^{-1}\exp(-\frac{Y_i^2}{2\sigma^2})$,  normalized so that $\sum_{i=1}^n w_i=1$. A first observation is that the zeros of $F$ all lie in the interval $[\Ymin,\Ymax]$ because $F$ is strictly positive (and negative) over $(-\infty, \Ymin)$ (resp. $(\Ymax,\infty)$). Let $\theta_0=\frac{\Ymin+\Ymax}{2}$ and consider
\begin{align*}
    f(z)=F(z+\theta_0) e^{-(z+\theta_0)\frac{\theta_0}{\sigma^2}},\qquad z\in\mathbb{C}. 
\end{align*}
Notice that the real roots of $F$ over $[\Ymin,\Ymax]$ coincide with the real roots of $f$ over $[-r,r]$, where $r=\frac{\Ymax-\Ymin}{2}$. Now we shall apply the following result from \cite{polyanskiy2020self} to bound the number of zeros of $f$ over the disc $\{|z|\leq r\}$.

\begin{lemma}[Lemma 4, \citealp{polyanskiy2020self}] \label{lemma:number of zeros}
Let $f$ be a non-zero holomorphic function on a disc of radius $r_1$. Let $n_f(r):=|\{z\in \mathbb{C}: |z|\leq r,f(z)=0\}|$ and $M_f(r):=\operatorname{sup}_{|z|\leq r} |f(z)|$. For any $r<r_2<r_1$, we have 
\begin{align*}
    n_f(r)\leq \frac{1}{\log \frac{r_1^2+r_2r}{r_1(r_2+r)}} \log \frac{M_f(r_1)}{M_f(r_2)}. 
\end{align*}
\end{lemma}

Let $r_i=r+\delta_i$, $i=1,2$, where $\delta_1>\delta_2>0$ is to be determined. Notice that 
\begin{align*}
    f(z)= \sum_{i=1}^n w_i e^{\frac{(Y_i-\theta_0)(z+\theta_0)}{\sigma^2}}(Y_i-\theta_0-z)=\mathbb{E}e^{\frac{\bar{Y}(z+\theta_0)}{\sigma^2}}(\bar{Y}-z),
\end{align*}
where $\bar{Y}$ is a discrete random variable defined as $\mathbb{P}[\bar{Y}=Y_i-\theta_0]=w_i$ and $\bar{Y}\in[-r,r]$. 
Since $r_2>r$ and $r_2-r=\delta_2$, we have 
\begin{align*}
    M_f(r_2) \geq |f(r_2)| \geq \delta_2\mathbb{E} e^{\frac{\bar{Y}(r_2+\theta_0)}{\sigma^2}}\geq \delta_2 e^{-\frac{r(\Ymax+\delta_2)}{\sigma^2}}.  
\end{align*}
Similarly, we have 
\begin{align*}
    M_f(r_1)=\underset{|z|\leq r_1}{\operatorname{sup}}\, |f(z)| \leq (r+r_1) \mathbb{E}|e^{\frac{\bar{Y}(z+\theta_0)}{\sigma^2}}| \leq (r+r_1)e^{\frac{r(r_1+\theta_0)}{\sigma^2}}= (2r+\delta_1)e^{\frac{r(\Ymax+\delta_1)}{\sigma^2}}. 
\end{align*}
Therefore we have
\begin{align*}
    \log \frac{M_f(r_1)}{M_f(r_2)} \leq \log \left(\frac{2r+\delta_1}{\delta_2}\right) + \frac{r(2\Ymax+\delta_1+\delta_2)}{\sigma^2}.  
\end{align*}
Setting $\delta_1=ar$ and $\delta_2=br$ with $a>b$, we get 
\begin{align*}
    n_f(r) \leq \frac{\log(\frac{2+a}{b})}{\log (\frac{(a+1)^2+b+1}{(a+1)(b+2)})} + \frac{1}{\sigma^2} \frac{r(\Ymax+a+b)}{\log (\frac{(a+1)^2+b+1}{(a+1)(b+2)})}.
\end{align*}
Setting $a=8$, $b=2$, we get 
\begin{align*}
    n_f(r) \leq 1.90 + \frac{(\Ymax+10)r}{0.85\sigma^2}.
\end{align*}  
\end{proof}

The second assertion follows from Theorem 4.1 in \citet{lindsay1983geometry} by noting that if $\sigma>\frac{\Ymax-\Ymin}{2}$, then the mixture quadratic, which takes the form of $M(\theta)=(\theta-\Ymin)(\theta-\Ymax)+\sigma^2$, is strictly positive for all $\theta$.
Hence the NPMLE with be the delta measure at the mean.

\section{Proof of Theorem \ref{thm:npmle component consistency}}\label{sec:proof of main theorem}

Theorem \ref{thm:npmle component consistency} will be proved by combining the following intermediate results presented in Lemma \ref{lemma:thresholding and Ek hat}, Lemma \ref{lemma:partition}, and Lemma \ref{lemma:ub}. 
The first result establishes that the sets $\{\widehat{S}_{\sigma,k}\}_{k=1}^{N_\sigma}$ indeed approximate $\{S_{\sigma,k}\}_{k=1}^{N_\sigma}$ (cf. Definition \ref{def:proj components}) asymptotically. 

\begin{lemma}\label{lemma:thresholding and Ek hat}
Let $\delta_n$ and $t_n$ be two sequences satisfying
\begin{align*}
    \delta_n\rightarrow 0, \quad t_n\rightarrow 0, \quad t_n\geq 2^{-d}\delta_n^{-(d+1)}d^{-1/2}W_1(\estmixing,\projmixing).
\end{align*}
Then for $n$ large enough, the level set $\{x: \widehat{g}_n(x)>t_n\}$ is a union of $N_\sigma$ sets $\widehat{S}_{\sigma,k}$ satisfying
\begin{align}\label{eq:sandwich Ek hat}
    \suppk\backslash O_n \subset \widehat{S}_{\sigma,k} \subset \suppk(2\sqrt{d}\delta_n),\qquad k=1,\ldots,N_\sigma
\end{align}
where $O_n$ is a sequence of sets whose Lebesgue measures converge to zero. 
\end{lemma}
\begin{proof}

Let $S_\sigma=\supp(\projmixing)$ and $S_\sigma(\eta)=\{x:\dist(x,S_\sigma)\leq \eta\}$ be the $\eta$-enlargement of $S_\sigma$. 
First of all, let's point out that $g_n=\projmixing\ast I_{\delta_n}$ is a non-vanishing density over $S_\sigma(\delta_n)$.
Indeed, for any $x\in S_\sigma(\delta_n)$, there exists a set of positive measure $N_x \subset B_{\delta_n}(x) \cap S_\sigma$ so that $\projmixing\geq c$ over $N_x$ for some $c>0$. 
Then it follows that 
\begin{align*}
    g_n(x)=\int I_{\delta_n}(x-\theta)d\projmixing(\theta) \geq c \int_{N_x} I_{\delta_n}(x-\theta) d\theta >0.
\end{align*}

Next, we shall establish an upper bound on $\widehat{g}_n(x)$ when $x$ is away from the support of $\projmixing$. 
Denote $\estmixing=\sum_{\ell=1}^m w_\ell \delta_{\theta_\ell}$ and  $\mathcal{D}_n=\{\ell:\operatorname{dist}(\theta_\ell,\operatorname{supp}(\projmixing))>\sqrt{d}\delta_n\}$.
By the definition of $\widehat{g}_n$, we have
\begin{align*}
    \widehat{g}_n(x)&=\sum_{\ell \in \mathcal{D}_n} w_{\ell} I_{\delta_n}(x-\theta_\ell)+ \sum_{\ell \notin \mathcal{D}_n}  w_{\ell} I_{\delta_n}(x-\theta_\ell) 
    \leq I_{\delta_n}(0)\sum_{\ell \in \mathcal{D}_n} w_{\ell}+\sum_{\ell \notin \mathcal{D}_n}  w_{\ell} I_{\delta_n}(x-\theta_\ell)
\end{align*}
By \cite[][Lemma 5.3]{aragam2023uniform}, the first term can be bounded by $(2\delta_n)^{-d} (\sqrt{d}\delta_n)^{-1} W_1(\estmixing,\projmixing)$.
To bound the second term, notice that for any $x$ such that $\operatorname{dist}(x,\operatorname{supp}(\projmixing))>2\sqrt{d}\delta_n$, we have $|x-\theta_{\ell}|>\sqrt{d}\delta_n$ for any $\ell\notin D_{\sqrt{d}\delta_n}$ and so the second sum in the above equals zero. 
In particular we have shown that 
\begin{align}
    \widehat{g}_n \leq 2^{-d}\delta_n^{-(d+1)} d^{-1/2} W_1(\estmixing,\projmixing) \quad\text{on}\quad \{x:\operatorname{dist}(x,\operatorname{supp}(\projmixing))>2\sqrt{d}\delta_n\}, \label{eq:outlier density level}
\end{align}
Therefore, this suggests a threshold $t_n\geq 2^{-d}\delta_n^{-(d+1)} d^{-1/2} W_1(\estmixing,\projmixing)$ and we can then write 
\begin{equation}\label{eq:decomposition of smoothed Gn}
\begin{aligned}
   \{x: \widehat{g}_n(x)>t_n\}&   =\{x: \widehat{g}_n(x)>t_n\}\cap \{x:\operatorname{dist}(x,\operatorname{supp}(\projmixing))\leq 2\sqrt{d}\delta_n\}\\
  & =\bigcup_{k=1}^{N_\sigma} \{x: \widehat{g}_n(x)>t_n\}\cap \sp{A}{k}(2\sqrt{d}\delta_n)=:\bigcup_{k=1}^{N_\sigma} \sphat{A}{k}.
\end{aligned}
\end{equation}

Next, we shall show that each $\sphat{S}{k}$ is almost a connected set.  
To accomplish this, we need a bound on $\|\widehat{g}_n-g_n\|_1$. 
Let $\Pi$ be any coupling between $\estmixing$ and $\projmixing$
\begin{align*}
    \|\widehat{g}_n-g_n\|_1& = \int_{\R^d}\left|\int_\Theta I_{\delta_n}(x-\theta)d\estmixing(\theta)-\int_\Theta I_{\delta_n}(x-\omega)d\projmixing(\omega)\right|dx\\
    &= \int_{\R^d}\left|\int_\Theta I_{\delta_n}(x-\theta)d\Pi(\theta,\omega)-\int_\Theta I_{\delta_n}(x-\omega)d\Pi(\theta,\omega)\right|dx\\
    &\leq \int_\Theta \int_{\R^d}|I_{\delta_n}(x-\theta)-I_{\delta_n}(x-\omega)|\,dx\, d\Pi(\theta,\omega)\\
    &=\int_\Theta \bigg[\int_{\R^d}\frac{|I_{\delta_n}(x-\theta)-I_{\delta_n}(x-\omega)|}{|\theta-\omega|}\,dx\bigg] |\theta-\omega|\, d\Pi(\theta,\omega)\\
    & \leq\underset{\theta\neq \omega}{\operatorname{sup}} \bigg[\int_{\R^d}\frac{|I_{\delta_n}(x-\theta)-I_{\delta_n}(x-\omega)|}{|\theta-\omega|}dx\bigg]\int_\Theta |\theta-\omega| d\Pi(\theta,\omega)\\
    & \leq \underbrace{\underset{\theta\neq \omega}{\operatorname{sup}} \bigg[\int_{\R^d}\frac{|I_{\delta_n}(x-\theta)-I_{\delta_n}(x-\omega)|}{|\theta-\omega|}dx\bigg]}_{(B)} W_1(\estmixing,\projmixing).
\end{align*}
To bound the term $(B)$, it suffices to upper bound the non-overlapping volumes between the supports of the two box kernels, which is bounded by twice the sum of $d$ hyperrectangles each with side length $(2\delta_n)^{d-1}$ and $|\theta_i-\omega_i|$. Therefore we have  
\begin{align*}
     \int_{\R^d} |I_{\delta_n}(x-\theta)-I_{\delta_n}(x-\omega)|dx \leq 2(2\delta_n)^{-d}\sum_{i=1}^d  (2\delta_n)^{d-1} |\theta_i-\omega_i| \leq \sqrt{d}\delta_n^{-1} |\theta-\omega|, 
\end{align*}
and hence 
\begin{align*}
    \|\widehat{g}_n-g_n\|_1 \leq \sqrt{d}\delta_n^{-1} W_1(\estmixing,\projmixing).
\end{align*}
It follows that the measure of the set $N_{t_n}=\big\{|\widehat{g}_n-g_n|>t_n\big\}$ is bounded by 
\begin{align}\label{eq:measure of Ntn}
    \operatorname{Leb}\big(\big\{|\widehat{g}_n-g_n|>t_n\big\}\big) \leq \frac{\|\widehat{g}_n-g_n\|_1}{t_n} \leq \sqrt{d}t_n^{-1}\delta_n^{-1} W_1(\estmixing,\projmixing).
\end{align}
Now notice that we have 
\begin{align*}
    \{g_n>2t_n\} \cap N_{t_n}^c =\{g_n>2t_n\} \cap \{|\widehat{g}_n-g_n|\leq t_n\} \subset \{\widehat{g}_n>t_n\},
\end{align*}
and hence 
\begin{align*}
     \big(\{g_n>2t_n\} \cap N_{t_n}^c \big)\cap\suppk \subset \{\widehat{g}_n>t_n\}\cap \suppk(2\sqrt{d}\delta_n)=\widehat{S}_{\sigma,k},
\end{align*}
where the last step is the definition of $\sphat{S}{k}$ in \eqref{eq:decomposition of smoothed Gn}.
We shall show that the left hand side is almost an connected set when $n$ is large.  
Indeed, since $g_n$ is non-vanishing over $S_\sigma(\delta_n)$, the set $\{g_n>2t_n\}$ will eventually become $S_\sigma$. 
Similarly, the fact \eqref{eq:measure of Ntn} that $N_{t_n}$ (as a subset of the compact set $\Theta$) has vanishing measure implies that it shrinks towards the empty set.  
Therefore, by intersecting $\suppk$ with $\{g_n>2t_n\}\cap N_{t_n}^c$, we only introduce tiny ``holes'' whose sizes converge to zero. 
In other words, we have shown that 
\begin{align*}
    \suppk \backslash O_n \subset \sphat{S}{k} \subset \suppk(2\sqrt{d}\delta_n),
\end{align*}
where $O_n$ is a set whose Lebesgue measure converges to zero.

\end{proof}

Since the connected components $\{S_{\sigma,k}\}'s$ are closed and disjoint, they are separated by a positive distance: 
\begin{align}\label{eq:separation}
    \underset{j\neq k}{\operatorname{min}}\, \operatorname{dist} (S_{\sigma,j},\suppk)=:\xi > 0,
\end{align}
where we recall the set-wise distance defined in Section \ref{sec:notation}.
Lemma \ref{lemma:thresholding and Ek hat} shows that each $\sphat{S}{k}$ almost recovers the support $\sp{S}{k}$ up to some vanishingly small ``holes'' when $n$ is large. 
Such tiny holes do not ruin the separation structure of the $\suppk$'s as in \eqref{eq:separation}, applying single-linkage clustering to the open sets in $\{\widehat{g}_n>t_n\}$ until $N_\sigma$ clusters remain would return correctly the $\sphat{S}{k}$'s.
The partition $\{E_k\}_{k=1}^{N_\sigma}$ constructed from the $\sphat{S}{k}$'s as in \eqref{eq:def Ek} is shown to satisfy the following property, which is the key of our estimation procedure.

\begin{lemma}\label{lemma:partition}
Let $E_k$ be defined in \eqref{eq:def Ek}. 
We have $\Theta = \bigcup_{k=1}^{N_\sigma} E_k$ and $ \sp{S}{k}(\xi/4)\subset E_k$,
where $\xi$ is as in \eqref{eq:separation} and $\sp{S}{k}(\eta)=\{x:\dist(x,\sp{S}{k})\leq \eta\}$.  
\end{lemma}
\begin{proof}

We need to show that for any $x\in \sp{S}{k}(\xi/4)$, $\dist(x,\sphat{S}{k}) \leq \dist(x,\sphat{S}{j})$ for all $j\neq k$. 
\paragraph{Upper bound on $\dist(x,\sphat{S}{k})$:} Let $z\in \sp{S}{k}$ be a point such that $d(x,z)\leq \xi/4$.  
Then 
\begin{align*}
    \dist(x,\sphat{S}{k}) \leq \frac{\xi}{4} + \dist(z,\sphat{S}{k}).
\end{align*}
Since $z\in \sp{S}{k}$, we have by \eqref{eq:sandwich Ek hat} that $\dist(z,\sphat{S}{k})<\frac{\xi}{4}$ when $n$ is large.
Therefore we have $\dist(x,\sphat{S}{k})<\frac{\xi}{2}$. 

\paragraph{Lower bound on $\dist(x,\sphat{S}{j})$ for $j\neq k$: } 
Let $z_k\in \sp{S}{k}$ be a point such that $d(x,z_k)\leq \frac{\xi}{4}$. 
Let $y\in \sphat{S}{j}$ and $z_j\in \sp{S}{j}$ be such that $d(y,z_j)\leq 2\sqrt{d}\delta_n$. 
Then by \eqref{eq:separation}
\begin{align*}
    \xi<\dist(\sp{S}{k},\sp{S}{j})\leq d(z_k,z_j) \leq d(z_k,x)+d(x,y)+d(y,z_j) \leq \frac{\xi}{4} + d(x,y) + 2\sqrt{d}\delta_n,
\end{align*}
so that 
\begin{align}\label{eq:Ek lb}
    \dist(x,\sphat{S}{j}) \geq d(x,y) \geq \frac{3\xi}{4}- 2\sqrt{d}\delta_n>\frac{\xi}{2}
\end{align}
when $n$ is large. 
Combining this with the upper bound on $\dist(x,\sphat{S}{k})$ establishes the lemma. 

\end{proof}

Intuitively, Lemma \ref{lemma:partition} states that the set $E_k$ contains  an $\xi/4$ enlargement of the support of precisely one $G_{\sigma,k}$. 
Similarly as in \eqref{eq:decomposition of G sigma}, we then have the following decomposition for $\widehat{G}_{n,\sigma}$ based on the $E_k$'s:
\begin{align}\label{eq:npmle components}
    \estmixing = 
    \sum_{k=1}^{K} \underbrace{\estmixing(E_k)}_{:=\widehat{\lambda}_{n,\sigma,k}} \underbrace{\estmixing(\cdot\,|E_k)}_{:=\widehat{G}_{n,\sigma,k}} =:\sum_{k=1}^{K} \widehat{\lambda}_{n,\sigma,k} \widehat{G}_{n,\sigma,k}. 
\end{align}

Now we are ready to finish the proof of Theorem \ref{thm:npmle component consistency} by showing the following result and combining it with Proposition \ref{prop:proj unique}.

\begin{lemma}\label{lemma:ub}
We have 
\begin{align*}
    \underset{k}{\operatorname{max}}\, \Big[|\widehat{\lambda}_{n,\sigma,k}-\lambda_{\sigma,k}| \vee \|\widehat{f}_{n,\sigma,k}-f_{\sigma,k}\|_1\Big] \leq C_{\sigma} \big(-\log W_1(\estmixing,\projmixing) \big)^{-1/2},
\end{align*}
where $C_\sigma$ is a constant depending on $\sigma$. 
\end{lemma}
\noindent

\begin{proof}

To simplify the notation, we will suppress the dependence of $\widehat{G}_{n,\sigma}, \widehat{G}_{n,\sigma,k},\widehat{f}_{n,\sigma,k}, \widehat{\lambda}_{n,\sigma,k}$ on $n$ and denote them as $\widehat{G},\widehat{G}_k, \widehat{f}_k, \widehat{\lambda}_k$. 
First we notice that to bound $|\widehat{\lambda}_k-\lambda_k|$ and $\|\widehat{f}_{k}-f_k\|_1$, it suffices to bound 
$\|\widehat{\lambda}_{k}\widehat{f}_{k}-\lambda_k f_k\|_1$. 
Indeed we have 
\begin{align*}
    |\widehat{\lambda}_{k}-\lambda_k|=\left|\int \widehat{\lambda}_{k}\widehat{f}_{k}-\lambda_k f_k\right|\leq \|\widehat{\lambda}_{k}\widehat{f}_{k}-\lambda_k f_k\|_1
\end{align*}
and 
\begin{align*}
    \|\widehat{\lambda}_{k}\widehat{f}_{k}-\lambda_k f_k\|_1 &\geq \|\lambda_k \widehat{f}_{k}-\lambda_k f_k\|_1-\|\widehat{\lambda}_{k} \widehat{f}_{k}-\lambda_k \widehat{f}_{k}\|_1 \\
    &= \lambda_k\|\widehat{f}_{k}-f_k\|_1-|\widehat{\lambda}_{k}-\lambda_k|
\end{align*}
so that 
\begin{align}
    |\widehat{\lambda}_{k}-\lambda_k|+\lambda_k\|\widehat{f}_{k}-f_k\|_1 \leq 3\|\widehat{\lambda}_{k}\widehat{f}_{k}-\lambda_k  f_k\|_1.  \label{eq:gtl reduction 2}
\end{align}
Now we shall focus on $\|\widehat{\lambda}_{k}\widehat{f}_{k}-\lambda_k f_k\|_1$, which we decompose as three terms by introducing a mollifier
\begin{align*}
    \|\widehat{\lambda}_{k}\widehat{f}_{k}-\lambda_k f_k\|_1 &
    =\left\|\int_{E_k} \phi_{\sigma}(x-\theta) d\widehat{G}(\theta)-\int_{E_k} \phi_{\sigma}(x-\theta)d\widetilde{G}(\theta)\right\|_1\leq J_1+J_2+J_3,
\end{align*}
where
\begin{align*}
    J_1&:=\left\|\int_{E_k} \phi_{\sigma}(x-\theta) d\widehat{G}(\theta)-\int_{E_k} \phi_{\sigma}(x-\theta) d(\widehat{G}\ast H_{\delta})(\theta)\right\|_1 \\
    J_2&:= \left\|\int_{E_k} \phi_{\sigma}(x-\theta) d(\widehat{G}\ast H_{\delta})(\theta)-\int_{E_k} \phi_{\sigma}(x-\theta) d(\widetilde{G}\ast H_{\delta} )(\theta)\right\|_1 \\
    J_3&:= \left\|\int_{E_k} \phi_{\sigma}(x-\theta) d(\widetilde{G}\ast H_{\delta})(\theta)-\int_{E_k} \phi_{\sigma}(x-\theta) d\widetilde{G}(\theta)\right\|_1.
\end{align*}
Here $H$ is a symmetric density function with bounded first moment whose Fourier transform is supported in $[-1,1]^d$ and $H_{\delta}=\delta^{-d}H(\delta^{-1}\cdot)$.

\vspace{0.2cm}
\noindent{\textbf{Bound for $J_1$}}: Recall $\widehat{G}=\sum_{k=1}^K \widehat{\lambda}_k \widehat{G}_k$ with $\operatorname{supp}(\widehat{G}_k)\subset E_k$. Then $\widehat{G}\ast H_{\delta}=\sum_{k=1}^{K}\widehat{\lambda}_{k} \widehat{G}_{k}\ast H_{\delta}$ and
\begin{align*}
    \int_{E_k}\phi_{\sigma}(x-\theta) d\widehat{G}(\theta)&=\widehat{\lambda}_k \int_{E_k} \phi_{\sigma}(x-\theta)d\widehat{G}_k(\theta)=\widehat{\lambda}_k \int_{\Theta} \phi_{\sigma}(x-\theta)d\widehat{G}_k(\theta)\\
    \int_{E_k}\phi_{\sigma}(x-\theta) d(\widehat{G}\ast H_{\delta})(\theta)&=\widehat{\lambda}_k\int_{E_k}\phi_{\sigma}(x-\theta)d(\widehat{G}_k\ast H_{\delta})(\theta) 
    +\sum_{j\neq k}\widehat{\lambda}_j\int_{E_k} \phi_{\sigma}(x-\theta)d(\widehat{G}_j\ast H_{\delta})(\theta).
\end{align*}
We have 
\begin{align*}
    J_1&\leq \left\|\widehat{\lambda}_k \int_{\Theta} \phi_{\sigma}(x-\theta)d\widehat{G}_k(\theta)-\widehat{\lambda}_k\int_{\Theta}\phi_{\sigma}(x-\theta)d(\widehat{G}_k\ast H_{\delta})(\theta)\right\|_1 \\
    &\quad + \left\|\widehat{\lambda}_k\int_{\Theta}\phi_{\sigma}(x-\theta)d(\widehat{G}_k\ast H_{\delta})(\theta)- \widehat{\lambda}_k\int_{E_k}\phi_{\sigma}(x-\theta)d(\widehat{G}_k\ast H_{\delta})(\theta)\right\|_1 \\
    &\quad + \left\|\sum_{j\neq k}\widehat{\lambda}_j\int_{E_k} \phi_{\sigma}(x-\theta)d(\widehat{G}_j\ast H_{\delta})(\theta)\right\|_1=:e_1+e_2+e_3.
\end{align*}
By \cite[][Lemma~1]{nguyen2013}, we have 
\begin{align}
    e_1\leq W_1(\widehat{G}_k,\widehat{G}_k\ast H_{\delta})\leq C\delta, \label{eq:e1 G hat}
\end{align}
where the last step can be proved as in \cite[][Theorem~2]{nguyen2013}: letting $\theta\sim \widehat{G}_k$ and $\varepsilon \sim H_{\delta}$ gives $W_1(\widehat{G}_k,\widehat{G}_k\ast H_{\delta})\leq\mathbb{E}|\theta-(\theta+\varepsilon)|\leq C\delta$. 
For $e_2$ we have 
\begin{align*}
    e_2 \leq \widehat{\lambda}_k \left\|\int_{\Theta\backslash E_k} \phi_{\sigma}(x-\theta)d(\widehat{G}_k\ast H_{\delta})(\theta)\right\|_1=\widehat{\lambda}_k(\widehat{G}_k\ast H_{\delta})(\Theta\backslash E_k),
\end{align*}
where
\begin{align*}
    (\widehat{G}_k\ast H_{\delta})(\Theta\backslash E_k)&=\int_{\Theta\backslash E_k} \int_{E_k} H_{\delta}(\theta-z)d\widehat{G}_k(z) d\theta \\
    &=\int_{E_k} \int_{\Theta\backslash E_k} H_{\delta}(\theta-z)d\theta d\widehat{G}_k(z) \\
    &=\int_{S_k(\xi/2)} \int_{\Theta\backslash E_k} H_{\delta}(\theta-z)d\theta d\widehat{G}_k(z)+ \int_{E_k\backslash S_k(\xi/2)} \int_{\Theta\backslash E_k} H_{\delta}(\theta-z)d\theta d\widehat{G}_k(z)\\
    &=:i_1+i_2.
\end{align*}
Recall that $E_k\supset S_k(\xi/4)$ so we have $\operatorname{dist}(\Theta\backslash E_k,S_k(\xi/8))\geq \frac{\xi}{8}$. 
Then for $z\in S_k(\xi/8)$
\begin{align*}
    \int_{\Theta\backslash E_k}H_{\delta}(\theta-z)d\theta \leq \int_{|x|>\xi/8} H_{\delta}(x)dx=\int_{|x|>\xi/8\delta}H(x)dx\leq \frac{8\delta}{\xi} \int_{|x|>\xi/8\delta}|x|H(x)dx\leq \frac{8C\delta}{\xi}
\end{align*}
and hence $i_1\leq 8C\xi^{-1}\delta$.
For $i_2$, we have 
\begin{align*}
    i_2 \leq \widehat{G}_k(E_k \backslash S_k(\xi/8))=\frac{\widehat{G}(E_k\backslash S_k(\xi/8))}{\widehat{\lambda}_k}.
\end{align*}
Therefore 
\begin{align}
    e_2 \leq  \widehat{\lambda}_k(\widehat{G}_k\ast H_{\delta})(\Theta\backslash E_k)\leq 8C\widehat{\lambda}_k \xi^{-1}\delta +\widehat{G}(E_k\backslash S_k(\xi/8)), \label{eq:e2 G hat}
\end{align}
and further that 
\begin{align}
    e_3\leq \sum_{j\neq k}\widehat{\lambda}_j (G_j\ast H_{\delta})(E_k)\leq \sum_{j\neq k}\widehat{\lambda}_j (G_j\ast H_{\delta})(\Theta\backslash E_j) \leq \sum_{j\neq k}8C\widehat{\lambda}_j\xi^{-1}\delta+\widehat{G}(E_j\backslash S_j(\xi/8)). \label{eq:e3 G hat}
\end{align}
Combining \eqref{eq:e1 G hat}, \eqref{eq:e2 G hat} and \eqref{eq:e3 G hat} we get 
\begin{align}
    J_1\leq C\xi^{-1}\delta+ \widehat{G}(A_{\xi/8})\leq C_{\xi}\left[\delta+W_1(\widehat{G},G)\right], \label{eq:I1 G hat}
\end{align}
where $A_{\eta}=\{x: \operatorname{dist}(x,\operatorname{supp}(G))>\eta\}$.

\vspace{0.2cm}
\noindent{\textbf{Bound for $J_3$}}: The term $J_3$ can be bounded similarly as 
\begin{align}
    J_3 \leq C_{\xi}\delta. \label{eq:I3 G hat}
\end{align}
Note that since the support of $G$ is $\bigcup_{k=1}^K S_k$, the corresponding error term $G(A_{\delta})$ is zero.

\vspace{0.2cm}
\noindent{\textbf{Bound for $J_2$}}: Let $M=\diam(\Theta)$. 
\begin{align*}
    J_2&\leq \int_{\R^d} \int_{E_k} \phi_{\sigma}(x-\theta)d|\widehat{G}\ast H_{\delta}-G\ast H_{\delta}|(\theta) dx\\
    &\leq\int_{\R^d} \int_{\R^d} \phi_{\sigma}(x-\theta)d|\widehat{G}\ast H_{\delta}-G\ast H_{\delta}|(\theta) dx\\
    &=\int_{\R^d}|\widehat{G}\ast H_{\delta}(\theta)-G\ast H_{\delta}(\theta)|d\theta\\
    &=\int_{|\theta|\leq M+\xi} |\widehat{G}\ast H_{\delta}(\theta)-G\ast H_{\delta}(\theta)|d\theta + \int_{|\theta|>M+\xi}|\widehat{G}\ast H_{\delta}(\theta)-G\ast H_{\delta}(\theta)|d\theta.
\end{align*}
The second term can be bounded by noticing that 
\begin{align*}
    \int_{|\theta|>M+\xi} G\ast H_{\delta}(\theta)d\theta
    =\int_{|z|\leq M} \int_{|\theta|>M+\xi} H_{\delta}(\theta-z)d\theta dG(z)
    \leq \int_{|x|>\xi/\delta}H(x)dx \leq C\xi^{-1}\delta 
\end{align*}
and similarly for $\int_{|\theta|>M+\xi} \widehat{G}\ast H_{\delta}(\theta)d\theta$. The first term can be bounded using Cauchy-Schwarz by 
\begin{align*}
    \sqrt{\int_{|\theta|\leq M+\xi}1d\theta \int_{|\theta|\leq M+\xi}|\widehat{G}\ast H_{\delta}(\theta)-G\ast H_{\delta}(\theta)|^2d\theta}\leq \sqrt{C_d(M+\xi)^d}\|\widehat{G}\ast H_{\delta}-G\ast H_{\delta}\|_2 
\end{align*}
Letting 
$h_{\delta}=\mathcal{F}^{-1}(\mathcal{F}H_{\delta}/\mathcal{F}\phi_{\sigma})$ (since $\mathcal{F}H_{\delta}$ is continuous and compactly supported, and $\mathcal{F}\phi_{\sigma}$ is never zero, $\mathcal{F}H_{\delta}/\mathcal{F}\phi_{\sigma}\in L^1$ and $h_{\delta}$ is well-defined), we have $H_{\delta}=\phi_{\sigma}\ast h_{\delta}$ and then
\begin{align*}
    \widehat{G}\ast H_{\delta}&=(\widehat{G}\ast \phi_{\sigma})\ast h_{\delta}=:\widehat{Q} \ast h_{\delta}\\
    G\ast H_{\delta}&=(G\ast \phi_{\sigma})\ast h_{\delta}=:Q \ast h_{\delta}.
\end{align*}
Thus by Young's inequality we have
\begin{align*}
    \|\widehat{G}\ast H_{\delta}-G\ast H_{\delta}\|_2=\|\widehat{Q}\ast h_{\delta}-Q\ast h_{\delta}\|_2\leq \|\widehat{Q}-Q\|_1 \|h_{\delta}\|_2
\end{align*}
and by Plancherel's identity
\begin{align*}
    \|h_{\delta}\|^2_2=\left\|\frac{\mathcal{F}H_{\delta}}{\mathcal{F}\phi_{\sigma}}\right\|_2^2\leq C \int_{|w|<1/\delta} \exp\left(\sigma^2|w|^2\right)dw\leq C\exp\left(d\sigma^2\delta^{-2}\right),
\end{align*}
where we have used the fact that $\mathcal{F}H$ is bounded supported on $[-1,1]^d$, which implies that $\mathcal{F}H_{\delta}$ is supported on $[-1/\delta,1/\delta]^d$. 
Therefore we have 
\begin{align}
    J_2 \leq C \left[\delta+\|\widehat{Q}-Q\|_1\exp\left(2^{-1}d\sigma^2\delta^{-2}\right) \right]. \label{eq:I2 G hat}
\end{align}
and furthermore by combining \eqref{eq:I1 G hat}, \eqref{eq:I3 G hat}, \eqref{eq:I2 G hat} we have 
\begin{align*}
    \|\widehat{\lambda}_k\widehat{f}_k-\lambda_k f_k\|_1 \leq C\left[\delta+\|\widehat{Q}-Q\|_1\exp\left(2^{-1}d\sigma^2\delta^{-2}\right) +W_1(\widehat{G},G)\right]. 
\end{align*}
Setting $d\sigma^2\delta^{-2}=-\log \|\widehat{Q}-Q\|_1$, we get 
\begin{equation*}%
\begin{aligned}
    \|\widehat{\lambda}_k\widehat{f}_k-\lambda_kf_k\|_1& \leq C_{\sigma}\left[ \big(-\log \|\widehat{Q}-Q\|_1 \big)^{-1/2}+\|\widehat{Q}-Q\|_1^{1/2}+W_1(\widehat{G},G)\right]. \\
    & \leq C_{\sigma}\left[\big(-\log \|\widehat{Q}-Q\|_1 \big)^{-1/2}+W_1(\widehat{G},G)^{1/2}\right]\\
    & \leq C_{\sigma} \big(-\log W_1(\widehat{G},G) \big)^{-1/2},
\end{aligned}
\end{equation*}
where we have used Lemma \ref{lemma:TV < W1} below to bound $\|\widehat{Q}-Q\|_1$ and the fact that $\sqrt{x}\leq (-\log x)^{-1/2}$. This finishes the proof with \eqref{eq:gtl reduction 2}. 

\end{proof}

\begin{lemma}\label{lemma:TV < W1}
    Let $G,H\in \mathcal{P}(\Theta)$. Then 
    \begin{align*}
    \|\phi_{\sigma}\ast G- \phi_{\sigma}\ast H\|_1 \leq \frac{\Phi_d}{\sigma} W_1(G,H),\qquad \Phi_d=\frac{1}{(2\pi)^{d/2}}\int_{\R^d}|z|\exp\left(-\frac{|z|^2}{2}\right) dz. 
\end{align*}
\end{lemma}
\begin{proof}
    Let $\Pi$ be a coupling between $G$ and $H$. Then 
    \begin{align*}
        (\phi_{\sigma}\ast G)(x)-(\phi_{\sigma}\ast H)(x) & = \int_{\Theta\times \Theta} [\phi_{\sigma}(x-\theta) -\phi_{\sigma}(x-\omega)] d\Pi(\theta,\omega),
    \end{align*}
with 
\begin{align*}
\phi_{\sigma}(x-\theta) -\phi_{\sigma}(x-\omega) =\int_0^1 \nabla \phi_{\sigma}(x-\omega+t(\omega-\theta))^T(\theta-\omega)dt.
\end{align*}
Then we have 
\begin{align*}
    \|\phi_{\sigma}\ast G- \phi_{\sigma}\ast H\|_1 &\leq \int_{\Theta\times\Theta}\int_{\R^d}|\phi_{\sigma}(x-\theta) -\phi_{\sigma}(x-\omega)|dxd\Pi(\theta,\omega)\\
    & \leq \int_{\Theta\times\Theta} \int_0^1 \int_{\R^d} |\nabla \phi_{\sigma}(x-\omega+t(\omega-\theta))|dxdt|\theta-\omega|d\Pi(\theta,\omega)\\
    & \leq \|\nabla \phi_{\sigma}\|_1 \int_{\Theta\times\Theta} |\theta-\omega|d\Pi(\theta,\omega)\\
    & =\frac{\Phi_d}{\sigma}\int_{\Theta\times\Theta} |\theta-\omega|d\Pi(\theta,\omega), \qquad \qquad \Phi_d=\frac{1}{(2\pi)^{d/2}}\int_{\R^d}|z|\exp\left(-\frac{|z|^2}{2}\right) dz
\end{align*}
where taking the infimum over all couplings gives the desired result. 
\end{proof}

\end{document}